\documentclass[sigconf, nonacm]{acmart}
\usepackage{amsmath,amsfonts,dsfont,multirow,multicol,epsfig,url,array,makecell,balance,color,epstopdf}
\usepackage{algorithmic}
\usepackage[ruled, vlined, linesnumbered]{algorithm2e}
\usepackage{hhline}
\usepackage{array}
\usepackage{enumerate}
\usepackage{enumitem}
\usepackage{subfigure}
\usepackage{booktabs}
\usepackage{xcolor,colortbl}
\usepackage{bm}
\usepackage{lipsum}
\usepackage{tabularx}
\usepackage{diagbox}
\usepackage{caption}
\usepackage{threeparttable}
\usepackage{etoolbox}
\SetCommentSty{mycommfont}
\newcommand{\tabincell}[2]{}
\newcolumntype{P}[1]{>{\centering\arraybackslash}p{#1}}

\setlist[itemize]{leftmargin=*}

\newtheorem{definition}{Definition}
\newtheorem{lemma}{Lemma}
\newtheorem{theorem}{Theorem}
\newtheorem{fact}{Fact}

\def\header{\vspace{0.8mm} \noindent}

\def\tblcapup{\vspace{0mm}}

\newcommand{\pushright}[1]{\ifmeasuring@#1\else\omit\hfill$\displaystyle#1$\fi\ignorespaces}
\newcommand{\pushleft}[1]{\ifmeasuring@#1\else\omit$\displaystyle#1$\hfill\fi\ignorespaces}

\def\dmax{d_{\max}}
\def\davg{d}

\def\e{\varepsilon}
\def\rela{c}
\def\pf{p_f}

\def\epi{\bm{\hat{\pi}}}
\def\bpi{\bm{\bar{\pi}}}
\def\vpi{\bm{\pi}}

\def\r{\bm{r}}
\def\er{\bm{r}}
\def\P{\mathbf{P}}

\def\A{\mathbf{A}}
\def\D{\mathbf{D}}

\def\incre{X}
\def\I{\mathcal{I}}
\def\th{\theta}

\def\E{\mathrm{E}}
\def\Var{\mathrm{Var}}

\def\rev{\color{black}}

\def\setpush{{\em SetPush}\xspace}

\def\sublinear{{SubgraphPush}\xspace}

\newcommand\vldbdoi{XX.XX/XXX.XX}
\newcommand\vldbpages{XXX-XXX}
\newcommand\vldbvolume{14}
\newcommand\vldbissue{11}
\newcommand\vldbyear{2023}
\newcommand\vldbauthors{\authors}
\newcommand\vldbtitle{\shorttitle} 
\newcommand\vldbavailabilityurl{https://github.com/wanghzccls/SetPush-code}
\newcommand\vldbpagestyle{plain} 

\begin{document}
\title{Estimating Single-Node PageRank in $\tilde{O}\left(\min\{d_t, \sqrt{m}\}\right)$ Time}
\subtitle{[Technical Report]}

\author{Hanzhi Wang}
\affiliation{%
  \institution{Renmin University of China}
  \city{Beijing}
  \country{China}
}
\email{hanzhi_wang@ruc.edu.cn}

\author{Zhewei Wei}
\thanks{Zhewei Wei is the corresponding author. The work was partially done at Gaoling School of Artificial Intelligence, Peng Cheng Laboratory, Beijing Key Laboratory of Big Data Management and Analysis Methods and MOE Key Lab of Data Engineering and Knowledge Engineering.} 
\affiliation{%
  \institution{Renmin University of China}
  \city{Beijing}
  \country{China}
}
\email{zhewei@ruc.edu.cn}

\begin{abstract}
{\rev 
PageRank is a famous measure of graph centrality that has numerous applications in practice. The problem of computing a single node's PageRank has been the subject of extensive research over a decade. However, existing methods still incur large time complexities despite years of efforts. Even on undirected graphs where several valuable properties held by PageRank scores, the problem of locally approximating the PageRank score of a target node remains a challenging task. Two commonly adopted techniques, Monte-Carlo based random walks and backward push, both cost $O(n)$ time in the worst-case scenario, which hinders existing methods from achieving a sublinear time complexity like $O(\sqrt{m})$ on an undirected graph with $n$ nodes and $m$ edges. 




In this paper, we focus on the problem of single-node PageRank computation on undirected graphs. We propose a novel algorithm, \setpush, for estimating single-node PageRank specifically on undirected graphs. With non-trival analysis, we prove that our \setpush achieves the $\tilde{O}\left(\min\left\{d_t, \sqrt{m}\right\}\right)$ time complexity for estimating the target node $t$'s PageRank with constant relative error and constant failure probability on undirected graphs. 
We conduct comprehensive experiments to demonstrate the effectiveness of \setpush. 
}
 	
\end{abstract}


\maketitle

\pagestyle{\vldbpagestyle}
\begingroup\small\noindent\raggedright\textbf{PVLDB Reference Format:}\\
\vldbauthors. \vldbtitle. PVLDB, \vldbvolume(\vldbissue): \vldbpages, \vldbyear.\\
\href{https://doi.org/\vldbdoi}{doi:\vldbdoi}
\endgroup
\begingroup
\renewcommand\thefootnote{}\footnote{\noindent
This work is licensed under the Creative Commons BY-NC-ND 4.0 International License. Visit \url{https://creativecommons.org/licenses/by-nc-nd/4.0/} to view a copy of this license. For any use beyond those covered by this license, obtain permission by emailing \href{mailto:info@vldb.org}{info@vldb.org}. Copyright is held by the owner/author(s). Publication rights licensed to the VLDB Endowment. \\
\raggedright Proceedings of the VLDB Endowment, Vol. \vldbvolume, No. \vldbissue\ %
ISSN 2150-8097. \\
\href{https://doi.org/\vldbdoi}{doi:\vldbdoi} \\
}\addtocounter{footnote}{-1}\endgroup

\ifdefempty{\vldbavailabilityurl}{}{
\vspace{.3cm}
\begingroup\small\noindent\raggedright\textbf{PVLDB Artifact Availability:}\\
The source code, data, and/or other artifacts have been made available at \url{\vldbavailabilityurl}.
\endgroup
}

\section{Introduction} \label{sec:intro}
PageRank is first proposed by Google~\cite{page1999pagerank} to rank the importance of web pages in the search engine. It is formulated based on two intuitive arguments: (i) highly linked pages are more important than the pages with fewer links; (ii) the page that linked by an important page is also important. 
{\rev If we convert the web structure to a graph, the PageRank scores of all pages in the web correspond to the probability distribution of simulating random walks on the graph. 
Specifically, consider a graph $G=(V,E)$ with $|V|=n$ nodes and $|E|=m$ edges. 
We select a node $s$ from the graph's vertex set $V$ uniformly at random, and simulate an $\alpha$-random walk from node $s$. The PageRank score of node $t\in V$ is equal to the probability that an $\alpha$-random walk simulated from node $s$ terminates at node $t$. Here we call $s$ the source node. 
$\alpha$-random walk refers to the random walk process that at each step (e.g., at node $u$), the walk either terminates at $u$ with probability $\alpha$, or moves to a randomly selected neighbor of $u$ with probability $1-\alpha$. We call $\alpha$ {\em the teleport probability} or {\em the damping factor}, which is a constant satisfying $\alpha \in (0,1)$. 
}


Over the last decade, PageRank has emerged as one of the most well-adopted graph centrality measure
~\cite{gleich2015beyondtheweb}. 
The applications of PageRank has been far beyond its origin in web search, covering a wide range of research domains, such as social networks, recommender systems, databases, as well as biology, chemistry, neuroscience and etc. 
For example, in social networks, PageRank serves as a classic role in evaluating the centrality of individuals. Kwak et al.~\cite{kwak2010twitter} use PageRank to characterize the properties of Twitter. In recommender systems, the PageRank scores of items are adopted to find potential predictions~\cite{boldi2008query}. Moreover, for the problem of database queries, the PageRank score indicates a query direction to the frequently retrieved results, and thus accelerates the query efficiency~\cite{balmin2004objectrank}. Additionally, PageRank are adopted to study molecules in chemistry~\cite{mooney2012molecularnetworks}, gene in biology~\cite{morrison2005generank} and brain regions in neuroscience~\cite{zuo2012neuro}. More applications of PageRank can be found in the comprehensive survey summarized by Gleich~\cite{gleich2015beyondtheweb}. 

{\rev 
At the same time, a plethora of variants stem from PageRank, including Personalized PageRank~\cite{page1999pagerank}, heat kernel PageRank~\cite{chung2007heatkernelPageRank}, reverse PageRank~\cite{bar2008reversePageRank}, weighted PageRank~\cite{xing2004weightedPageRank} and so on. 
}
For example, Personalized PageRank, one of the most famous variant of PageRank, has been an essential node proximity metric adopted in various web search and representation tasks~\cite{gupta2013wtf, klicpera2019APPNP, Bojchevski2020PPRGo}. Recall that PageRank serves as a global centrality measure in a graph. In comparison, the Personalized PageRank value of a node indicates a localized score, reflecting the relative importance of the node with respect to a given source node. Likewise, the heat kernel PageRank has a successful history in the local clustering scenario. A series of algorithms~\cite{chung2007heatkernelPageRank, yang2019TEA, kloster2014heat} leverage the scores of Heat Kernel PageRank to identify a well-connected cluster around the given seed node. These variants and their wide-spread applications also demonstrate the prominence of PageRank in graph analysis and mining tasks. 

{\rev 
Given the huge success achieved by PageRank, the problem of computing PageRank scores has been the subject of extensive research for more than a decade~\cite{bressan2018sublinear, bar2008reversePageRank, fogaras2005MC, andersen2007contribution, lofgren2013personalized, lofgren2016BiPPR, lofgren2014FastPPR}. One particular interest is the problem of {\em single-node PageRank computation}, which aims to compute a single node's PageRank on large-scale graphs. Such problem is an important primitive in graph analysis and learning tasks of both practical and theoretical interest. 

From the theoretical aspect, 
the query time complexity of single-node PageRank has a close connection to various graph analysis problems. For example, as we shall show in Section~\ref{sec:pre}, node $t$'s PageRank is equal to the average over all nodes $u$'s $\vpi_u(t)$, where $\vpi_u(t)$ denotes the Personalized PageRank (PPR) score of node $t$ with respect to node $u$. We call such problem single-target PPR queries, in which we aim to estimate  $\vpi_u(t)$ of every node $u\in V$. 
The theoretical insight for single-node PageRank computation can therefore be used for single-target PPR queries by definition. 
Moreover, 
Bressan et al.~\cite{bressan2018sublinear} propose a novel method called {\em \sublinear} for single-node PageRank computation, and adapt the \sublinear method to computing single-node Heat Kernel PageRank (HKPR) by leveraging the analogue between PageRank and HKPR. 

On the other hand, in many practical cases, all we need is an approximation of a few nodes' PageRank scores. 
For example, in the application scenario of web search, the changes in the importance of a few popular websites (e.g., the top-10 most popular websites ranked last year) is of particular interest. Since websites' global importance can be reflected from their PageRank scores, the PageRank scores of the ten websites are therefore frequently requested. Note that it would be prohibitively slow to score all nodes in the graph every time, especially on large-scale graphs with millions or even billions of nodes and edges. Therefore, an ideal solution is a {\em local} algorithm, which is able to efficiently return the target node's approximation scores by only exploring a small fraction of graph edges around the target node. However, as pointed out by~\cite{bressan2018sublinear}, most of existing approaches require an $\Omega(n)$ time complexity for the single-node PageRank computation. 
Designing an efficient local algorithm with $o(n)$ query time complexity remains a challenge. 







}


{\rev 
\header{\bf Single-Node PageRank Computation on Undirected Graphs. } Existing methods for single-node PageRank computation mainly focus on directed graphs, which, however, incur large query time complexity despite decades of efforts due to the hardness. In this paper, we settle for a slightly less ambitious target to efficiently estimate single-node PageRank {\em on undirected graphs}. Note that the problem of single-node PageRank computation on undirected graphs is still of great importance from both practical and theoretical aspects. Specific reasons are illustrated in the following. 
\begin{itemize}
\vspace{-1mm}
\item From the theoretical aspect, a number of existing algorithms do not offer any worst-case guarantee on directed graphs without considering a uniform random choice of the target node. For these methods, meaningful complexity bounds  can only be derived on undirected graphs when we consider an arbitrary target node (e.g., the LocalPush~\cite{lofgren2013personalized}, FastPPR~\cite{lofgren2014FastPPR}, and BiPPR~\cite{lofgren2016BiPPR} methods as listed in Table~\ref{tbl:comparison}). 
On the other hand, there are several crucial properties of the PageRank scores that are only held on undirected graphs. This motivates us 
to study the problem of single-node PageRank computation specifically on undirected graphs for achieving better complexity results by utilizing these crucial properties delicately.

\item Second, 
from the practical aspect, many downstream graph mining and learning tasks are only defined on undirected graphs. 
For example, in the scenario of local clustering, the celebrated local clustering method~\cite{FOCS06_FS} employs (Personalized) PageRank vector to identify local clusters, while the well-adopted {\em conductance} metric to measure the quality of identified clusters is defined on undirected graphs. Therefore, in local clustering, all we need is the PageRank scores on undirected graphs. 
Additionally, Graph Neural Networks (GNNs) have drawn increasing attention in recent years. A plethora of GNN models leverage PageRank computation to propagate node features~\cite{klicpera2019APPNP, Bojchevski2020PPRGo,Chen2020GCNII}. Since the graph Laplacian matrix for feature propagation is only applicable to undirected graphs, this line of research invokes PageRank computation algorithms only on undirected graphs. 
\end{itemize}
}

\begin{figure}[t]
\centering
\begin{tabular}{c}
\hspace{-1mm}\includegraphics[width=88mm,trim=0 0 10 0,clip]{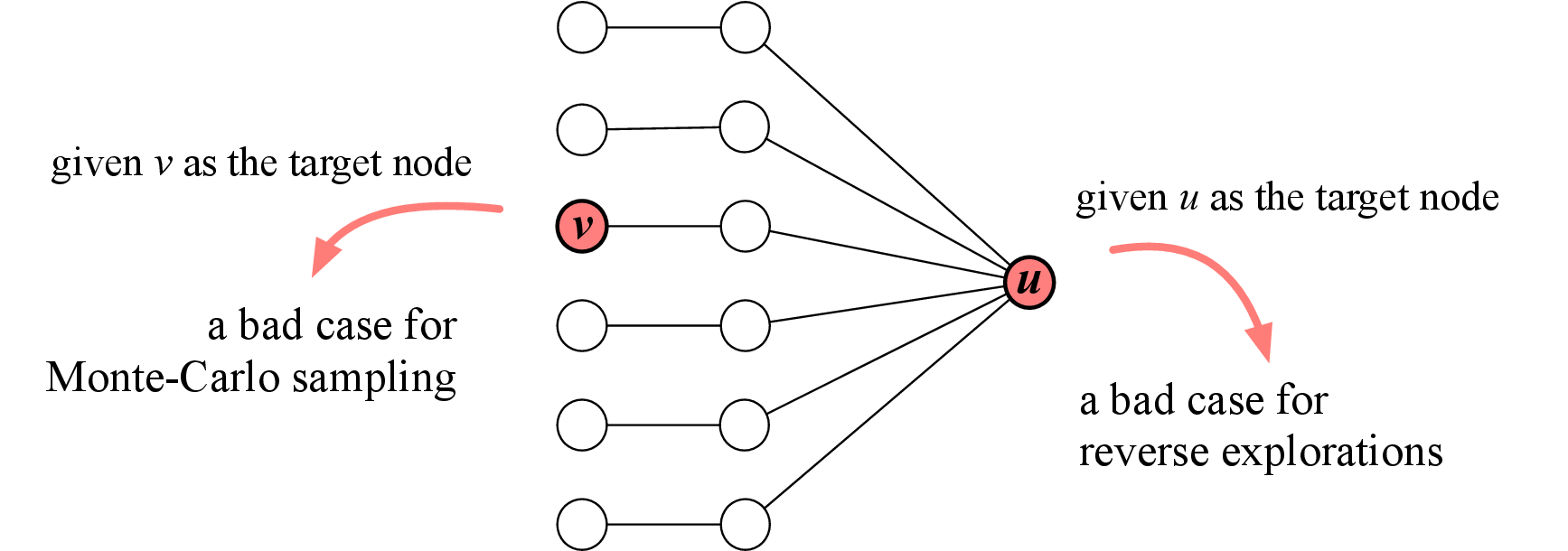}
\end{tabular}
\vspace{-4mm}
\caption{A special case. }
\label{fig:high_level}
\vspace{-5mm}
\end{figure}

\begin{table*} [t]
\centering
\renewcommand{\arraystretch}{1.4}
\tblcapup
\caption{\rev Comparison of algorithms for solving the problem of single-node PageRank computation on undirected graphs under constant relative error and failure probability. $d$ and $d_{\max}$ denotes the average and maximum degree of graph $G$, respectively. The complexity results marked by $\star$ are only applicable on undirected graphs. 
}\label{tbl:comparison}
\vspace{-2mm}
\begin{tabular} 
{|P{3.5cm}|c@{\hspace{+1mm}}|@{\hspace{+1mm}}c@{\hspace{+0.5mm}}|P{6.3cm}|} \hline
{\bf Query Time Complexity} & \multirow{2}*{\bf Baseline Methods} &  {\bf Query Time Complexities} & {\bf {Improvement of \setpush over Baselines}} \\ 
{\bf of Our \setpush} & & {\bf of Baseline Methods} & (the larger, the better) \\ \hline
\multirow{7}*{$\tilde{O}\left(\min\left\{d_t, \sqrt{m}\right\}\right) \star$} & The Power Method~\cite{page1999pagerank} & $\tilde{O}(m)$ & $\max \left\{ m / d_t, \sqrt{m} \right\}$  \\  \cline{2-4}
& Monte-Carlo~\cite{fogaras2005MC} & $\tilde{O}(n)$ & $\max \left\{n/d_t,  \sqrt{n/\davg} \right\}$ \\ \cline{2-4}
& LocalPush~\cite{lofgren2013personalized} & $\tilde{O}\left(\min\left\{n\cdot d_t, m\right\}\right) \star$ & $\sqrt{m}$ \\ \cline{2-4}
&  RBS~\cite{wang2020RBS} & $\tilde{O}\left(n\right)$ & $\max \left\{n / d_t, \sqrt{n/\davg} \right\}$ \\ \cline{2-4}
& FastPPR~\cite{lofgren2014FastPPR} &  $\tilde{O}\left(\sqrt{n\cdot d_t}\right) \star$ & $\max \left\{\sqrt{n / d_t}, \sqrt{d_t / d}\right\}$ \\ \cline{2-4}
& BiPPR~\cite{lofgren2016BiPPR, lofgren2015BiPPRundirected} & $\tilde{O}\left(\sqrt{n\cdot d_t}\right) \star$ & $\max \left\{\sqrt{n / d_t}, \sqrt{d_t / d}\right\}$ \\ \cline{2-4}
& \sublinear~\cite{bressan2018sublinear} & {\rev $\tilde{O}\left(\min\left\{\frac{m^{2/3}\cdot \dmax^{1/3}}{\davg^{2/3}}, \frac{m^{4/5}}{\davg^{3/5}} \right\}\right)$} 
& $\max \left\{\min \left\{\frac{n^{\frac{2}{3}}\cdot d_{\max}^{1/3}}{d_t}, \frac{n^{\frac{4}{5}} \cdot d^{\frac{1}{5}}}{d_t}\right\}, \min\left\{\frac{n^{\frac{1}{6}}\cdot d_{\max}^{1/3}}{d^{\frac{1}{2}}}, \frac{n^{\frac{3}{10}}}{d^{\frac{3}{10}}}\right\}\right\}$\\ \hline
\end{tabular}
\vspace{-3mm}
\end{table*}

\vspace{-0.5mm}
{\rev
\header{\bf Limitations of Existing Methods on Undirected Graphs. } 
Below we briefly illustrate the limitations of existing methods for the single-node PageRank computation on undirected graphs. A simplified problem formulation is given as follows. A formal definition can be found in Section~\ref{sec:pre}.  
Specifically, the inputs to the single-node PageRank problem are an undirected graph $G=(V,E)$ and a target node $t \in V$. The goal is to estimate the target node $t$'s PageRank $\vpi(t)$ within a constant relative error. 
We also allow a constant failure probability for scalability. 
For the single-node PageRank computation problem, existing methods can be broadly classified into three categories: 
\begin{itemize}
    \item \underline{The Monte-Carlo method}~\cite{fogaras2005MC, fogaras2003start, lofgren2016BiPPR} estimate $\vpi(t)$ by repeatedly simulating $\alpha$-random walks in the graph. However, according to the Pigeonhole principle, the lower bound of the required number of random walks is $\Omega\left(1/n\right)$. Thus, in the worse-case scenario where $\vpi(t)=O\left(1/n\right)$, the Monte-Carlo method requires at least $O(n)$ computational time for estimating a single node's PageRank. In Figure~\ref{fig:high_level}, we provide a toy example to illustrate the hard instance by regarding node $v$ as the given target node which satisfies $\vpi(v)=\Theta\left(1/n\right)$. 

    \item \underline{The reverse exploration method} attempts to derive an estimate of $\vpi(t)$ by reversely exploring the graph from the target node $t$ to its ancestors. A primitive operation commonly adopted in these methods is {\em backward push}, which deterministically pushes the probability mass initially at the target node $t$ reversely to its ancestors step by step. 
    Unfortunately, in each backward push operation (e.g., at node $u$), we at least require $O(d_u)$ time to reversely push the probability mass currently at $u$ to every neighbor of $u$, where $d_u$ denotes the degree of node $u$. Thus, 
    in the worst case where $d_u=\Omega(n)$, we cost $O(n)$ time only after one step of backward push. 
    Figure~\ref{fig:high_level} provides a toy example for this bad case where $d_u=\Theta(n)$. 
    

    \item \underline{The hybrid method} combines the Monte-Carlo method and the reverse exploration method together. However, a simple combination cannot resolve the limitations of the Monte-Carlo and reverse exploration methods as mentioned above. 
    In fact, despite years of efforts, the problem of computing single-node PageRank on undirected graphs has not been well solved. 

\end{itemize}

}

\vspace{-1mm}
\subsection{Our Contributions}\label{subsec:contribution}
{\rev 
In this paper, we consider the problem of single-node PageRank computation {\em on undirected graphs}. We propose a novel algorithm called \setpush, which achieves 
the $\tilde{O}\left(\min \left\{d_t, \sqrt{m}\right\}\right)$ query time complexity for the single-node PageRank computation under constant relative error and failure probability. Here $m$ denotes the number of edges in the graph, $d_t$ denotes the degree of the given target node $t$. Additionally, $\tilde{O}$ is a variant of the Big-Oh notation that ignores poly-logarithmic factors~\cite{bressan2018sublinear, wang2020RBS, teng2016scalable}. 
Detailed contributions achieved by this paper are summarized as below. 
\begin{itemize}
\vspace{-1mm}
\item \header{\bf Theoretical Improvements. } 
We theoretically demonstrate the superiority of our \setpush over existing methods on undirected graphs. 
Specifically, 
in the last column of Table~\ref{tbl:comparison}, we present the theoretical improvements of our \setpush over existing methods. In particular, the value of ``Improvement" equals the query time complexity of a baseline method over that of our \setpush. Thus, the value of ``Improvement" is the larger, the better. It's worth mentioning that the complexity results of FastPPR, BiPPR, LocalPush and our \setpush given in Table~\ref{tbl:comparison} are only applicable to undirected graphs, while the other complexities hold both on directed and undirected graphs. 
We observe that the expected time complexity of our \setpush is no worse than that of each baseline method listed in Table~\ref{tbl:comparison}. 
Actually, except on a compete graph where the average node degree $d=n$, the time complexity of our \setpush is asymptotically better than that of every method listed in Table~\ref{tbl:comparison}. 
\item \header{\bf A Novel Push Operation. } 
The core of our \setpush is a novel push operation, 
which simultaneously mixes the deterministic backward push and randomized Monte-Carlo sampling in an atomic step. Benefit from this push operation, we cost $\Theta(d_u)$ time only at the node $u$ with small $d_u$, and randomly sample a fraction of $u$'s neighbors to push probability mass if $d_u$ is large. As a result, we successfully remove the $O(d_u)$ term introduced by the vanilla push operation at node $u$, and achieve a superior time complexity over the baseline method. 


\item \header{\bf Algorithm Development on Undirected Graphs. } Our \setpush algorithm is designed specifically on undirected graphs. We show that by making full use of the theoretical properties held by PageRank values on undirected graphs, we can achieve a better time complexity for single-node PageRank computation compared to existing methods on undirected graphs. 
\end{itemize}




}




\vspace{-2mm}
\section{Preliminaries} \label{sec:pre}
This section introduces several basic concepts that are frequently adopted in the single-node PageRank computation. 
Table~\ref{tbl:def-notation} shows the notations that are frequently used in this paper. 

\begin{table} [t]
\centering
\renewcommand{\arraystretch}{1.3}
\begin{small}
\tblcapup
\caption{\rev Table of notations.}\label{tbl:def-notation}
\vspace{-4mm}
\resizebox{1\linewidth}{!}{%
\begin{tabular} 
{|c@{\hspace{1mm}}|P{2.25in}|} \hline
{\bf Notation} &  {\bf Description}  \\ \hline
$G=(V,E)$ & undirected graph with vertex set $V$ and edge set $E$ \\ \hline
$n, m$ & the numbers of nodes and edges in $G$ \\ \hline
$N(u)$ & the adjacency list of node $u$\\ \hline
$\A$ & the adjacency matrix of $G$\\ \hline
$d_u$ & the degree of node $u$ \\ \hline
$\davg$ & the average node degree of the graph\\ \hline
$\dmax$ & the maximum node degree of the graph\\ \hline
$\D$ & the diagonal degree matrix that $D_{uu}=d_u$\\ \hline
$\P=\A\D^{-1}$ & the transitional probability matrix\\ \hline
$\alpha$	& the teleport probability that an $\alpha$-discounted random walk terminates at each step \\ \hline
$\vpi(t),\epi(t)$ & the true and estimated PageRank of node $t$. \\ \hline
$\vpi_t,\epi_t$ & the true and estimated Personalized PageRank vectors with regard to node $t$. \\ \hline
$c$ & constant relative error\\ \hline
$\tilde{O}$ & the Big-Oh natation ignoring the log factors \\ \hline
\end{tabular}
}
\vspace{-4mm}
\end{small}
\end{table}

\subsection{PageRank}
Given an undirected and unweighted graph $G=(V,E)$ with $n$ nodes and $m$ edges, the PageRank vector $\vpi$ is an $n$-dimensional vector, which can be mathematically formulated as: 
\begin{align}~\label{eqn:def_pagerank}
\vpi=(1-\alpha)\A\D^{-1}\cdot \vpi +\frac{\alpha}{n}\cdot \bm{1}. 
\end{align}
Here $\A$ denotes the adjacency matrix of the graph, $\D$ is the diagonal degree matrix that $\D_{uu}=d_u$, $\bm{1}\in \mathbb{R}^n$ denotes an all-one vector, and $\alpha$ is a {\em constant} damping factor, which is strictly less than $1$ (i.e., $\alpha \in (0,1)$). For each node $t\in V$, we use $\vpi(t)$ to denote the PageRank value of node $t$. 
According to the definition formula given in Equation~\eqref{eqn:def_pagerank}, the PageRank value of node $t$ satisfies the following recurrence relation: 
\begin{align}\label{eqn:ite_pagerank}
\vpi(t)=(1-\alpha)\sum_{u\in N(t)}\frac{\vpi(u)}{d_u}+\frac{\alpha}{n}, 
\end{align}
where $u$ is one of the neighbor of node $t$, and $d_u$ denotes the degree of node $u$. In particular, Equation~\eqref{eqn:ite_pagerank} also indicates a lower bound of any node's PageRank that $\vpi(t)\ge \frac{\alpha}{n}$ for each $t\in V$. 

\header{\bf $\alpha$-random walk. } 
By the definition formula of PageRank vector $\vpi$ given in Equation~\eqref{eqn:def_pagerank}, 
we can further derive: 
\begin{align}\label{eqn:power_series}
\vpi=\left(\mathbf{I}-(1-\alpha)\A \D^{-1}\right)^{-1}\cdot \left(\frac{\alpha}{n}\cdot \bm{1}\right). 
\end{align} 
As pointed out in~\cite{lofgren2015PHDthesis}, Equation~\eqref{eqn:power_series} can be solved using a power series expansion~\cite{avrachenkov2007monte}: 
\begin{align}\label{eqn:ite_power_method}
\vpi=\sum_{i=0}^{\infty} \alpha(1-\alpha)^i\cdot (\A\D^{-1})^{i}\cdot \frac{1}{n}\cdot \bm{1}, 
\end{align}
where $\vpi$ corresponds to a random walk probability distribution. Specifically, 
a random walk on the graph is a sequence of nodes $W=\{w_0, w_1, w_2, \ldots\}$ that the $i$-th step (i.e., the node $w_i$) in the walk is selected uniformly at random from the neighbor of node $w_{i-1}$. The PageRank value of node $t$ equals to the probability that a so called {\em $\alpha$-random walk} (or $\alpha$-discounted random walks in some literature)~\cite{wang2017fora, wang2020RBS} simulated from a uniformly selected source node $s$ terminates at node $t$. Note that in each step (e.g., currently at node $u$), an $\alpha$-random walk: 
\begin{itemize}
    \item with probability $(1-\alpha)$, select a neighbor $v$ uniformly at random from the adjacency list $N(u)$ of node $u$, and moves from $u$ to $v$; 
    \item with probability $\alpha$, terminates at the current node $u$. 
\end{itemize}
Therefore, the length $L$ of an $\alpha$-random walk is a geometrical random number following the geometric distribution $L\sim G(\alpha)$. The expectation of $L$ is therefore a constant that $\E\left[L\right]=\frac{1}{\alpha}$. 

\header{\bf Problem Definition. } In this paper, we concern the problem of single-node PageRank computation. Specifically, given a target node $t$, a relative error parameter $\rela$, and a failure probability parameter $p_f$, we aim to derive a $(c,p_f)$ approximation of $\vpi(t)$, which is formally defined as follows. 
\begin{definition} [$(\rela,\pf)$-Approximation of Single-Node PageRank]\label{def:problem}
Given a target node $t$ in the graph $G=(V,E)$, 
$\epi(t)$ is an $(\rela, \pf)$-approximation of the single-node PageRank $\vpi(t)$ if 
\vspace{-1mm}
\begin{equation}\nonumber 
\left|\epi(t) - \vpi(t) \right| \le \rela \cdot \vpi(t)
\end{equation}
holds with probability at least $1-\pf$. 
\end{definition}

Note that in a line of research~\cite{bressan2018sublinear, lofgren2014FastPPR, wang2020RBS}, $\rela$ is set as a {\em constant} and thus is omitted in the Big-Oh notation. In this paper, we assume $\rela$ is a constant following this convention. Additionally, we assume $\pf$ is also a {\em constant} without loss of generality. 
It's worth mentioning that a constant failure probability $\pf$ can be easily reduced to arbitrarily small with only adding a log factor to the running time by utilizing the Median-of-Mean trick~\cite{charikar2002mediantrick}.

\subsection{Personalized PageRank}
Apart from PageRank, the seminal paper~\cite{page1999pagerank} also propose a variant of PageRank, called Personalized PageRank (PPR), to evaluate the {\em personalized} centrality of graph vertices with respect to a given source node. The definition formula of PPR is analogous to that of PageRank except for the initial distribution: 
\begin{align}\label{eqn:def_ppr}
\vpi_s=(1-\alpha)\A \D^{-1}\cdot \vpi_s +\alpha \bm{e}_s. 
\end{align}
Specifically, $\vpi_s \in \mathcal{R}^{n}$ is called the single-source PPR vector, where $\vpi_s(t)$ denotes the PPR value of node $t$ with respect to node $s$. $\bm{e}_s$ is an one-hot vector that $\bm{e}_s(s)=1$ and $\bm{e}_s(u)=0$ if $u\neq s$. Analogously, by applying the power series expansion~\cite{avrachenkov2007monte}, we can derive: 
\begin{align}\label{eqn:ite_power_method_ppr}
\vpi_s=\sum_{\ell=0}^{\infty}\alpha (1-\alpha)^\ell \left(\A \D^{-1}\right)^{\ell}\cdot \bm{e}_s. 
\end{align}
Equation~\eqref{eqn:ite_power_method_ppr} provides a probabilistic interpretation on the PPR score. Specifically, the PPR value $\vpi_s(u)$ corresponds to the probability that an $\alpha$-random walk generated from node $s$ terminates at node $u$. Additionally, by comparing Equation~\eqref{eqn:ite_power_method_ppr} with Equation ~\eqref{eqn:ite_power_method}, we note that the PageRank score $\vpi(t)$ is actually an average over all $\vpi_u(t)$ for $\forall u\in V$: 
\begin{align}\label{eqn:PageRank_PPR}
\vpi(t)=\frac{1}{n}\cdot\sum_{s\in V}\vpi_s(t). 
\end{align}
In particular, on undirected graphs, PPR vectors exhibit an underlying {\em reversibility property} that for any node-pair $(u,v)\in V^2$~\cite{lofgren2015BiPPRundirected}:   
\begin{align}\label{eqn:birectional_ppr}
\vpi_u(v)\cdot d_u=\vpi_v(u) \cdot d_v. 
\end{align}

\header{\bf $\boldsymbol{\ell}$-hop PPR. } Given a source node $s$, a target node $t$ and an integer $\ell\ge 0$, the $\ell$-hop PPR $\vpi_s^{(\ell)}(t)$ corresponds to the probability that an $\alpha$-random walk generated from node $s$ terminates at node $t$ exactly in its $\ell$-th step. The $\ell$-hop PPR vector $\vpi_s^{(\ell)}$ is defined as below. 
\begin{align}\label{eqn:def_lhopppr}
\vpi^{(\ell)}_s=\alpha (1-\alpha)^\ell \cdot \left(\A \D^{-1}\right)^{\ell} \bm{e}_s. 
\end{align}
By Equation~\eqref{eqn:def_lhopppr} and Equation~\eqref{eqn:ite_power_method}, we can thus derive $\vpi_s=\sum_{\ell=0}^{\infty} \vpi_s^{(\ell)}$. 
Moreover, the $\ell$-hop PPR value $\vpi^{(\ell)}_s(u)$ admits the following recursive equation that for each node $v\in V$ and each integer $\ell\ge 1$:  
\begin{align}\label{eqn:PPR_recur}
\vpi_t^{(\ell+1)}(v)=\sum_{u\in N(v)}\frac{(1-\alpha)}{d_u}\cdot \vpi_t^{(\ell)}(u). 
\end{align}
Moreover, the $\ell$-hop PPR vector also exhibits the reversibility property on undirected graphs. More specifically, for every two nodes $u,v$ in an undirected $G$ and every $\ell\in \{0,1, \ldots \}$, we have: 
\begin{align}\label{eqn:undirectedPPR}
\vspace{-2mm}
\vpi^{(\ell)}_s(t)\cdot d_s=\vpi^{(\ell)}_t(s)\cdot d_t. 
\end{align}




\section{Analysis of Existing Methods} \label{sec:related}

\begin{figure*}[t]
\centering
\begin{tabular}{c}
\hspace{+1mm}
\includegraphics[height=48mm,trim=0 0 30 0,clip]{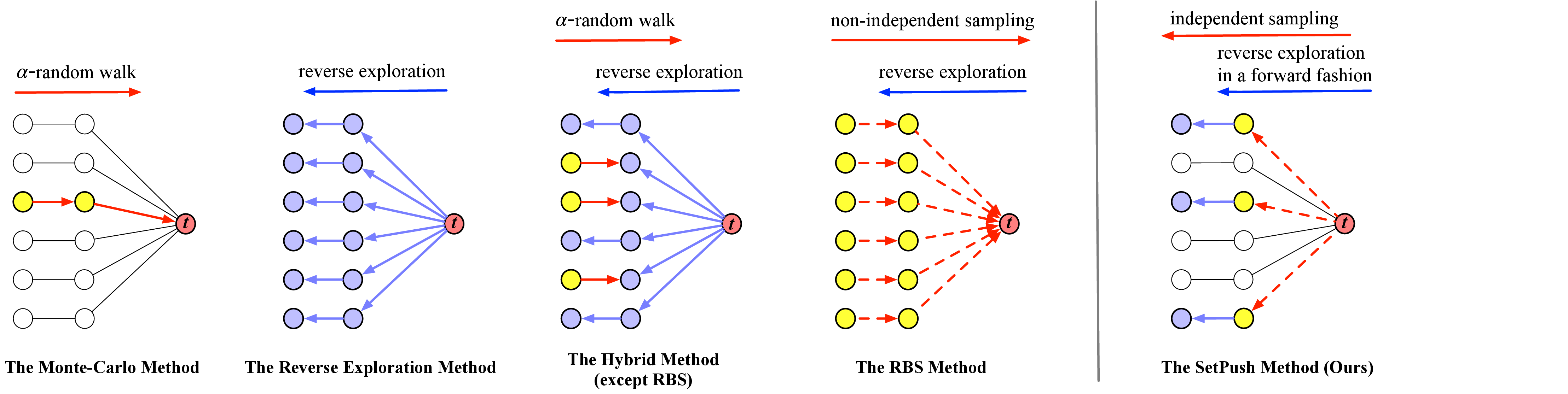}
\end{tabular}
\vspace{-4mm}
\caption{Comparison of existing methods. }
\label{fig:related}
\vspace{-4mm}
\end{figure*}

In this section, we present a brief review on existing approaches for single-node PageRank computation. Specifically, we classify existing methods into four categories: the power method~\cite{page1999pagerank}, the Monte-Carlo method~\cite{fogaras2005MC}, the reverse exploration method~\cite{andersen2007contribution, lofgren2013personalized} and the hybrid method~\cite{lofgren2014FastPPR, lofgren2016BiPPR, bressan2018sublinear, wang2020RBS}. 
Figure~\ref{fig:related} provides a sketch to illustrate the differences among these methods.

\subsection{The Power Method} 
The power method~\cite{page1999pagerank} is an iterative method for computing PageRank values of all nodes in the graph. It defines an $n$-dimensional vector $\epi$ as an approximation of the PageRank vector $\vpi$, where $\epi(t)$ is an estimate of node $t$'s PageRank $\vpi(t)$. The power method initially sets $\epi$ as $\frac{1}{n}\cdot \bm{1}$, and iteratively updates $\epi$ according to the definition formula given in Equation~\eqref{eqn:def_pagerank} 
until $\epi$ merely converges. As demonstrated in~\cite{haveliwala2003convergence_rate}, the convergence rate of the power method is given by $(1-\alpha)$. For the typical setting that $\alpha=0.2$, the convergence rate of PageRank becomes $0.8$, which turns out to be every fast even on large-scale graphs. 

However, a major drawback of the power method is that the power method involves a multiplication between the transition matrix $\P=\A \D^{-1}$ and the PageRank vector $\vpi$ in each iteration. Note that $\P$ is an $n\times n$ matrix with $m$ nonzero entries and $\vpi$ is an $n$-dimensional vector. Thus, the power method requires at least $O(m)$ time in each iteration, which is time-costly especially for single-node PageRank queries on large-scale graphs. 

\subsection{The Monte-Carlo Method} 
Recall that the PageRank score of node $t$ equals the probability that an $\alpha$-random walk simulated from a uniformly selected source node terminates at node $t$. Thus, the Monte-Carlo method~\cite{fogaras2005MC} generates $n_r$ $\alpha$-random walks in the graph, where the source node of each walk is independently selected from $V$ uniformly at random. Then the Monte-Carlo method computes $\frac{1}{n_r}\cdot \sum_{w=1}^{n_r} \I^{(w)}(t)$ as an estimate of $\vpi(t)$,  
where $\I^{(w)}(t)$ is an indicator variable that $\I^{w}(t)=1$ if the $w$-th random walk terminates at node $t$. 
{\rev 
By the Chernoff bound, the number of $\alpha$-random walks that is required to derive a $(c,p_f)$-approximation of $\vpi(t)$ can be bounded as $n_r=\tilde{O}\left(\frac{1}{\e^2}\right)$. Recall that the expected length $L$ of an $\alpha$-random walk is $\E[L]=\frac{1}{\alpha}$, which is a constant. Consequently, the expected time cost of the Monte-Carlo method for achieving the $(\rela, \pf)$-approximation of single-node PageRank is bounded by $O(n_r)=\tilde{O}\left(n\right)$.
}

\subsection{The Reverse Exploration Method}\label{subsec:reverse_exploration}
Another line of research~\cite{andersen2007contribution, lofgren2013personalized,jeh2003scaling} computes single-node PageRank via reverse explorations. Specifically, given a target node $t$, this line of methods aim to estimate the contribution that each node makes to node $t$'s PageRank. Specifically, as a well-known reverse exploration method, LocalPush~\cite{lofgren2013personalized} reversely explores the graph from the target node $t$ to its ancestors, propagating the probability mass initially at the target node $t$ to its neighbors step by step. 
To be more specific, the LocalPush method repeatedly conducts {\em backward push} operations, updating two variables $\r^b(v)$ and $\epi^b(v)$ for each node $v$ in graph $G$ during the query phase.
In particular, $\r^b(v)$ is called the (reverse) {\em residue} of node $v$, which records the probability mass that is to be reversely pushed from node $v$ to its ancestors. $\epi^b(v)$ is called the (reverse) {\em reserve} of $v$, which records the probability mass that has been received by node $v$ so far. Initially, LocalPush sets $\r^b(v)=\epi^b(v)=0$ for every $v\in V$ except $\r^b(t)=1$. 
During the query phase, 
LocalPush repeatedly conducts the following backward push operations from all nodes $v$ with $\r^{b}(v)\ge \e$. Specifically, in the backward push operation at node $v$, LocalPush updates $\epi^b(v)$ and $\r^b(v)$ as follows: 
\begin{itemize}
    \item convert $\alpha$ fraction of the probability mass currently at $\r^b(v)$ to its reserve: $\epi^b(v) \gets \epi^b(v)+\alpha \cdot \r^b(v)$; 
    \item reversely push the remained mass at $\r^b(v)$ to the neighbors of node $v$: for each $u\in N(v)$, $\r^b(u) \gets \r^b(u)+(1-\alpha)\cdot \frac{\r^b(v)}{d_u}$; 
    \item set $\r^b(v)$ as $0$: $\r^b(v)\gets 0$. 
\end{itemize} 
When no node in graph $G$ has the residue that is larger than $\e \in (0,1)$, the algorithm terminates. 
LocalPush then uses $\epi(t)=\frac{1}{n}\cdot \sum_{u\in V}\epi^b(u)$ as an estimate of $\vpi(t)$. 

In particular, 
Lofgren et al.~\cite{lofgren2013personalized} prove that throughout the backward push process, $\epi^b(v)$ is always an underestimate of $\vpi_v(t)$ that $\vpi_v(t)-\epi^b(v)\le \e$, where $\vpi_v(t)$ denotes the PPR of $t$ (w.r.t node $v$), and $\epi^b(v)$ is the reserve of node $v$. Hence, by setting the push threshold $\e=\frac{\rela \alpha}{n}$, we can derive: 
\begin{align}\label{eqn:localpush_time}
\vpi(t)-\epi(t)=\frac{1}{n}\sum_{v\in V}\left(\vpi_v(t)-\epi^b(v)\right)\le \frac{1}{n}\sum_{v\in V} \frac{\rela \alpha}{n}\le \rela \vpi(t) 
\end{align}
when the LocalPush algorithm terminates. In the last inequality of Equation~\eqref{eqn:localpush_time}, we also adopt the lower bound $\vpi(t)\ge \frac{\alpha}{n}$, as shown in Equation~\eqref{eqn:ite_power_method}. 
In other words, by setting $\e=\frac{\rela \alpha}{n}$, the estimate $\epi(t)$ derived by LocalPush is a $(\rela, \pf)$-approximation of $\vpi(t)$. Furthermore, Lofgren et al.~\cite{lofgren2013personalized} bound the worst-case time complexity of reverse exploration method as $\sum_{u\in V}\frac{\vpi_u(t)\cdot d_u}{\e}$. By plugging into $\e=\frac{\rela \alpha}{n}$ and the reversibility property $\vpi_u(t)\cdot d_u=\vpi_t(v)\cdot d_t$ as shown in Equation~\eqref{eqn:undirectedPPR}, we have: 
\begin{align*}
\sum_{v\in V}\frac{\vpi_v(t)\cdot d_v \cdot n}{\rela \alpha}=\frac{n}{\rela \alpha}\cdot \left(\sum_{v\in V}\vpi_t(v)\cdot d_t\right)=\frac{n\cdot d_t}{\rela \alpha}=O\left(n \cdot d_t\right)    
\end{align*}

Note that the $O(n\cdot d_t)$ complexity may become $O(n^2)$ on some dense graphs where $d_t \to n$. To circumvent this problem, Lofgren and Goel~\cite{lofgren2013personalized} use a priority queue ordered by the residue $\r^b(v)$ of node $v$. Each time we pop off the node $v$ with the greatest $\r^b(v)$ on the graph and conduct the backward push at node $v$. 
As a result, the worse-case time complexity of LocalPush is improved to $\tilde{O}\left(\min\left\{n\cdot d_t, m\right\}\right)$ for deriving a $(\rela, \pf)$-approximation of $\vpi(t)$. 


\subsection{The Hybrid Method} 
Another set of papers~\cite{lofgren2014FastPPR, lofgren2016BiPPR, bressan2018sublinear, wang2020RBS} prove some novel results by combining the Monte-Carlo method and the reverse exploration method together. The key idea is first proposed in FastPPR~\cite{lofgren2014FastPPR}, which introduces a bi-directional approximation algorithm for single-node PageRank:
\begin{align}\label{eqn:bidirection_estimator}
\vspace{-2mm}
\vpi(t)=\sum_{v\in B(t)}\Pr\left\{RW(\alpha)=v\right\} \cdot \vpi_v(t). 
\vspace{-4mm}
\end{align}
Here $B(t)$ is a blanket set of the target node $t$ that all $\alpha$-random walks to node $t$ pass through set $B(t)$. Additionally, $\Pr\left\{RW(\alpha)=v\right\}$ denotes the probability that node $v$ is the first node in $B(t)$ hit by a randomly simulated $\alpha$-random walk. FastPPR first invokes the reverse exploration method to estimate all the PPR values $\vpi_v(t)$ for $v\in V$. Then FastPPR simulates $\alpha$-random walks to collect these estimators according to Equation~\eqref{eqn:bidirection_estimator}. 
As a result, to achieve an $\e$-absolute error of $\vpi(t)$, FastPPR first allows an $\sqrt{\e}$ absolute error for each $\epi_v(t)$ derived in the reverse exploration phase, and only take $\frac{1}{\sqrt{\e}}$ $\alpha$-random walks in the Monte-Carlo simulation phase. Thus, the query time complexity of FastPPR can be bounded by $\frac{1}{\alpha \rela^2}\cdot \sqrt{\frac{d_t\cdot n}{\alpha}}\cdot \sqrt{\frac{\log{\left(1/\pf\right)\cdot \log{\left(n/\alpha\right)}}}{\log{\left(1/(1-\alpha)\right)}}}=\tilde{O}\left(\sqrt{n\cdot d_t}\right)$ for achieving a $(\rela, \pf)$-approximation of $\vpi(t)$. The result is subsequently improved by BiPPR~\cite{lofgren2016BiPPR, lofgren2015BiPPRundirected} to $\frac{1}{\alpha \rela}\cdot \sqrt{\frac{d_t\cdot n}{\alpha}}\cdot \sqrt{\log{(1/\pf)}}=\tilde{O}\left(\sqrt{n\cdot d_t}\right)$. Furthermore, Bressan et al.~\cite{bressan2018sublinear} proposed the \sublinear method, which optimizes the complexity result to {\rev $\tilde{O}\left(\min\left\{\frac{m^{2/3}\dmax^{1/3}}{\davg^{2/3}}, \frac{m^{4/5}}{\davg^{3/5}} \right\}\right)$}. Here $\davg$ and $\dmax$ denote the average and maximum degree of all the nodes in the graph, respectively. 


 \header{\bf The RBS Method. } Reviewing the hybrid methods mentioned above, the Monte-Carlo sampling phase and the reverse exploration phase serve as two separate phases and are conducted sequentially. In comparison, a recent method, RBS~\cite{wang2020RBS}, proposes to mix the two phases in a more flexible way. Specifically, the RBS method follows the framework of the reverse exploration, which reversely propagates the probability mass from the given target node $t$ to its ancestors in the graph. The difference is, in each backward push step (e.g. at node $v$), the RBS method only deterministically pushes the probability mass at $\r^b(v)$ to a small fraction of $v$'s neighbors (i.e., deterministically increase the residue of $u\in N(v)$ if the residue increment $\frac{(1-\alpha)\r^b(v)}{d_u}\ge \theta$, where $\theta$ is a threshold for deterministic push). For the other neighbors $u$, the RBS method generates a uniform random $rand \in (0,1)$ and only updates the residues $\r^b(u)$ if $\frac{(1-\alpha)\r^b(v)}{d_u}\ge rand \cdot \theta$. 
 {\rev 
 By this means, the RBS method avoids to touch all neighbors, and successfully reduces an $O(\davg)$ gap between the time complexity of LocalPush~\cite{andersen2007contribution} and the lower bound for single-target PPR queries. Here $\davg$ denotes the average node degree in the graph. For the single-node PageRank computation, the expected time complexity of RBS can be bounded by $\tilde{O}\left(n\right)$ by setting $\theta=\frac{\rela^2 \cdot \vpi(t)}{12\cdot \log_{1-\alpha}{(c \alpha /2n)}}$. 
 
 The theoretical insight introduced by RBS is encouraging, which enlightens us that we may flexibly mix the deterministic reverse exploration and the randomized Monte-Carlo sampling in each step, instead of separately performing the two phases one by one. 
 }


\vspace{-2mm}
\section{Algorithm}\label{sec:algorithm}
{\rev 

This section presents our \setpush algorithm. Before introducing the details, we first illustrate the reasons why existing methods are unable to achieve the $\tilde{O}\left(\min\left\{d_t, \sqrt{m}\right\}\right)$ time complexity for the single-node PageRank computation on undirected graphs. 
} 

\vspace{-1mm}
\subsection{Limitations of Existing Methods}\label{subsec:reasonswhy}
\begin{itemize}
    \item 
    For the Monte-Carlo method, the lower bound of the query time complexity for deriving a $(c, p_f)$-approximation of $\vpi(t)$ is $\Omega\left(\frac{1}{\vpi(t)}\right)$. By the definition formula of PageRank, the initial probability distribution of simulating $\alpha$-random walks is $\frac{1}{n}\cdot \bm{1}$. Therefore, in the worst-case scenario where $\vpi(t)=O\left(\frac{1}{n}\right)$ (e.g., the node $v$ in Figure~\ref{fig:high_level}), 
    the Monte-Carlo method needs to simulate at least $\Omega\left(n\right)$ $\alpha$-random walks in order to hit node $t$ once. 
    \item For the reverse exploration method, we require at least  $O(d_t)$ time to 
    reversely push the probability mass initially at node $t$ to all of its neighbors (i.e., the $d_t$ neighbors). 
    Consider the node $u$ in Figure~\ref{fig:high_level}, where the neighborhood size of $u$ is $O(n)$. When we conduct backward push operations from node $u$, the time cost has reached $\Omega(n)$ only after the first backward push operation. 

    \item 
    For the hybrid method, 
    the above mentioned limitations still exist. Exceptions are the \sublinear~\cite{bressan2018sublinear} and RBS~\cite{wang2020RBS} methods. 
    \begin{itemize}
        \item 
        The \sublinear method defines a blacklist to record all high-degree nodes in the graph. 
        In the reverse exploration phase, the \sublinear method only performs the backward push operations from the nodes that are excluded from the blacklist. By this means, the \sublinear method effectively mitigates the limitations of the backward push operations as mentioned above. 
        However, the \sublinear method still includes a Monte-Carlo sampling phase to simulate $\alpha$-random walks from a uniformly selected source node. Thus, the lower bound of $\Omega \left(1/\vpi(t)\right)$ for the query time complexity of the Monte-Carlo sampling methods still exists, which hinders the \sublinear method from achieving the $\tilde{O}\left(\min\left\{d_t, \sqrt{m}\right\}\right)$ time complexity for the single-node PageRank computation on undirected graphs. 

        \item For RBS, its major drawback comes from the sampling operation that RBS adopts in each backward push operation. Specifically, the sampling operation adopted in each backward push operation of RBS is non-independent. Consider the bad case scenario as shown in Figure~\ref{fig:related}. The residue increment $\frac{(1-\alpha)\r^b(v)}{d_u}$ of each neighbor $u\in N(t)$ is identical. 
        As a result, for all neighbors $u\in N(t)$, the conditions to conduct a randomized push (i.e., $\frac{(1-\alpha)\r^b(v)}{d_u}\ge rand \cdot \theta$) are satisfied simultaneously, which, again, leads to the $O(d_t)=O(n)$ time cost in such bad case scenario. 
    \end{itemize}
\end{itemize}
In the following, we shall describe our \setpush in details and explain the superiority of our \setpush over existing methods. Specifically, we first define a concept called {\em truncated PageRank} in Section~\ref{subsec:trunc_pagerank}. Our \setpush is based on a $\left(\frac{\rela}{2}, \pf\right)$-approximation of the truncated PageRank. After that, in Section~\ref{subsec:highlevelidea} and ~\ref{subsec:setpushalg}, we provide the high-level ideas and detailed algorithm structure of \setpush.

\subsection{Truncated PageRank}\label{subsec:trunc_pagerank}
Given a target node $t$ in an undirected graph $G=(V,E)$, a constant damping factor $\alpha\in (0,1)$, and a constant relative error $\rela$, we refer to $\bpi(t)$ as the {\em truncated PageRank} of node $t$ if 
\begin{align}\label{eqn:def_trunc_pagerank}
\bpi(t)=\frac{1}{n}\cdot \sum_{s\in V}\sum_{\ell=0}^{L} \vpi^{(\ell)}_s(t), 
\end{align}
where $L=\log_{1-\alpha}\frac{\rela \alpha}{2n}=O\left(\log{n}\right)$. Analogously, we call the $n$-dimensional vector $\bpi=\frac{1}{n}\cdot \sum_{s\in V}\sum_{\ell=0}^{L} \vpi^{(\ell)}_s$ the truncated PageRank vector. By Equation~\eqref{eqn:ite_power_method_ppr}, Equation~\eqref{eqn:PageRank_PPR} and Equation~\eqref{eqn:def_lhopppr}, we can further derive: 
\begin{align*}
\vspace{-2mm}
\vpi=\frac{1}{n}\cdot \sum_{s\in V}\sum_{\ell=0}^{\infty}\vpi^{(\ell)}_s=\bpi+\frac{1}{n}\cdot \sum_{s\in V}\sum_{\ell=L+1}^{\infty}\alpha (1-\alpha)^\ell \cdot \left(\A \D^{-1}\right)^\ell \cdot \bm{e}_s. 
\end{align*}
Therefore, for every $t\in V$, we have: 
\begin{align*}
\vpi(t)=\bpi(t)+\frac{1}{n}\cdot \sum_{s\in V}\sum_{\ell=L+1}^{\infty}\alpha (1-\alpha)^\ell \cdot \bm{e}_t^{\top}\cdot \left(\A \D^{-1}\right)^\ell \cdot \bm{e}_s. 
\end{align*}
We note that for each $\ell\in \{0, 1, 2, \ldots\}$, $\left(\bm{e}_t^{\top}\cdot \left(\A \D^{-1}\right)^\ell \cdot \bm{e}_s\right) \in [0,1]$. Thus, we have $\frac{1}{n}\cdot \sum_{s\in V}\bm{e}_t^{\top}\cdot \left(\A \D^{-1}\right)^\ell \cdot \bm{e}_s \le 1$. As a consequence, we can derive $\vpi(t)\le \bpi(t)+\sum_{\ell=L+1}^{\infty} \alpha (1-\alpha)^\ell=\bpi(t)+(1-\alpha)^{L+1}$. 
Recall that $L=\log_{1-\alpha} \frac{\rela \alpha}{2n}$. Then it follows: 
\begin{align}\label{eqn:bpagerank_basic}
\vpi(t) \le \bpi(t)+ \frac{c}{2}\cdot \frac{\alpha}{n}\le \bpi(t)+ \frac{c}{2}\cdot \vpi(t), 
\end{align}
where we apply the lower bound of $\vpi(t)$ that $\vpi(t)\ge \frac{\alpha}{n}$ as shown in Equation~\eqref{eqn:ite_pagerank}. Furthermore, Lemma~\ref{lem:trunc_pagerank_approx} implies that deriving a $(\rela, \pf)$-approximation of $\vpi(t)$ can be achieved by deriving a $(\frac{c}{2}, p_f)$-approximation of $\bpi(t)$. 

\begin{lemma}\label{lem:trunc_pagerank_approx}
Given a target node $t$ in the graph $G=(V,E)$, $\epi(t)$ is a $(\rela, \pf)$-approximation of node $t$'s PageRank $\vpi(t)$ if 
\begin{align*}
|\epi(t)-\bpi(t)|\le \frac{\rela}{2}\cdot \vpi(t)
\end{align*}
holds with probability at least $1-\pf$. 
\end{lemma}

\begin{proof}
For each node $t\in V$, we observe: 
\begin{equation*}
\begin{aligned}
&\left|\epi(t)-\vpi(t)\right|=\left|\epi(t)-\bpi(t)+\bpi(t)-\vpi(t)\right|\\
&\le \left|\epi(t)-\bpi(t)\right|+\left|\bpi(t)-\vpi(t)\right|\le \left|\epi(t)-\bpi(t)\right|+\frac{c}{2}\cdot \vpi(t), 
\end{aligned}    
\end{equation*}
where we plugging Equation~\eqref{eqn:bpagerank_basic} into the last inequality. Thus, if $\left|\epi(t)\hspace{-0.5mm}-\hspace{-0.5mm}\bpi(t)\right|\hspace{-0.5mm}\le \hspace{-0.5mm}\frac{c}{2}\cdot \vpi(t)$ holds with probability at least $1-\pf$, $\epi(t)$ is a $(\rela, \pf)$-approximation of $\vpi(t)$, which follows the lemma. 
\end{proof}


\subsection{Key Idea of \setpush}\label{subsec:highlevelidea}
Given an undirected graph $G=(V,E)$ and a target node $t$, our \setpush computes a $(c,p_f)$-approximation of node $t$'s PageRank by deriving a $(c/2, p_f)$-approximation $\epi(t)$ of $\bpi(t)$ following
\begin{align}\label{eqn:actual_eqn}
\epi(t)=\frac{1}{n}\cdot \sum_{s\in V} \sum_{\ell=0}^L \frac{d_t}{d_s}\cdot \epi_t^{(\ell)}(s). 
\end{align}
In particular, $\epi_t^{(\ell)}(s)$ is an unbiased estimator of the $\ell$-hop PPR value $\vpi_t^{(\ell)}(s)$. To understand Equation~\eqref{eqn:actual_eqn}, recall that $\vpi_t^{(\ell)}(s)\cdot d_t=\vpi_s^{(\ell)}(t)\cdot d_s$ as shown in Equation~\eqref{eqn:undirectedPPR}. Thus, if for each $s\in V$, $\epi_t^{(\ell)}(s)$ is an unbiased estimator of the $\ell$-hop PPR value $\vpi_t^{(\ell)}(s)$, then $\frac{d_t}{d_s}\cdot \epi_t^{(\ell)}(s)$ is an unbiased estimator of $\epi_s^{(\ell)}(t)$. According to the definition formula of the truncated PageRank $\bpi(t)$ as shown in Equation~\eqref{eqn:def_trunc_pagerank}, $\epi(t)$ is therefore an unbiased estimator of $\vpi(t)$. 

To compute $\epi_t^{(\ell)}(s)$, we maintain a variable called {\em $\ell$-hop residue} $\r^{(\ell)}_t(u)$ for each node $u$ in $G$. Initially, we set $\r^{(\ell)}_t\hspace{-1mm}=\bm{0}$ for $\forall \ell \hspace{-0.5mm}\in \hspace{-0.5mm}\{1,2, \ldots, L\}$ and $\r^{(0)}_t\hspace{-1mm}=\hspace{-0.5mm}\bm{e}_t$, where $\bm{0}$ is an $n$-dimensional all zero vector. During the query phase, we repeatedly conduct the following steps to update $\r^{(\ell+1)}_t$ based on $\r^{(\ell)}_t$ by iterating $\ell$ from $0$ to $L\hspace{-0.5mm}-\hspace{-0.5mm}1$: 
\begin{itemize}
\item Pick a node $u$ with nonzero $\r^{(\ell)}_t(u)$; 
\item If $(1-\alpha)\cdot \r^{(\ell)}_t(u) \ge \theta \cdot d_u$, we uniformly distribute $(1-\alpha)\cdot \r^{(\ell)}_t(u)$ to the $(\ell+1)$-hop residue $\r_t^{(\ell+1)}(v)$ of each $v\in N(u)$. To be more specific, for $\forall v\in N(u)$, $\r^{(\ell+1)}_t(v)\gets \r^{(\ell+1)}_t(v)+\frac{(1-\alpha)}{d_u}\cdot \r^{(\ell)}_t(u)$. Note that $\theta \in (0,1)$ is a tunable threshold and we provide a detailed analysis to the choice of $\theta$ in Section~\ref{sec:analysis}. 
\item Otherwise, we independently select some neighbors of $u$, and only distribute the probability mass at $\r^{(\ell)}_t(u)$ to those sampled neighbors. Notably, for each $v\in N(u)$, the expectation of $\r_t^{(\ell+1)}(v)$'s increment is still guaranteed to be $\frac{(1-\alpha)}{d_u}\cdot \r^{(\ell)}_t(u)$. 
\end{itemize}

After all the $L$ iterations have been processed, we return $\epi(t)=\frac{1}{n}\cdot \sum_{s\in V} \sum_{\ell=0}^L \frac{d_t}{d_s}\cdot \alpha \cdot  \r^{(\ell)}_t(s)$ as an estimator of $\vpi(t)$. 

As we shall demonstrate in Section~\ref{sec:analysis}, the $\ell$-hop residue vector $\r_t^{(\ell)}$ is an unbiased estimate of $\frac{1}{\alpha}\cdot \vpi_t^{(\ell)}$. In other words, $\E\left[\r_t^{(\ell)}(u)\right]=\frac{1}{\alpha}\cdot \vpi_t^{(\ell)}(u)$ holds for each $u\in V$. 
To see this, we observe that $\vpi^{(0)}_t=\alpha \cdot \bm{e}_t$ holds by definition. Recall that we set $\r^{(0)}_t=\bm{e}_t$ as mentioned above. Therefore, $\E\left[\r^{(\ell)}_t\right]=\frac{1}{\alpha}\cdot \vpi^{(\ell)}_t$ holds when $\ell=0$. Furthermore, let us assume $\r^{(\ell)}_t=\frac{1}{\alpha}\cdot \vpi^{(\ell)}_t$ holds for any $i \in [0,\ell]$. Then for each $v\in V$, the expectation of $\r^{(\ell+1)}_t(v)$ satisfies: 
\begin{align*}
\E\left[\r^{(\ell+1)}_t(v)\right]=\hspace{-2mm}\sum_{u\in N(v)}\hspace{-2mm}\frac{(1-\alpha)}{d_u}\cdot \E\left[\r^{(\ell)}_t(u)\right]=\hspace{-2mm}\sum_{u\in N(v)}\hspace{-2mm}\frac{(1-\alpha)}{d_u}\cdot \frac{\vpi^{(\ell)}_t(u)}{\alpha}. 
\end{align*}
By Equation~\eqref{eqn:PPR_recur}, we can therefore derive $\E\left[\r^{(\ell+1)}_t(v)\right]=\frac{1}{\alpha}\cdot \vpi^{(\ell+1)}_t(v)$. Consequently, for every $\ell\in \{1, \ldots, L\}$, $\E\left[\r^{(\ell)}_t\right]\hspace{-1mm}=\frac{1}{\alpha}\cdot \vpi^{(\ell)}_t$ holds by induction. The formal proof can be found in Section~\ref{sec:analysis}. 

Furthermore, it can be proved that $\epi(t)$ is also an unbiased estimator of the truncated PageRank $\bpi(t)$. Specifically, recall that $\epi(t)=\frac{1}{n}\cdot \sum_{s\in V}\sum_{\ell=0}^{L} \frac{d_t}{d_s}\cdot \alpha \cdot \er^{(\ell)}_t(s)$ according to Algorithm~\ref{alg:VBES}. By applying the linearity of expectation, we can thus derive 
\begin{align*}
\E\left[\epi(t)\right]=\frac{1}{n}\cdot \sum_{s\in V}\sum_{\ell=0}^{L} \frac{d_t}{d_s}\cdot \alpha \cdot \E \left[\r_t^{(\ell)}(s)\right]=\frac{1}{n}\cdot \hspace{-1mm}\sum_{s\in V}\sum_{\ell=0}^{L} \frac{d_t}{d_s}\cdot \vpi_t^{(\ell)}(s).     
\end{align*}
Recall that in Equation~\eqref{eqn:undirectedPPR}, we show that $\frac{d_t}{d_s}\cdot \vpi^{(\ell)}_t(s)=\vpi^{(\ell)}_s(t)$, following $\E\left[\epi(t)\right]=\frac{1}{n}\cdot \hspace{-1mm}\sum_{s\in V}\sum_{\ell=0}^{L} \vpi_s^{(\ell)}(t)=\bpi(t)$. 


{\rev 
\header{\bf Advantages of the Push Operation Adopted in \setpush. }
}
{\rev 
Note that the $\ell$-hop residue $\r^{(\ell)}_t(u)$ defined above is similar in spirit to the one used in the vanilla backward push operation adopted in the reverse exploration method (see Section~\ref{subsec:reverse_exploration}), but differs in two crucial aspects as described below. 
\begin{itemize}
    \item To distribute the probability mass maintained at $\r^{(\ell)}_t(u)$, the backward push operation (except in RBS~\cite{wang2020RBS}) 
    touches every neighbor $v$ of $u$ to update the residue of $v$, which costs $O(d_u)$ deterministically. In comparison, for the node $u$ with $(1-\alpha)\cdot \r^{(\ell)}_t(u)\le \theta \cdot d_u$, we only select some neighbors $v\in N(u)$ to update $\r^{(\ell+1)}_t(v)$. Therefore, the time cost of each update process is only proportional to the size of the sampled outcomes. By this means, we successfully avoid the $O(d_u)$ term of time complexity introduced by the vanilla backward push. 
    \item Compared to the RBS method, we independently sample the neighbors $v$ from $N(u)$ to update $\r^{(\ell+1)}_t(v)$. As a consequence, the increment of $\r^{(\ell+1)}_t(v)$ for each $v\in V$ is independent with each other. In contrast, the sampling technique adopted in the RBS method~\cite{wang2020RBS} is non-independent, resulting in either large variance or expensive time cost. For example, consider the graph shown in Figure~\ref{fig:related} with node $t$ as the given target node. For the RBS method, the sampling condition of each $u\in N(t)$ is satisfied simultaneously, which costs either $O(n)$ time or unbounded approximation error. Instead, in \setpush, we can independently some $u\in N(t)$ to update $\r^{(\ell)}_t(u)$. 
\end{itemize}
}

\subsection{The \setpush Algorithm}\label{subsec:setpushalg}
\begin{algorithm}[t]
\caption{The \setpush Algorithm}
\label{alg:VBES}
\BlankLine
\KwIn{Undirected graph $G=(V,E)$, target node $t\in V$, constant damping factor $\alpha$, threshold $\theta$\\}
\KwOut{Estimator of $\vpi(t)$\\}
Initialize two $n$-dimensional vectors $\er^{(0)}_t \hspace{-1.5mm}\gets \hspace{-0.5mm}\bm{e}_t$ and $\epi_t \hspace{-0.5mm} \gets \hspace{-0.5mm} \alpha \bm{e}_t$\; 
$L \gets \log_{1-\alpha}\frac{c\alpha}{2n}$\;
\For{$\ell$ from $0$ to $L-1$}{
    Initialize an $n$-dimensional vector $\er^{(\ell+1)}_t \gets \bm{0}$\; 
    \For{each $u\in V$ with nonzero $\er^{(\ell)}_t(u)$}{
        \If{$(1-\alpha)\cdot \er^{(\ell)}_t(u)\ge \th\cdot d_u$}{
            \For{each $v\in N(u)$}{
                $\er^{(\ell+1)}_t(v)\gets \er^{(\ell+1)}_t(v)+\frac{(1-\alpha)}{d_u}\cdot \er^{(\ell)}_t(u)$\;
            }
        }
        \Else{
            Let $idx\gets 0$, and $p^*\gets \frac{(1-\alpha)\cdot \er^{(\ell)}_t(u)}{d_u\cdot \th}$\; 
            \While{true}{
                Generate a geometrical random $rg \sim G(p^*)$\; 
                $idx \gets idx +rg$\; 
                \If{$idx > d_u$}{
                    break\;
                }
                Let $v$ denote the $idx$-th node in $N(u)$\; 
                $\er^{(\ell+1)}_t(v)\gets \er^{(\ell+1)}_t(v)+\th$\;
            }
        }
    }
    Clear $\er^{(\ell)}_t$\; 
    $\epi_t \gets \epi_t+\alpha \cdot \er^{(\ell+1)}_t$\;
}
$\epi(t) \gets \frac{1}{n}\cdot \sum_{s\in V}\frac{d_t}{d_s}\cdot \epi_t(s)$\;
\Return $\epi(t)$ as an estimator of $\vpi(t)$;
\end{algorithm}

Algorithm~\ref{alg:VBES} illustrates the pseudocode of \setpush. Consider an undirected graph $G=(V,E)$, a target node $t$, a constant damping factor $\alpha\in (0,1)$ and a threshold parameter $\theta \in (0,1)$. Initially, we set $\r^{(0)}_t=\bm{e}_t$ and iteratively conduct the update process as described in Section~\ref{subsec:highlevelidea} from $\ell=0$ to $L-1$, where $L=\log_{1-\alpha}\frac{\rela \alpha}{2n}$. In particular, for the node $u$ with $0<(1-\alpha)\cdot \r^{(\ell)}_t(u)\le \theta \cdot d_u$, we adopt a {\em geometric sampling operation} to independently select neighbors $v$ from $N(u)$. Specifically, we independently sample every $v\in N(u)$ with probability $p^*=\frac{(1-\alpha)\cdot \r^{(\ell)}_t(u)}{d_u \cdot \th}$. For each sampled $v\in N(u)$, we increase the residue $\r_t^{(\ell+1)}(v)$ by $\th$. By this means, the expectation of $\r_t^{(\ell+1)}(v)$'s increment is still $\frac{(1-\alpha)}{d_u}\cdot \r^{(\ell)}_t(u)$. It's worth noting that we aim to complete the above described sampling process using the time of $O(d_u \cdot p^*)$. In other words, we require the expected time cost of the above described sampling process is asymptotically the same to the expected size of the sampling outcomes (i.e., the expected number of $u$'s neighbors that are successfully sampled). To achieve this goal, we define a variable $idx$ for referring to the index of $u$'s neighbor in $N(u)$ that is successfully sampled. Initially, we set $idx$ as $0$. Moreover, we define a geometric random number $rg$, and repeatedly generate $rg$ according to the geometric distribution $G(p^*)$. According to~\cite{devroye2006nonuniform,bringmann2012efficient}, a geometric random number can be generated in $O(1)$ time. 
We repeatedly generate $rg \sim G(p^*)$, update $idx \gets idx+rg$ and increase the residue $\r_t^{(\ell+1)}(v)$ of the $idx$-th neighbor $v$ in $N(u)$ by $\th$, until $idx > d_u$. 

To understand the sampling process mentioned above, recall that a geometric random number $rg\sim G(p^*)$ indicates the number of Bernoulli trials needed to get one success, where
each Bernoulli trial has two Boolean-valued outcomes: success (with probability $p^*$) and failure (with probability $1-p^*$). Therefore, by generating $rg \sim G(p^*)$, we are able to derive the index of the first sampled node in $N(u)$, using only $O(1)$ time. We iteratively generate $rg \sim G(p^*)$ to derive the index of the next sampled node from the index of the last sampled neighbor (recorded by $idx$). By this means, we are able to independently select each neighbor $v$ from $N(u)$ with probability $p^*$ using only $O(d_u\cdot p^*)=O\left(\frac{(1-\alpha)\cdot \r_t^{(\ell)}(u)}{\theta}\right)$ time in expectation. By carefully setting the value of $\theta$ (see Section~\ref{sec:analysis} for details), the expected time cost of \setpush can be consequently bounded by $O\left(\min\left\{d_t, \sqrt{m}\right\}\right)$. 
Additionally, after the $\ell$-th iteration ($\forall \ell \in \{0,1, \ldots, L-1\}$), we clear the $\ell$-hop residue vector $\r_t^{(\ell)}$ to save memory. Finally, we return $\epi(t)=\frac{1}{n}\cdot \sum_{s\in V} \sum_{\ell=0}^L \frac{d_t}{d_s}\cdot \alpha \r^{(\ell)}_t(s)$ as the estimator of $\vpi(t)$.




\section{Theoretical Analysis}
\label{sec:analysis}
In this section, we analyze the theoretical properties of our \setpush.

\subsection{Correctness}\label{subsec:correctness}
Recall that we have presented some intuitions on $\E\left[\r^{(\ell)}_t(u)\right]=\frac{1}{\alpha}\cdot \vpi^{(\ell)}_t(u)$ and $\E\left[\epi(t)\right]=\bpi(t)$ in Section~\ref{subsec:highlevelidea}. The following Lemmas further provide formal proofs on these intuitions. 

\begin{lemma}\label{lem:unbiasedness_er}
For each $\ell\in \{0,1,\ldots, L\}$ The residue vector $\er^{(\ell)}_t$ obtained in Algorithm~\ref{alg:VBES} is an unbiased estimator of $\frac{1}{\alpha}\cdot \vpi^{(\ell)}_t$, such that for each $v\in V$,   
\begin{align*}
\E \left[\er^{(\ell)}_t(v)\right]=\frac{1}{\alpha}\cdot \vpi^{(\ell)}_t(v).     
\end{align*}
\end{lemma}

\begin{proof}
Let $\incre^{(\ell+1)}(u,v)$ denote the increment of $\r^{(\ell+1)}_t(v)$ in the update procedure conducted at node $u$ with nonzero $\r^{(\ell)}_t(u)$. 
According to Algorithm~\ref{alg:VBES}, for each node $u\in V$ with nonzero $\er^{(\ell)}_t(u)$, $\incre^{(\ell+1)}(u,v)=\frac{1-\alpha}{d_u}\cdot \er^{(\ell)}_t(u)$ deterministically if $\frac{(1-\alpha)}{d_u}\cdot \er^{(\ell)}_t(u)\ge \th$. Otherwise, $\incre^{(\ell+1)}(u,v)=\th$ with probability $\frac{(1-\alpha)}{d_u \cdot \th}\cdot \er^{(\ell)}_t(u)$, or $0$ with probability $1-\frac{(1-\alpha)}{d_u \cdot \th}\cdot \er^{(\ell)}_t(u)$. As a consequence, the expectation of $\incre^{(\ell+1)}(u,v)$ equals $\th \cdot \frac{(1-\alpha)}{d_u \cdot \th}\cdot \er^{(\ell)}_t(u)=\frac{(1-\alpha)}{d_u}\cdot \er^{(\ell)}_t(u)$. More specifically, we have: 
\begin{align*}
\E \left[\incre^{(\ell+1)}(u,v) \mid \er^{(\ell)}_t\right]=\frac{(1-\alpha)}{d_u}\cdot \er^{(\ell)}_t(u). 
\end{align*}
where $\E \left[\incre^{(\ell+1)}(u,v) \mid \er^{(\ell)}_t\right]$ denotes the expectation of $\incre^{(\ell+1)}(u,v)$ conditioned on the fact that the $\ell$-hop residue $\r^{(\ell)}_t$ has been derived. Furthermore, since $\er^{(\ell+1)}_t(v)=\sum_{u\in N(v)}\incre^{(\ell+1)}(u,v)$, we can derive: 
\begin{align}\label{eqn:conditional_exp}
\E \left[\er^{(\ell+1)}_t(v)~\big|~\er^{(\ell)}_t\right]\hspace{-1mm}=\hspace{-3mm}\sum_{u\in N(v)}\hspace{-3mm}\E \left[\incre^{(\ell+1)}(u,v) ~\big|~\er^{(\ell)}_t\right]\hspace{-1mm}=\hspace{-3mm}\sum_{u\in N(v)}\hspace{-3mm}\frac{(1-\alpha)}{d_u}\hspace{-0.5mm}\cdot \hspace{-0.5mm}\er^{(\ell)}_t(u)
\end{align}
by applying the linearity of expectation. Given the fact: $\E[\er^{(\ell+1)}(v)]=\E\left[ \E \left[\er^{(\ell+1)}(v)~\big|~\er^{(\ell)}_t\right]\right]$, we can further derive: 
\begin{align}\label{eqn:recur_er}
\E \left[\er^{(\ell+1)}(v)\right]=\sum_{u\in N(v)}\frac{(1-\alpha)}{d_u}\cdot \E\left[\er^{(\ell)}_t(u)\right].
\end{align}
Based on the recursive formula as shown in Equation~\eqref{eqn:recur_er}, we are able to prove Lemma~\ref{lem:unbiasedness_er} by mathematical induction. Specifically, the base case $\er^{(0)}_t=\bm{e}_s=\frac{1}{\alpha}\cdot \vpi^{(0)}_t$ holds by definition. For the inductive case, assuming that $\E\left[\er^{(\ell)}_t(u)\right]=\frac{\vpi^{(\ell)}_t(u)}{\alpha}$ holds for each $u\in V$ and some $\ell\in \{0,1,\ldots, L-1\}$. By Equation~\eqref{eqn:recur_er}, we have:   
\begin{align*}
\E\left[\er^{(\ell+1)}_t(v)\right]\hspace{-1mm}=\hspace{-1mm}\frac{1}{\alpha}\cdot \hspace{-2mm}\sum_{u\in N(v)}\hspace{-1mm}\frac{(1-\alpha)}{d_u}\cdot \vpi^{(\ell)}_t(u)=\frac{1}{\alpha}\cdot \vpi^{(\ell+1)}_t(v),  
\end{align*}
where we apply the fact that $\vpi^{(\ell+1)}_t(v)=\sum_{u\in N(v)}\frac{(1-\alpha)}{d_u}\cdot \vpi^{(\ell)}_t(u)$ as shown in Equation~\eqref{eqn:PPR_recur}. Consequently, the inductive case holds, and Lemma~\ref{lem:unbiasedness_er} follows. 
\end{proof}

Based on Lemma~\ref{lem:unbiasedness_er}, we are able to prove that Algorithm~\ref{alg:VBES} returns an unbiased estimator of the truncated PageRank $\bpi(t)$. 

\begin{lemma}\label{lem:unbiasedness_ppr}
Algorithm~\ref{alg:VBES} returns an unbiased estimator $\epi(t)$ of the truncated PageRank score of node $t$. Specifically, $\E[\epi(t)]=\bpi(t)$. 
\end{lemma}

\begin{proof}
Note that Algorithm~\ref{alg:VBES} computes $\epi(t)$ as: 
\begin{align*}
\epi(t)=\frac{1}{n}\cdot \sum_{s\in V}\sum_{\ell=0}^{L} \frac{d_t}{d_s}\cdot \alpha \cdot \er^{(\ell)}_t(s). 
\end{align*}
By applying the linearity of expectation, we can derive: 
\begin{align*}
\E\left[\epi(t)\right]=\frac{1}{n}\cdot \hspace{-1mm}\sum_{s\in V}\sum_{\ell=0}^{L} \frac{d_t}{d_s}\cdot \alpha \cdot \E\left[\er^{(\ell)}_t(s)\right]=\frac{1}{n}\cdot \hspace{-1mm}\sum_{s\in V}\sum_{\ell=0}^{L} \frac{d_t}{d_s}\cdot \vpi^{(\ell)}_t(s), 
\end{align*}
where we employ the expectation of $\er^{(\ell)}_t(s)$ derived in Lemma~\ref{lem:unbiasedness_er}. Furthermore, by Equation~\eqref{eqn:undirectedPPR}, we have $\frac{d_t}{d_s}\cdot \vpi^{(\ell)}_t(s)=\vpi^{(\ell)}_s(t)$. Thus,  we can derive: 
\begin{align*}
\E\left[\epi(t)\right]=\frac{1}{n}\cdot \sum_{s\in V}\sum_{\ell=0}^{L} \vpi^{(\ell)}_s(t)=\bpi(t), 
\end{align*}
which follows the lemma. 
\end{proof}

Up to now, we have proved that $\epi(t)$ is an unbiased estimator of the truncated PageRank $\vpi(t)$. Next, we shall bound the variance of $\epi(t)$  and utilize the following Chebyshev Inequality~\cite{mitzenmacher2017probability} to bound the failure probability for deriving a $(\frac{c}{2}, p_f)$-approximation of $\bpi(t)$. 

\begin{fact}[Chebyshev's Inequality~\cite{mitzenmacher2017probability}]\label{fact:chebyshev}
Let $X$ denote a random variable. For any real number $\e>0$, $\Pr\left\{\left|X-\E[X]\right|\ge \e\right\}\le \frac{\Var[X]}{\e^2}$. 
\end{fact}

\subsection{Variance Analysis}\label{subsec:variance}
We claim that the variance of $\epi(t)$ can be bounded by $\frac{L\theta d_t}{n}\cdot \vpi(t)$, which is formally demonstrated in Theorem~\ref{thm:variance}. 

\begin{theorem}[Variance]\label{thm:variance}
The variance of the estimator $\epi(t)$ returned by Algorithm~\ref{alg:VBES} can be bounded as $\Var[\epi(t)]\le \frac{L \theta d_t }{n}\cdot \vpi(t)$. 
\end{theorem}

To prove Theorem~\ref{thm:variance}, we need several technical lemmas. Specifically, in Lemma~\ref{lem:conditional_variance}, we bound the variance of $\er^{(\ell+1)}_t(v)$ conditioned on $\er^{(\ell)}_t$ that is derived in the $\ell$-th iteration. 

\begin{lemma}\label{lem:conditional_variance} 
For each node $v\in V$ and each $\ell \in \{0, 1, \ldots, L-1\}$, the variance of $\er^{(\ell+1)}_t(v)$ can be bounded as
\begin{align*}
\Var\left[\er^{(\ell+1)}_t(v)~\big|~ \er^{(\ell)}_t\right]\le \sum_{u\in N(v)}\th \cdot \frac{(1-\alpha)\cdot \er^{(\ell)}_t(u)}{d_u}, 
\end{align*}
where $\Var\left[\er^{(\ell+1)}_t(v)~\big|~ \er^{(\ell)}_t\right]$ denotes the variance of $\er^{(\ell+1)}_t(v)$ conditioned on the value of $\er^{(\ell)}_t$ that has been derived in the $\ell$-th iteration. 
\end{lemma}

\begin{proof}
Recall that in the proof of Lemma~\ref{lem:unbiasedness_er}, we use $\incre^{(\ell+1)}(u,v)$ to denote the increment of $\r_t^{(\ell+1)}(v)$ in the update operations conducted at node $u$ with nonzero $\r_t^{(\ell)}(u)$. 
For the deterministic case when $\frac{(1-\alpha)}{d_u}\cdot \er^{(\ell)}_t(u)\ge \th$, $\incre^{(\ell+1)}(u,v)$ is deterministically set as $\frac{(1-\alpha)}{d_u}\cdot \er^{(\ell)}_t(u)$, and thus there is no variance caused. For the randomized case when  $\frac{(1-\alpha)}{d_u}\cdot \er^{(\ell)}_t(u)<\th$, we set $\incre^{(\ell+1)}(u,v)$ as $\th$ with probability $\frac{(1-\alpha)}{d_u\cdot \th}\cdot \er^{(\ell)}_t(u)$, or as $0$ with probability $1-\frac{(1-\alpha)}{d_u\cdot \th}\cdot \er^{(\ell)}_t(u)$. Therefore, in the randomized case, the variance of $\incre^{(\ell+1)}(u,v)$ conditioned on the residue vector $\er^{(\ell)}_t$ that has been derived in previous iterations can be bounded as: 
\begin{equation*}
\begin{aligned}
&\Var\left[\left. \incre^{(\ell+1)}(u,v)~\right|~\er^{(\ell)}_t \right]\le \E\left[\left. \left(\incre^{(\ell+1)}(u,v)\right)^2~\right|~\er^{(\ell)}_t \right]\\
&= \th^2 \cdot \frac{(1-\alpha)}{d_u\cdot \th}\cdot \er^{(\ell)}_t(u)=\th\cdot \frac{(1-\alpha)}{d_u}\cdot \er^{(\ell)}_t(u). 
\end{aligned}
\end{equation*}
Since $\er^{(\ell+1)}_t(v)=\sum_{u\in N(v)}\incre^{(\ell+1)}(u,v)$, we can further derive: 
\begin{align*}
\Var\left[\left. \er^{(\ell+1)}_t(v)~\right|~\er^{(\ell)}_t \right]=\Var\left[\left. \sum_{u\in N(v)}\incre^{(\ell+1)}(u,v)~\right|~\er^{(\ell)}_t \right]. 
\end{align*}
Notably, for each $u\in N(v)$, $\incre^{(\ell+1)}(u,v)$ is independent with each other according to the sampling procedures as described in Section~\ref{alg:VBES}. Thus, we can further derive: 
\begin{equation*}
\begin{aligned}
\Var\left[\hspace{-0.5mm}\left. \er^{(\ell+1)}_t\hspace{-0.5mm}(v)\right|\er^{(\ell)}_t \hspace{-0.5mm}\right]\hspace{-1mm}=\hspace{-3mm}\sum_{u\in N(v)}\hspace{-3mm}\Var\hspace{-0.5mm}\left[\hspace{-0.5mm}\left. \incre^{(\ell+1)}\hspace{-0.5mm}(u,v)\right|\er^{(\ell)}_t \hspace{-0.5mm}\right]\hspace{-1mm}\le \hspace{-3mm}\sum_{u\in N(v)}\hspace{-3.5mm}\frac{\th \hspace{-0.5mm}\cdot \hspace{-0.5mm}(1\hspace{-0.5mm}-\hspace{-0.5mm}\alpha)}{d_u}\hspace{-0.5mm}\cdot \hspace{-0.5mm}\er^{(\ell)}_t\hspace{-0.5mm}(u), 
\end{aligned}    
\end{equation*}
which follows the lemma. 
\end{proof}

In the second step, we prove: 
\begin{lemma}\label{lem:var_recur}
The variance of the estimator $\epi(t)$ obtained by Algorithm~\ref{alg:VBES} can be computed as: 
\begin{equation*}
\begin{aligned}
&\Var\left[\epi(t)\right]=\frac{\alpha^2}{n^2}\cdot \Var\left[\sum_{\ell=0}^L \sum_{s\in V}\frac{d_t}{d_s}\cdot \er^{(\ell)}_t(s)\right]\\
&=\hspace{-0.5mm}\frac{\alpha^2}{n^2}\hspace{-0.5mm}\cdot \sum_{\ell=0}^{L-2}\hspace{-0.5mm}\E\left[\Var\hspace{-0.5mm}\left[\left.\sum_{v\in V}\left(\sum_{s\in V}\frac{d_t}{d_s}\cdot \hspace{-1mm}\sum_{i=0}^{L-\ell-1}\hspace{-0.5mm}\frac{\vpi^{(i)}_v\hspace{-0.5mm}(s)}{\alpha}\right)\cdot \hspace{-0.5mm}\er^{(\ell+1)}_t(v)~\right|~\er^{(\ell)}_t\right]\right]. 
\end{aligned}    
\end{equation*}
\end{lemma}
To prove Lemma~\ref{lem:var_recur}, recall that $\epi(t)=\frac{1}{n}\sum_{s\in V}\frac{d_t}{d_s}\cdot \epi_t(s)$, and $\epi_t(s)=\sum_{\ell=0}^L \alpha \er^{(\ell)}_t(s)$ according to Algorithm~\ref{alg:VBES}. Thus, the variance of $\epi(t)$ derived by Algorithm~\ref{alg:VBES} can be computed as: 
\begin{align*}
\Var\hspace{-0.5mm}\left[\epi(t)\right]\hspace{-0.5mm}=\hspace{-0.5mm}\Var\left[\hspace{-0.5mm}\frac{1}{n}\hspace{-0.5mm}\cdot \hspace{-1mm}\sum_{s\in V}\hspace{-0.5mm}\frac{d_t}{d_s}\hspace{-0.5mm}\cdot \hspace{-1mm}\sum_{\ell=0}^{L}\hspace{-0.5mm}\alpha \hspace{-0.5mm}\cdot \hspace{-0.5mm}\er_t^{(\ell)}\hspace{-0.5mm}(s)\hspace{-0.5mm}\right]\hspace{-1mm}=\hspace{-0.5mm}\frac{\alpha^2}{n^2}\hspace{-0.5mm}\cdot \hspace{-0.5mm}\Var\left[\sum_{s\in V}\hspace{-1mm}\frac{d_t}{d_s}\hspace{-1mm}\cdot \hspace{-1mm}\sum_{\ell=0}^{L}\hspace{-0.5mm}\er_t^{(\ell)}\hspace{-0.5mm}(s)\hspace{-0.5mm}\right], 
\end{align*}
For the second equality in Lemma~\ref{lem:var_recur},  
the detailed proof is rather technical, and we defer it to the Appendix (i.e., Section~\ref{sec:appendix}) for readability. At a high level, we prove it by repeatedly applying the law of total variance. Details of the law of total variance are given as below. 

\begin{fact}[Law of Total Variance~\cite{weiss2005TotalVarianceLaw}]\label{fact:totalvar}
For two random variables $X$ and $Y$, the law of total variance states: 
\begin{align*}
\Var\left[Y\right]=\E\left[\Var \left[Y\mid X\right]\right]+\Var\left[\E\left[Y\mid X\right]\right]
\end{align*}
holds if the two variables $X$ and $Y$ are on the same probability space and the variance of $Y$ is finite. 
\end{fact}

Furthermore, we plug the variance bound derived in Lemma~\ref{lem:conditional_variance} into Lemma~\ref{lem:var_recur}, which follows Lemma~\ref{lem:var_partial}.  

\begin{lemma}\label{lem:var_partial}
For all $\ell\in [0, L]$, the residue vectors $\er^{(\ell)}$ obtained by Algorithm~\ref{alg:VBES} in the $\ell$-th iterations satisfy: 
\begin{equation*}
\begin{aligned}
\sum_{\ell=1}^{L-1}\hspace{-0.5mm}\E\hspace{-0.5mm}\left[\hspace{-0.5mm}\Var\hspace{-0.5mm}\left[\hspace{-0.5mm}\left.\sum_{v\in V}\hspace{-0.5mm}\left(\sum_{s\in V}\hspace{-0.8mm}\frac{d_t}{d_s}\hspace{-0.5mm}\cdot \hspace{-0.5mm}\hspace{-0.5mm}\sum_{i=0}^{L-\ell}\hspace{-1.5mm}\frac{\vpi^{(i)}_v\hspace{-0.5mm}(s)}{\alpha}\hspace{-0.5mm}\right)\hspace{-0.5mm}\cdot \hspace{-0.5mm}\er^{(\ell)}_t\hspace{-0.5mm}(v)~\right|\er^{(\ell-1)}_t\hspace{-0.5mm}\right]\right]\hspace{-1mm}\le \hspace{-0.5mm}\frac{ L\th d_t \hspace{-0.5mm}\cdot \hspace{-0.5mm} n\vpi(t)}{\alpha^2}. 
\end{aligned}    
\end{equation*}
\end{lemma}

\begin{proof}
According to Algorithm~\ref{alg:VBES}, given the residue vector $\r_t^{(\ell)}$, the residue's increment $X^{(\ell+1)}(u,v)$ of node $u$ (formally defined in the proof of Lemma~\ref{lem:unbiasedness_er}) is independent with that of other nodes $w\in V$. Therefore, the variance expression given in Lemma~\ref{lem:var_recur} can be rewritten as: 
\begin{equation}\label{eqn:exp_sum}
\begin{aligned}
&\sum_{\ell=1}^{L-1}\E\left[\Var\left[\left.\sum_{v\in V}\left(\sum_{s\in V}\frac{d_t}{d_s}\cdot \sum_{i=0}^{L-\ell}\frac{\vpi^{(i)}_v(s)}{\alpha}\right)\cdot \er^{(\ell)}_t(v)~\right|\er^{(\ell-1)}_t\right]\right]\\
&=\sum_{\ell=1}^{L-1}\E\left[\sum_{v\in V}\Var\left[\left.\left(\sum_{s\in V}\frac{d_t}{d_s}\cdot \sum_{i=0}^{L-\ell}\frac{\vpi^{(i)}_v(s)}{\alpha}\right)\cdot \er^{(\ell)}_t(v)~\right|\er^{(\ell-1)}_t\right]\right]. 
\end{aligned}    
\end{equation}
Note that $\sum_{s\in V}\frac{d_t}{d_s}\cdot \sum_{i=1}^{L-\ell}\frac{\vpi^{(i)}_v(s)}{\alpha}$ is a deterministic probability mass rather than a random variable. Thus, we have: 
\begin{equation*}
\begin{aligned}
&\Var\left[\left.\left(\sum_{s\in V}\frac{d_t}{d_s}\cdot \sum_{i=0}^{L-\ell}\frac{\vpi^{(i)}_v(s)}{\alpha}\right)\cdot \er^{(\ell)}_t(v)~\right|\er^{(\ell-1)}_t\right]\\
&=\left(\sum_{s\in V}\frac{d_t}{d_s}\cdot \sum_{i=0}^{L-\ell}\frac{\vpi^{(i)}_v(s)}{\alpha}\right)^2 \cdot \Var\left[\left. \er^{(\ell)}_t(v)~\right|\er^{(\ell-1)}_t\right]
\end{aligned}    
\end{equation*}
In particular, the value of $\sum_{s\in V}\frac{d_t}{d_s}\cdot \sum_{i=1}^{L-\ell}\frac{\vpi^{(i)}_v(s)}{\alpha}$ can be upper bounded as: 
\begin{align*}
\sum_{s\in V}\frac{d_t}{d_s}\cdot \hspace{-2mm}\sum_{i=0}^{L-\ell}\hspace{-1mm}\frac{\vpi^{(i)}_v(s)}{\alpha}\le \frac{d_t}{\alpha}\cdot \hspace{-1mm}\sum_{s\in V}\sum_{i=0}^{L-\ell}\hspace{-1mm}\vpi^{(i)}_v(s) \le \frac{d_t}{\alpha}\cdot \hspace{-1mm} \sum_{s\in V}\vpi_v(s)=\frac{d_t}{\alpha}. 
\end{align*}
Plugging into Equation~\eqref{eqn:exp_sum}, we can further derive: 
\begin{equation*}
\begin{aligned}
&\sum_{\ell=0}^{L-1}\E\left[\sum_{v\in V}\Var\left[\left.\left(\sum_{s\in V}\frac{d_t}{d_s}\cdot \hspace{-1mm}\sum_{i=0}^{L-\ell}\frac{\vpi^{(i)}_v(s)}{\alpha}\right)\cdot \er^{(\ell)}_t(v)~\right|\er^{(\ell-1)}_t\right]\right]\\
&\le \hspace{-1mm}\frac{d_t}{\alpha}\cdot \hspace{-1mm}\sum_{\ell=0}^{L-1}\E\left[\sum_{v\in V}\left(\sum_{s\in V}\frac{d_t}{d_s}\cdot \hspace{-1mm}\sum_{i=0}^{L-\ell}\frac{\vpi^{(i)}_v(s)}{\alpha}\right)\cdot \Var\left[\left. \er^{(\ell)}_t(v)~\right|\er^{(\ell-1)}_t\right]\right]. 
\end{aligned}    
\end{equation*}
Recall that in Lemma~\ref{lem:conditional_variance}, we have already bounded the conditional variance: 
$\Var\left[\left. \er^{(\ell)}_t(v)~\right|\er^{(\ell-1)}_t\right]\hspace{-1mm}\le \hspace{-0.5mm}\sum_{u\in N(v)}\hspace{-0.5mm}\frac{\th \cdot (1-\alpha)\cdot \er^{(\ell-1)}_t(u)}{d_u}$, 
Moreover, by Lemma~\ref{lem:unbiasedness_er} and Equation~\eqref{eqn:PPR_recur}, we have: 
\begin{equation*}
\begin{aligned}
&\E\left[\Var\left[\left. \er^{(\ell)}_t(v)~\right|\er^{(\ell-1)}_t\right]\right]\le \sum_{u\in N(v)}\frac{\th \cdot (1-\alpha)}{d_u}\cdot \E\left[\er^{(\ell-1)}_t(u)\right]\\
&=\sum_{u\in N(v)} \frac{\th \cdot (1-\alpha)}{\alpha \cdot d_u}\cdot \vpi^{(\ell-1)}_t(u)=\frac{\th}{\alpha}\cdot \vpi^{(\ell)}_t(v). 
\end{aligned}    
\end{equation*}
Therefore, it follows: 
\begin{equation*}
\begin{aligned}
&\sum_{\ell=1}^{L-1}\E\left[\sum_{v\in V}\Var\left[\left.\left(\sum_{s\in V}\frac{d_t}{d_s}\cdot \hspace{-1mm}\sum_{i=0}^{L-\ell}\frac{\vpi^{(i)}_v(s)}{\alpha}\right)\cdot \er^{(\ell)}_t(v)~\right|~\er^{(\ell-1)}_t\right]\right]\\
& \le \frac{d_t \cdot \th}{\alpha^2}\cdot \sum_{\ell=1}^{L-1} \sum_{s\in V}\frac{d_t}{d_s}\cdot \sum_{i=0}^{L-\ell}\sum_{v\in V} \frac{\vpi^{(i)}_v(s)}{\alpha} \cdot \vpi^{(\ell)}_t(v). 
\end{aligned}    
\end{equation*}
Note that $\sum_{v\in V}\frac{1}{\alpha}\cdot \vpi^{(i)}_v(s)\cdot \vpi^{(\ell)}_t(v)=\vpi^{(\ell+i)}_t(s)$. Moreover, 
\begin{align*}
\sum_{s\in V}\frac{d_t}{d_s}\cdot \sum_{i=0}^{L-\ell}\vpi^{(\ell+i)}_t(s)\le \sum_{s\in V}\frac{d_t}{d_s}\cdot \vpi_t(s)=\sum_{s\in V} \vpi_s(t)=n\vpi(t). 
\end{align*}
As a consequence, we can further derive: 
\begin{equation*}
\begin{aligned}
&\sum_{\ell=1}^{L-1}\E\left[\Var\left[\left.\sum_{v\in V}\left(\sum_{s\in V}\frac{d_t}{d_s}\cdot \hspace{-1mm}\sum_{i=0}^{L-\ell}\frac{\vpi^{(i)}_v(s)}{\alpha}\right)\cdot \er^{(\ell)}_t(v)~\right|\er^{(\ell-1)}_t\right]\right]\\
&\le \frac{d_t \cdot \th}{\alpha^2} \cdot \sum_{\ell=1}^{L-1} n\pi(t) \le \frac{1}{\alpha^2} \cdot L \th d_t \cdot n \vpi(t), 
\end{aligned}   
\end{equation*}
which follows the lemma. 
\end{proof}

Finally, by putting Lemma~\ref{lem:var_recur} and ~\ref{lem:var_partial} together, we can conclude that  
$\Var\left[\epi(t)\right]\le \frac{L\cdot \th \cdot d_t}{n^2}\cdot n\epi(t)$, 
following Theorem~\ref{thm:variance}. 

\subsection{Time Cost}\label{subsec:totalcost}
In the following, we analyze the expected time cost of the \setpush algorithm. Moreover, Theorem~\ref{thm:finalcost_analysis} provides the theoretical guarantees of the \setpush algorithm for achieving a $(\rela, \pf)$-approximation of the single-node PageRank. 

\begin{lemma}\label{lem:cost_theta}
The expected time cost of Algorithm~\ref{alg:VBES} can be bounded by ${\rev \frac{1}{\alpha \theta}=}~O\left(\frac{1}{\th}\right)$. 
\end{lemma} 

\begin{proof}
Let $Cost^{(\ell+1)}(u,v)$ denote the time cost of increasing $\r^{(\ell+1)}_t(v)$ during the update process conducted at node $u$ with nonzero $\r^{(\ell)}_t(u)$. 
According to Algorithm~\ref{alg:VBES}, $Cost^{(\ell+1)}(u,v)=1$ holds deterministically if $\frac{(1-\alpha)}{d_u}\cdot \er^{(\ell)}_t(u)\ge \th$. On the other hand, if $\frac{(1-\alpha)}{d_u}\cdot \er^{(\ell)}_t(u)< \th$, $Cost^{(\ell+1)}(u,v)=1$ (i.e., pushing probability mass from node $u$ to $v$) holds with probability $\frac{(1-\alpha)}{d_u\cdot \th}\cdot \er^{(\ell)}_t(u)$, or $Cost^{(\ell+1)}(u,v)=0$ holds with probability $1-\frac{(1-\alpha)}{d_u\cdot \th}\cdot \er^{(\ell)}_t(u)$. Thus, given the $\ell$-hop residue vector $\r^{(\ell)}_t$, the expectation of $Cost^{(\ell+1)}(u,v)$ can be bounded as: 
\begin{align*}
\E\left[Cost^{(\ell+1)}(u,v)~\big|~ \er^{(\ell)}_t\right]\le 1\cdot \frac{(1-\alpha)}{d_u\cdot \th}\cdot \er^{(\ell)}_t(u). 
\end{align*}
Furthermore, let $Cost^{(\ell+1)}$ denote the time cost of updating the $(\ell+1)$-hop residue vector $\r^{(\ell+1)}_t$ based on the $\ell$-hop residue vector $\r^{(\ell)}_t$. Then we have $Cost^{(\ell+1)}\hspace{-0.5mm}=\hspace{-0.5mm}\sum_{(u,v)\in E} Cost^{(\ell+1)}(u,v)$. It follows: 
\begin{equation*}
\begin{aligned}
\E\hspace{-0.5mm}\left[Cost^{(\ell+1)}~\big|~ \er^{(\ell)}_t\hspace{-0.5mm}\right]\hspace{-0.5mm}=\hspace{-3mm}\sum_{(u,v)\in E}\hspace{-2.5mm}\E\left[Cost^{(\ell+1)}(u,v)~\big|~ \er^{(\ell)}_t\hspace{-0.5mm}\right]\hspace{-0.5mm}=\hspace{-3mm}\sum_{(u,v)\in E}\hspace{-2mm}\frac{(1\hspace{-0.5mm}-\hspace{-0.5mm}\alpha)}{d_u \hspace{-0.5mm}\cdot \hspace{-0.5mm}\th}\hspace{-0.5mm}\cdot \er^{(\ell)}_t\hspace{-0.5mm}(u). 
\end{aligned}    
\end{equation*}
By the property of expectation, we further have:  
\begin{equation*}
\begin{aligned}
&\E\left[Cost^{(\ell+1)}\right]=\E\left[\E\left[Cost^{(\ell+1)}~\big|~ \er^{(\ell)}_t\hspace{-0.5mm}\right]\right]=\hspace{-2mm}\sum_{(u,v)\in E}\hspace{-1mm}\frac{(1\hspace{-0.5mm}-\hspace{-0.5mm}\alpha)}{d_u \hspace{-0.5mm}\cdot \hspace{-0.5mm}\th}\hspace{-0.5mm}\cdot \E\left[\er^{(\ell)}_t\hspace{-0.5mm}(u)\right]\\
&=\frac{1}{\alpha\th}\cdot \sum_{v\in V}\sum_{u\in N(v)}\hspace{-2mm}\frac{(1\hspace{-0.5mm}-\hspace{-0.5mm}\alpha)}{d_u}\hspace{-0.5mm}\cdot \vpi^{(\ell)}_t\hspace{-0.5mm}(u)=\frac{1}{\alpha\th}\cdot \sum_{v\in V}\vpi^{(\ell+1)}_t(v),  
\end{aligned}
\end{equation*}
where we apply Lemma~\ref{lem:unbiasedness_er} in the third equality given above. We also apply Equation~\eqref{eqn:PPR_recur} in the last equality as shown above. Furthermore, let $Cost=\sum_{\ell=0}^{L-1}Cost^{(\ell+1)}$ denote the total time cost of Algorithm~\ref{alg:VBES}. Thus, we can derive: 
\begin{align*}
\E\left[Cost\right]=\hspace{-0.5mm}\sum_{\ell=0}^{L-1}\E\left[Cost^{(\ell+1)}\right]\hspace{-0.5mm}=\hspace{-0.5mm}\frac{1}{\alpha\th}\cdot \hspace{-1mm}\sum_{v\in V}\sum_{\ell=0}^{L-1}\vpi^{(\ell+1)}_t(v)\hspace{-0.5mm}\le \hspace{-0.5mm}\frac{1}{\th}=O\left(\frac{1}{\th}\right),  
\end{align*}
by applying $\sum_{\ell=0}^{L-1}\vpi^{(\ell+1)}_t(v)\le \vpi_t(v)$, and $\sum_{v\in V}\vpi^{(\ell+1)}_t(v)=\alpha$. Therefore, the lemma follows. 
\end{proof}

In the end, we employ the bound of variance $\Var\left[\epi(t)\right]$ derived in Theorem~\ref{thm:variance} to the Chebyshev's Inequality given in Fact~\ref{fact:chebyshev}, to derive an appropriate setting of the threshold $\theta$. 

\begin{figure*}[t]
\begin{minipage}[t]{1\textwidth}
\centering
\begin{tabular}{cccc}
\hspace{-3mm} \includegraphics[width=38mm]{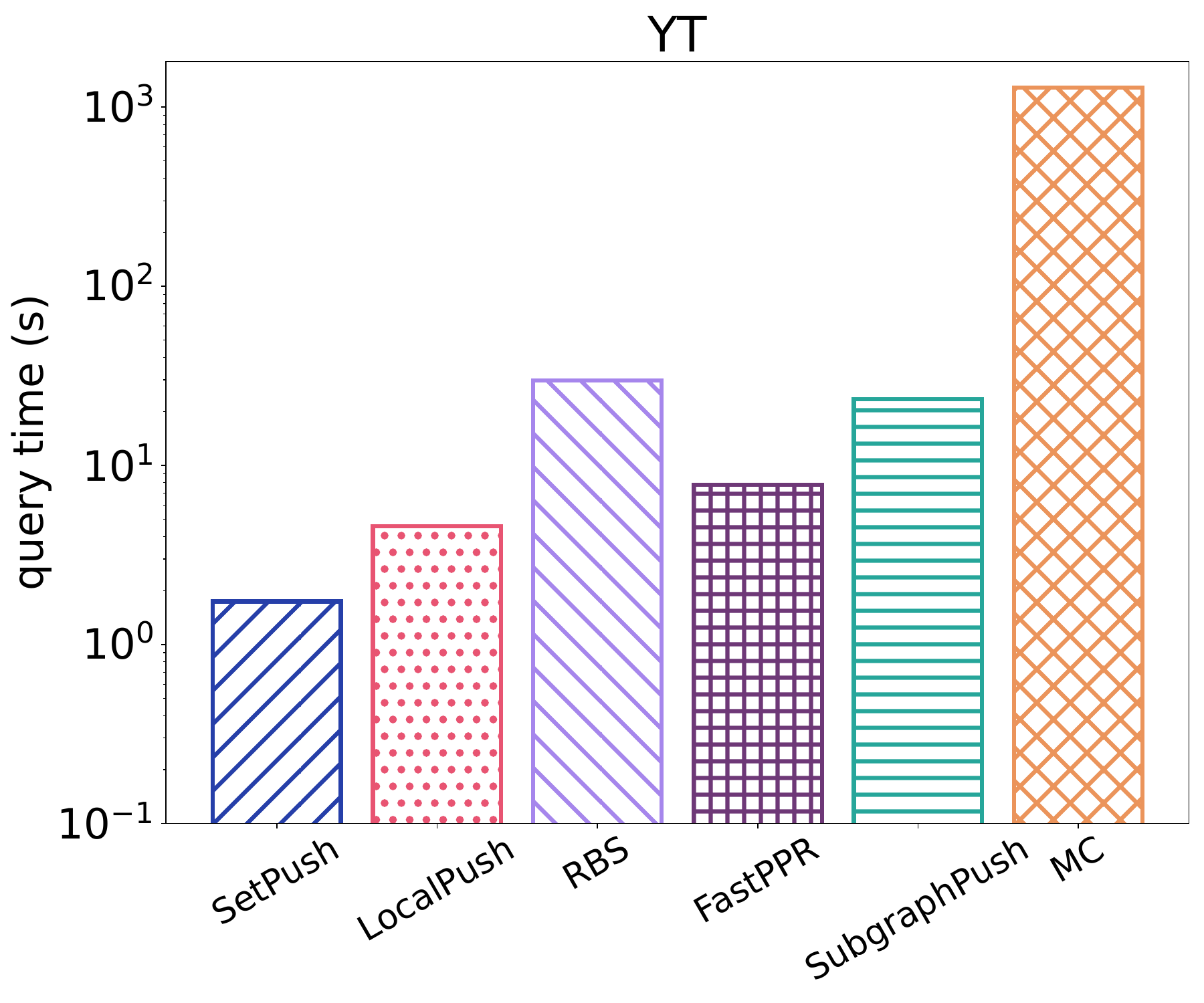} &
\hspace{+1mm} \includegraphics[width=38mm]{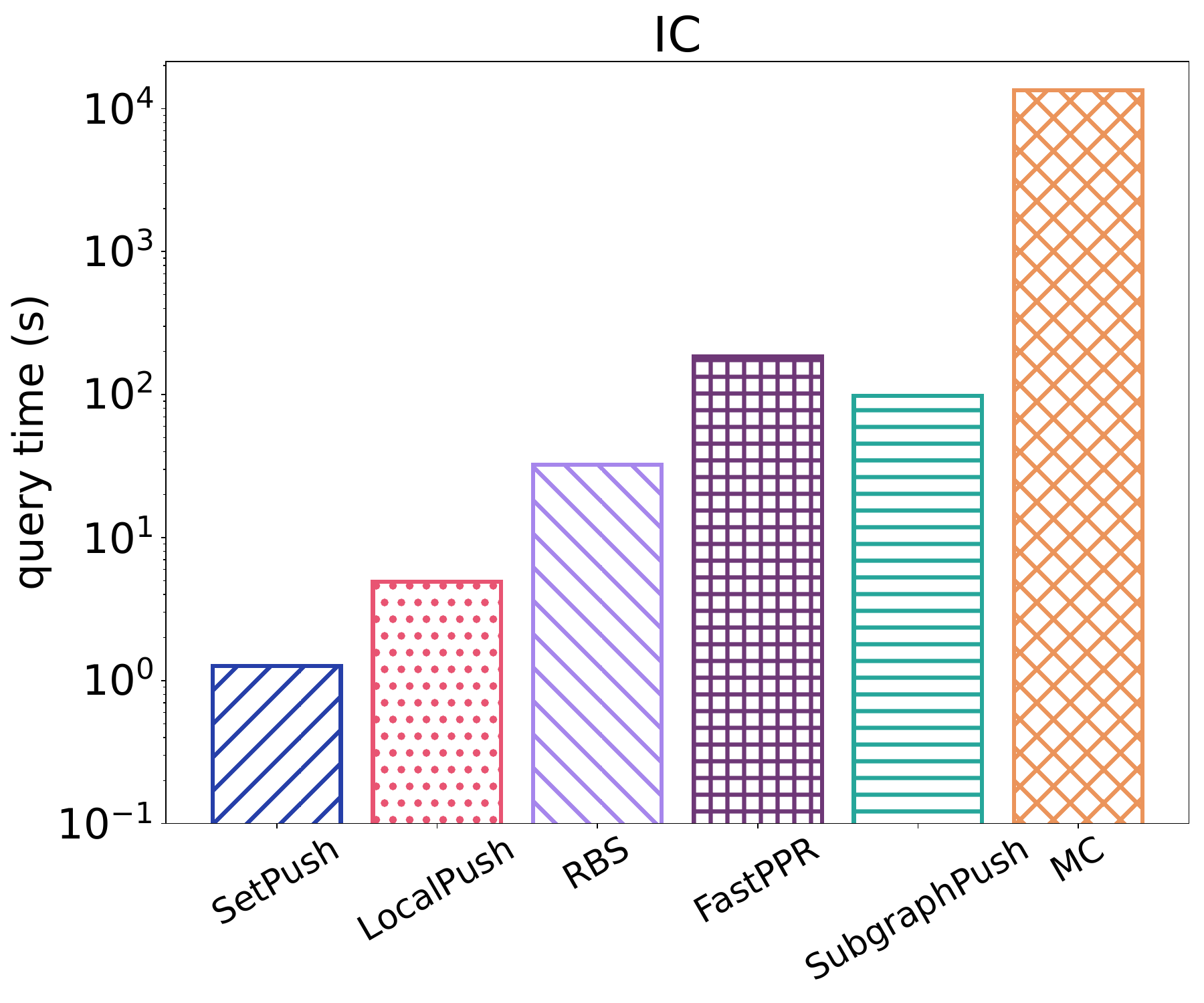} &
\hspace{+1mm} \includegraphics[width=38mm]{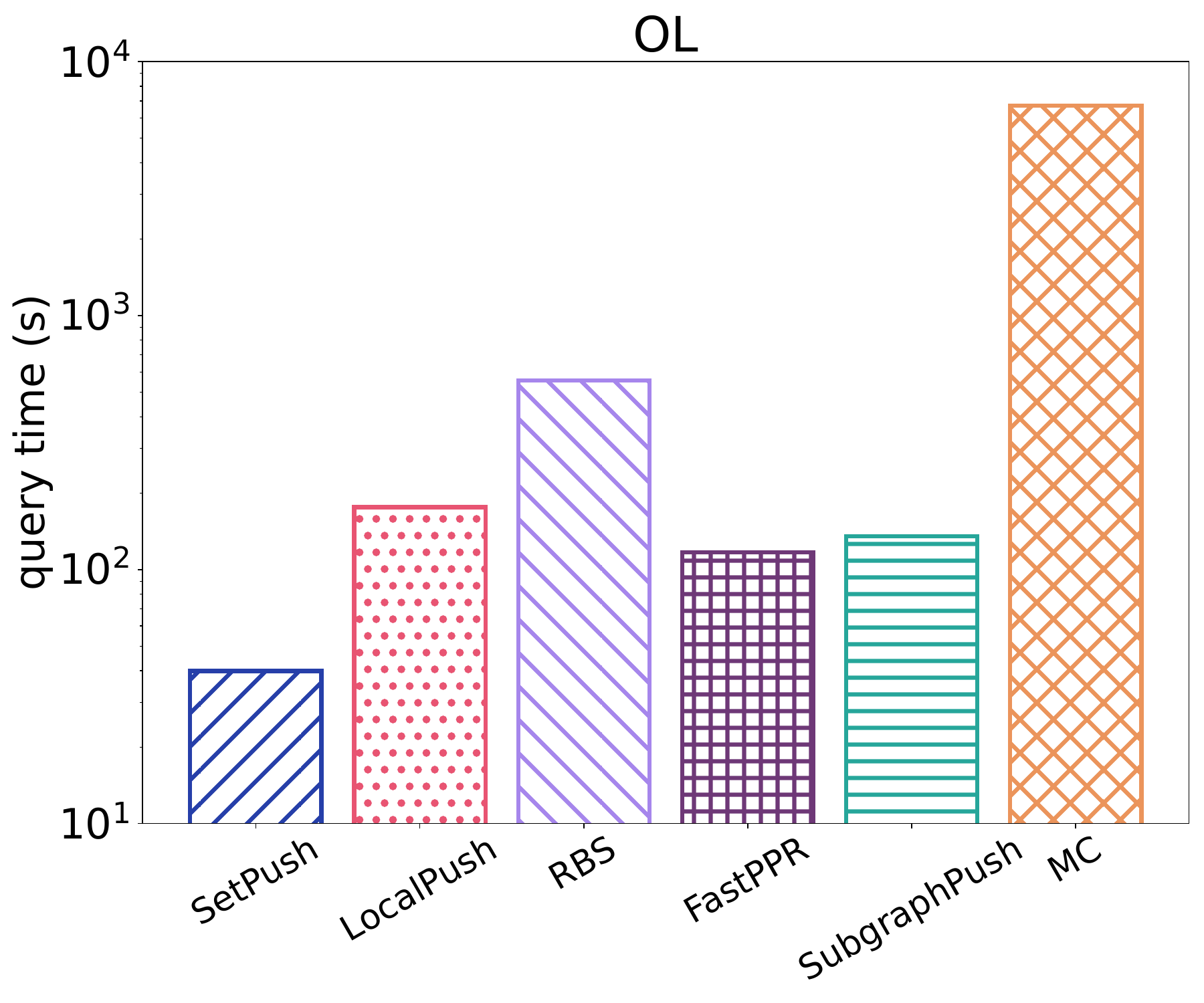} &
\hspace{+1mm} \includegraphics[width=38mm]{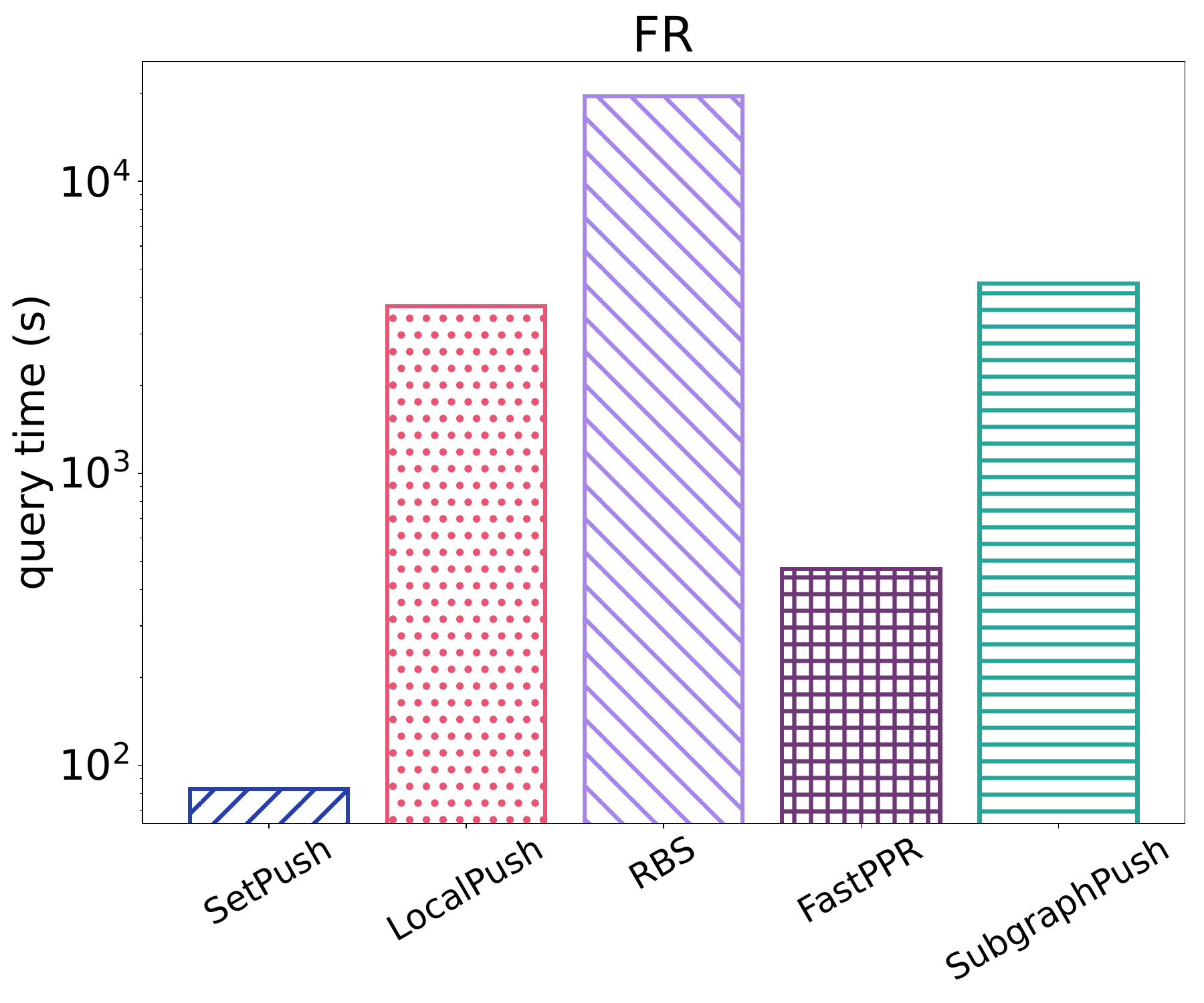} \\
\end{tabular}
\vspace{-6mm}
\caption{\rev Query time (seconds) of each algorithm with uniformly selected query node, $\rela=0.1$}
\label{fig:query_uniform_0.1}
\end{minipage}

\begin{minipage}[t]{1\textwidth}
\centering
\begin{tabular}{cccc}
\hspace{-3mm} \includegraphics[width=38mm]{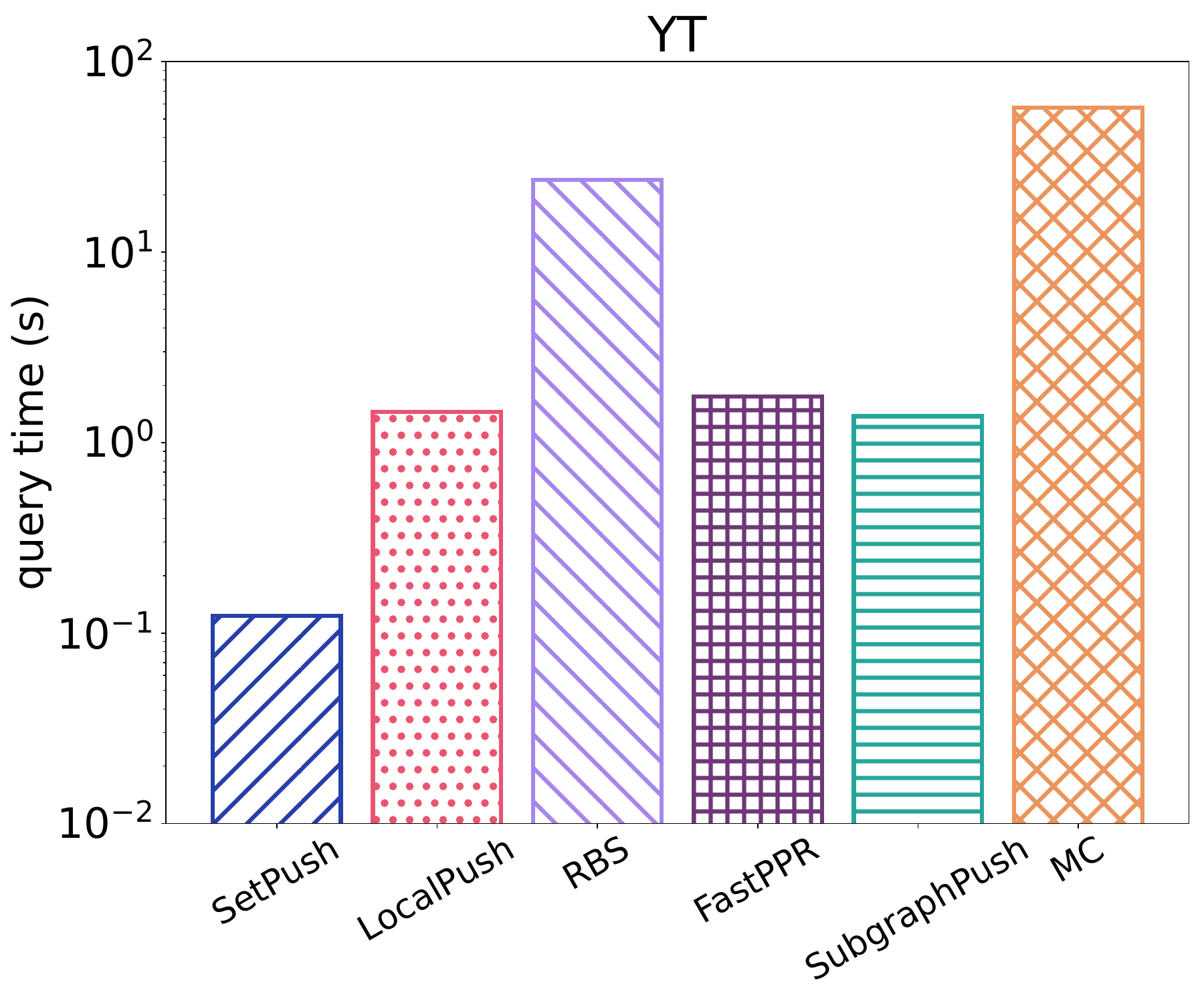} &
\hspace{+1mm} \includegraphics[width=38mm]{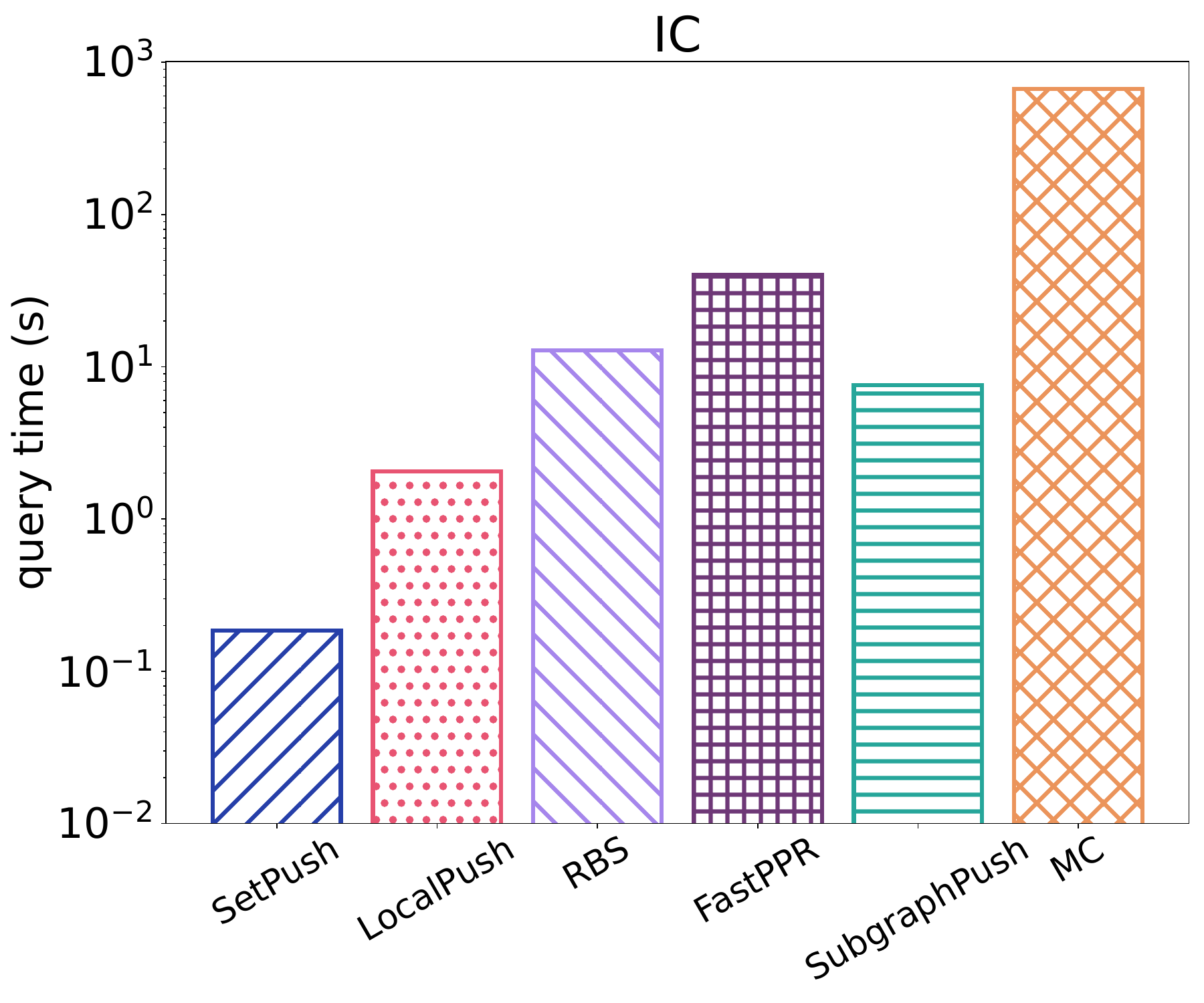} &
\hspace{+1mm} \includegraphics[width=38mm]{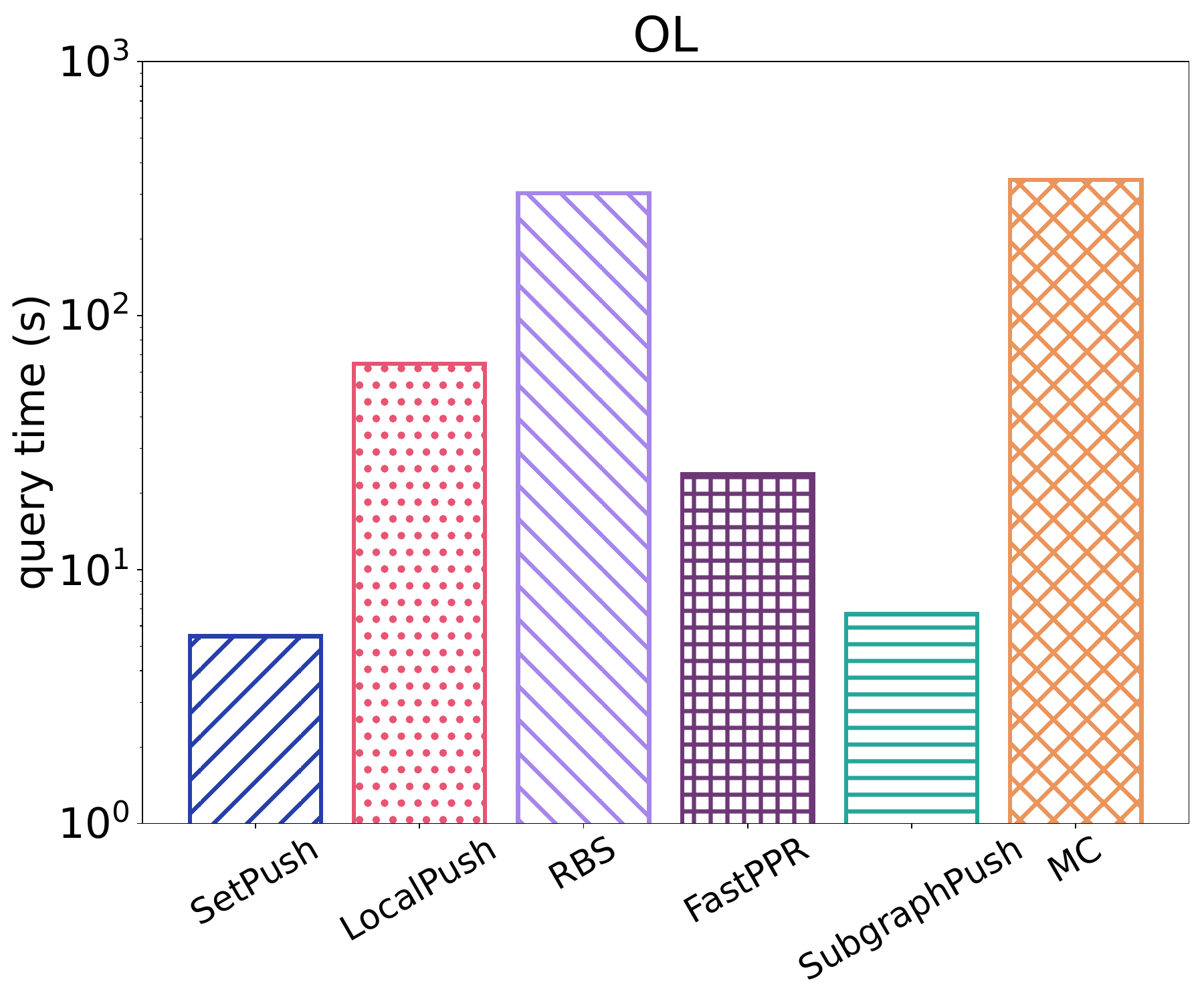} &
\hspace{+1mm} \includegraphics[width=38mm]{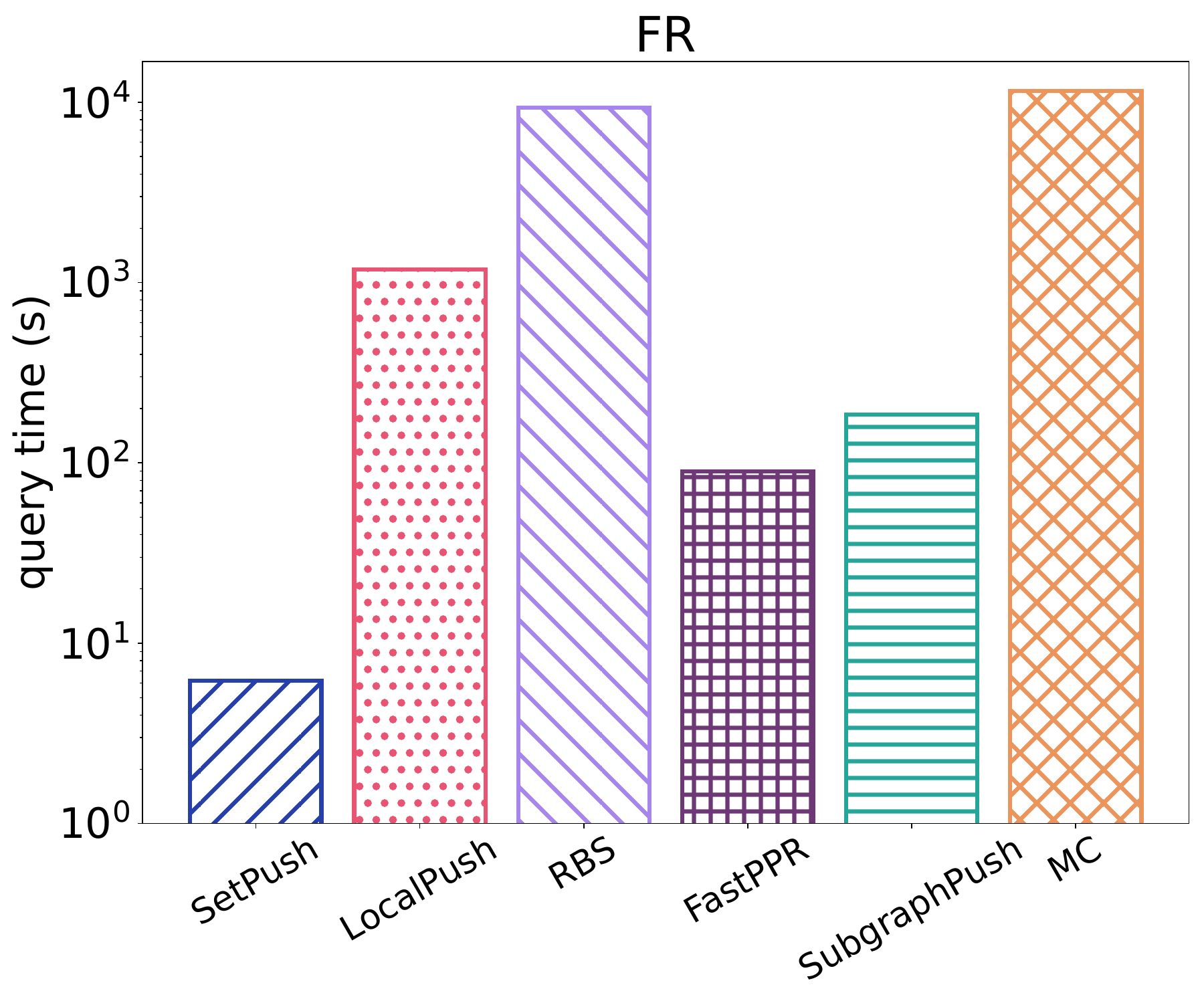} \\
\end{tabular}
\vspace{-6mm}
\caption{\rev Query time (seconds) of each algorithm with uniformly selected query node, $\rela=0.5$}
\label{fig:query_uniform_0.5}
\vspace{-2mm}
\end{minipage}
\end{figure*}

\begin{figure*}[t]
\begin{minipage}[t]{1\textwidth}
\centering
\vspace{-2mm}
\begin{tabular}{cccc}
\hspace{-3mm} \includegraphics[width=38mm]{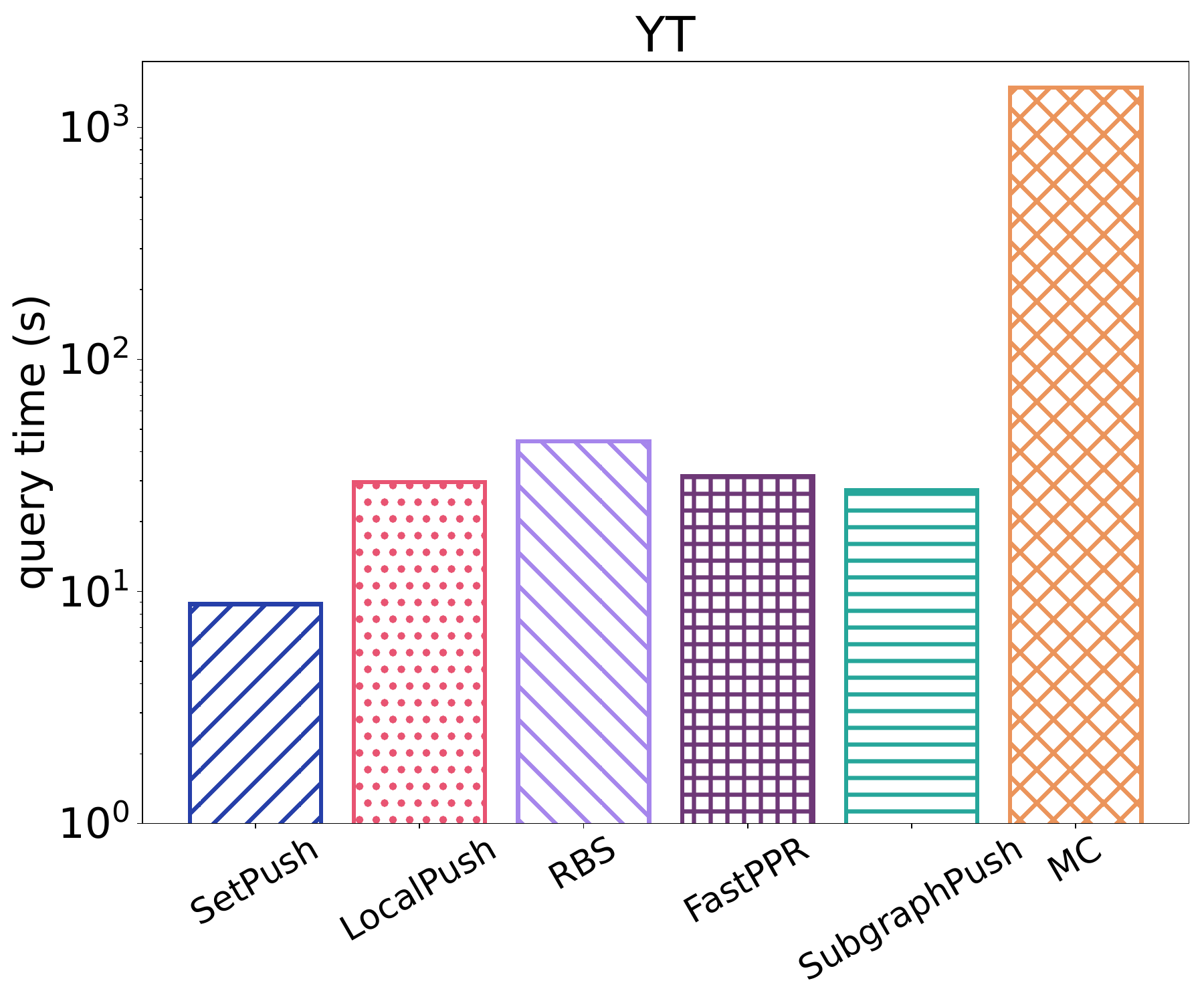} &
\hspace{+1mm} \includegraphics[width=38mm]{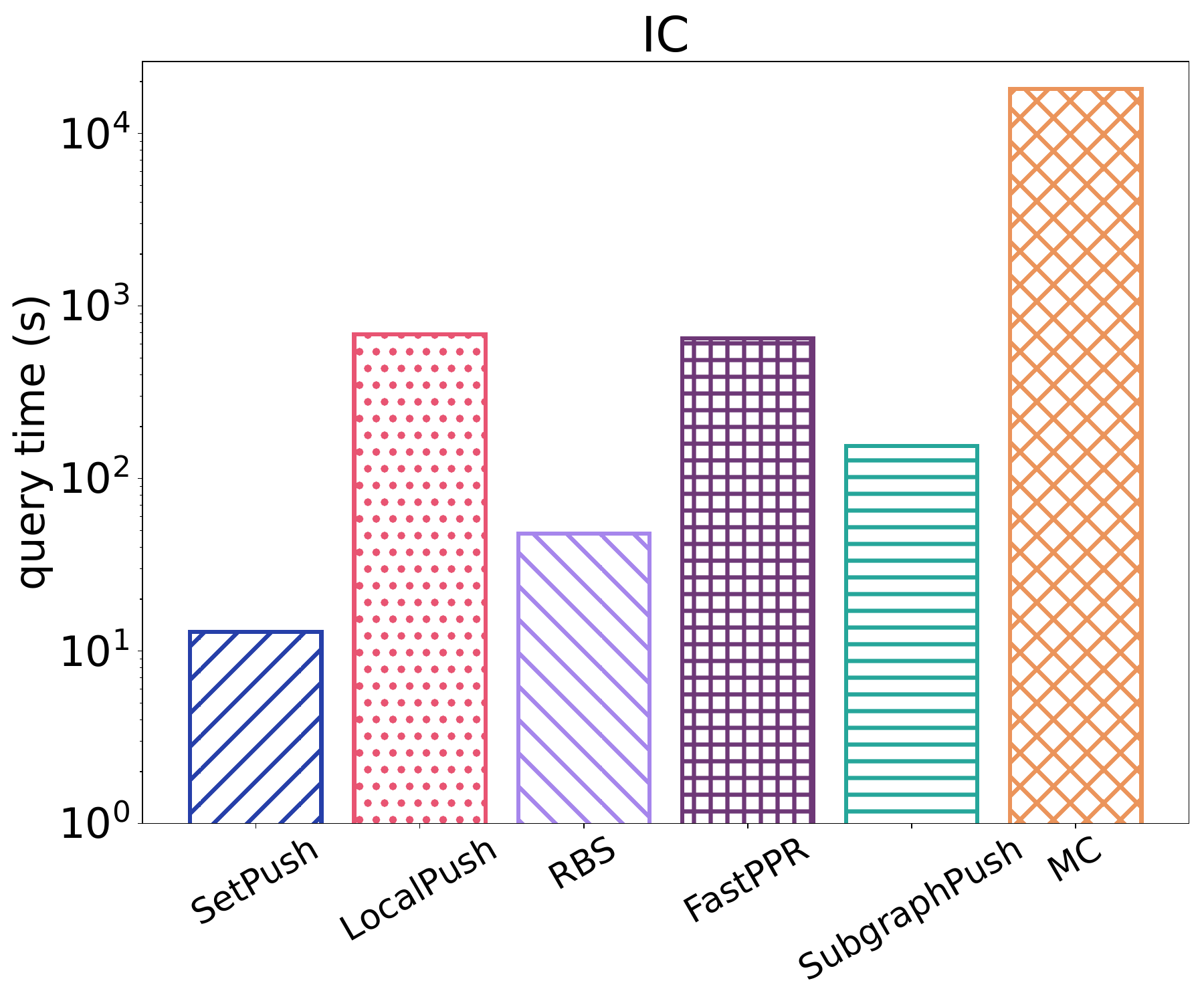} &
\hspace{+1mm} \includegraphics[width=38mm]{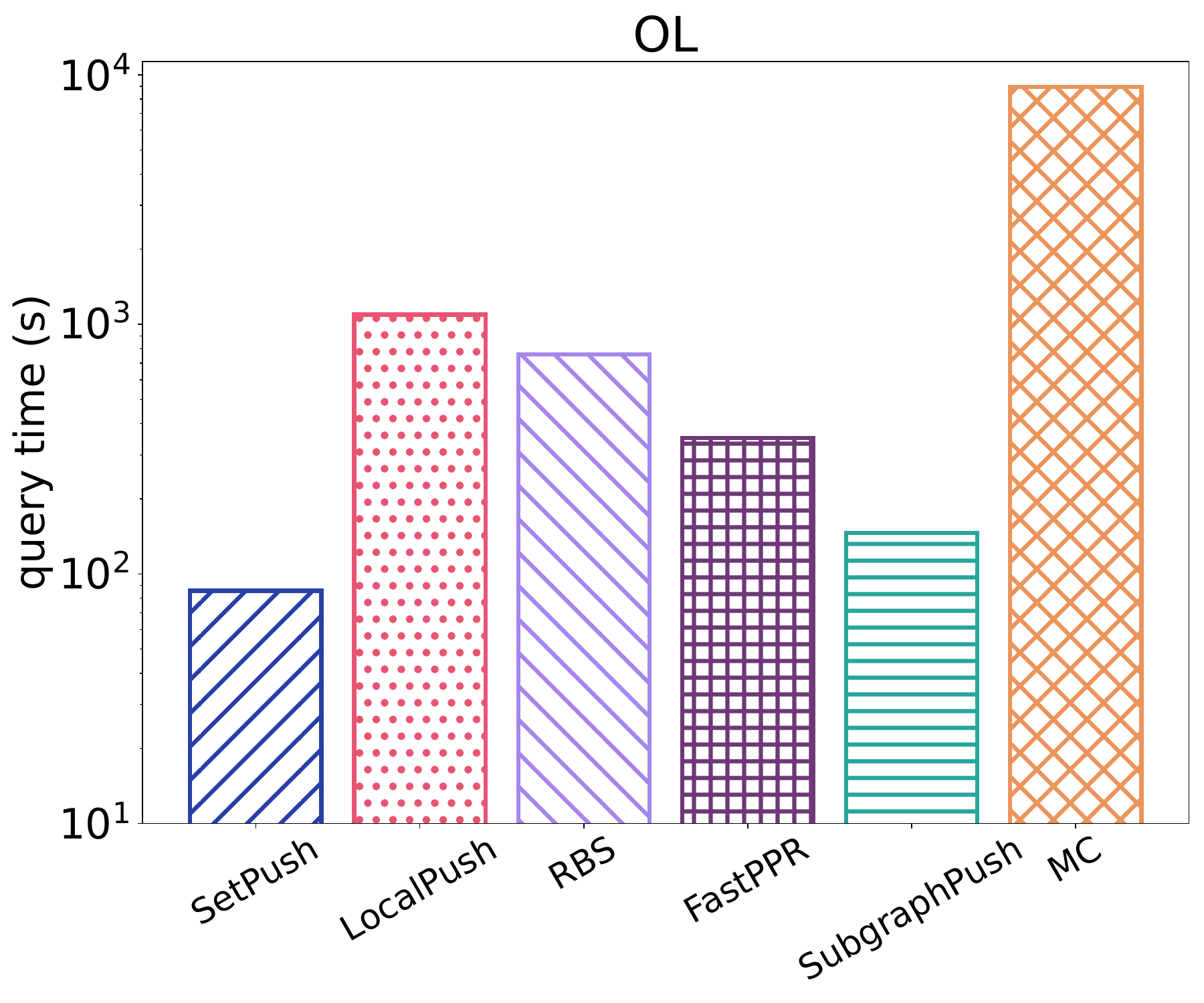} &
\hspace{+1mm} \includegraphics[width=38mm]{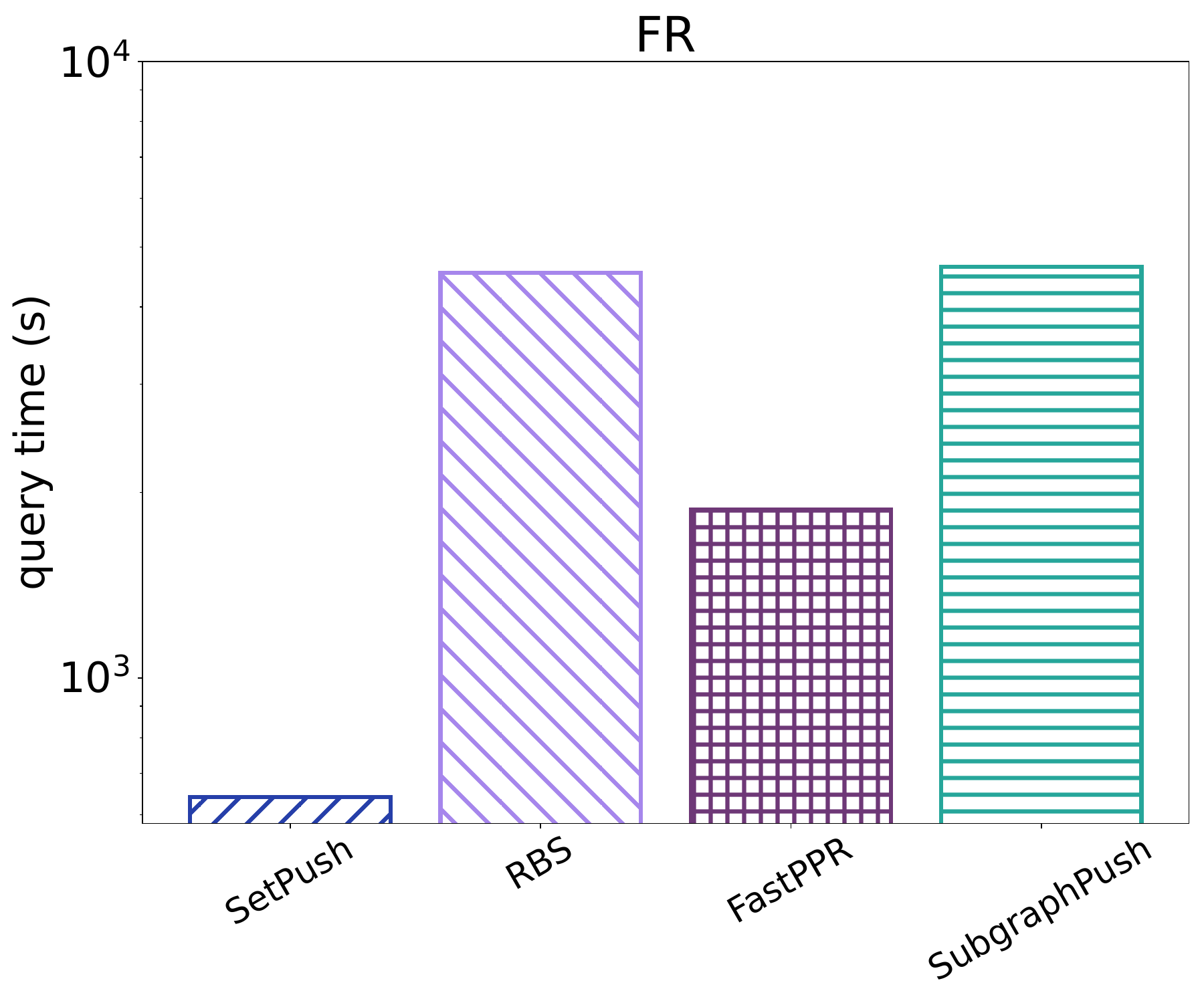} \\
\end{tabular}
\vspace{-6mm}
\caption{\rev Query time (seconds) of each algorithm with degree distributed query nodes, $\rela=0.1$}
\label{fig:query_degree_0.1}
\end{minipage}

\begin{minipage}[t]{1\textwidth}
\centering
\begin{tabular}{cccc}
\hspace{-3mm} \includegraphics[width=38mm]{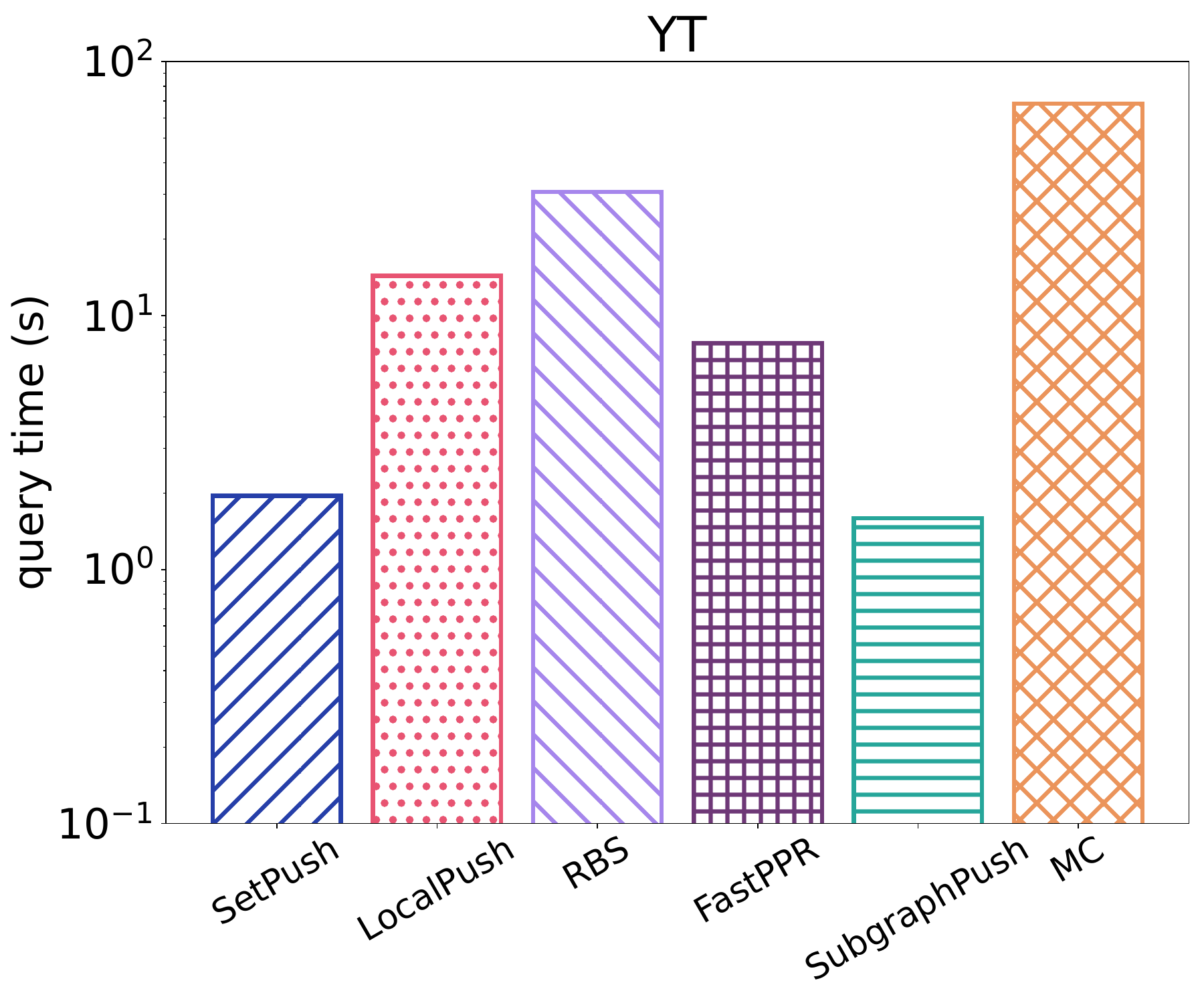} &
\hspace{+1mm} \includegraphics[width=38mm]{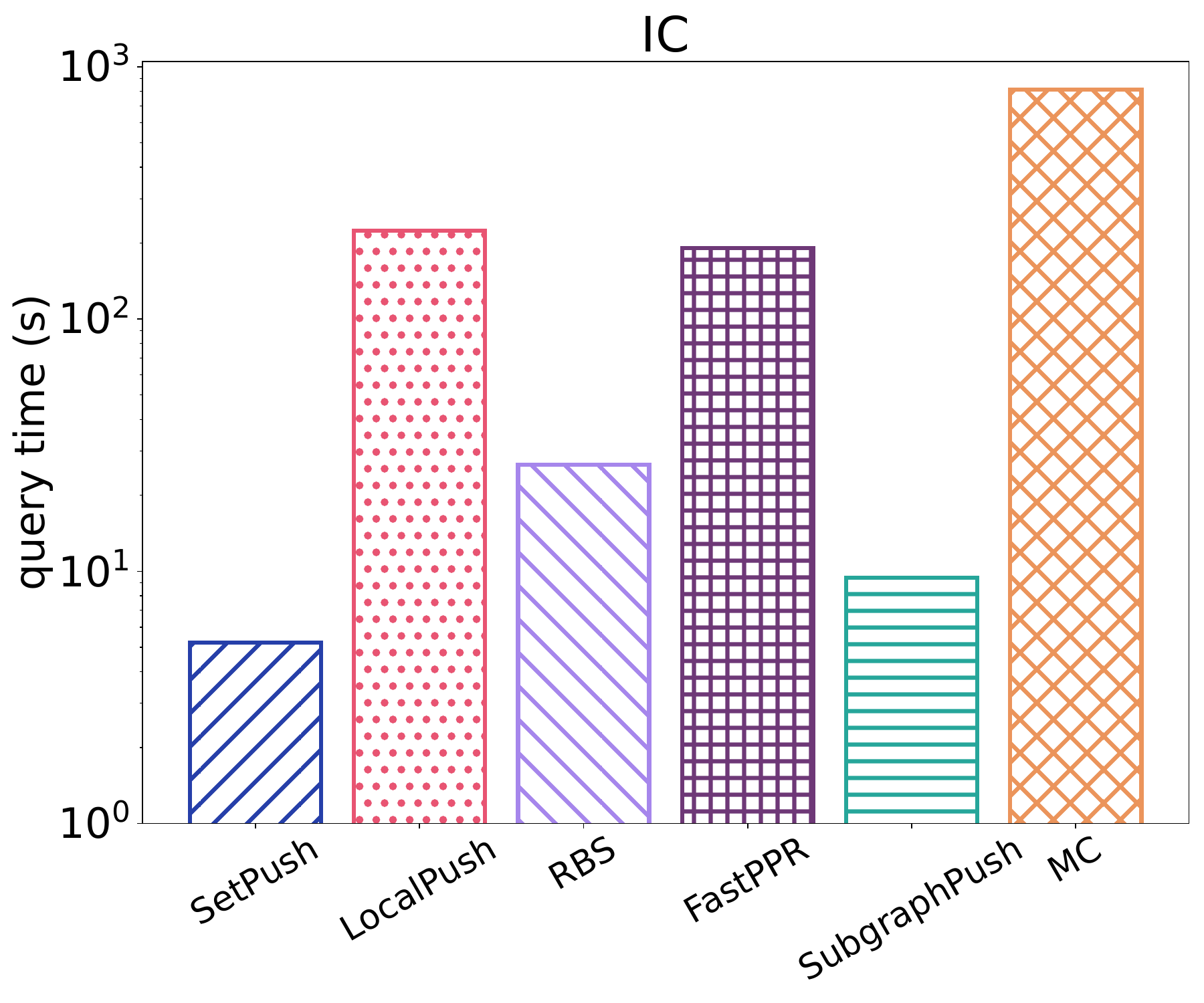} &
\hspace{+1mm} \includegraphics[width=38mm]{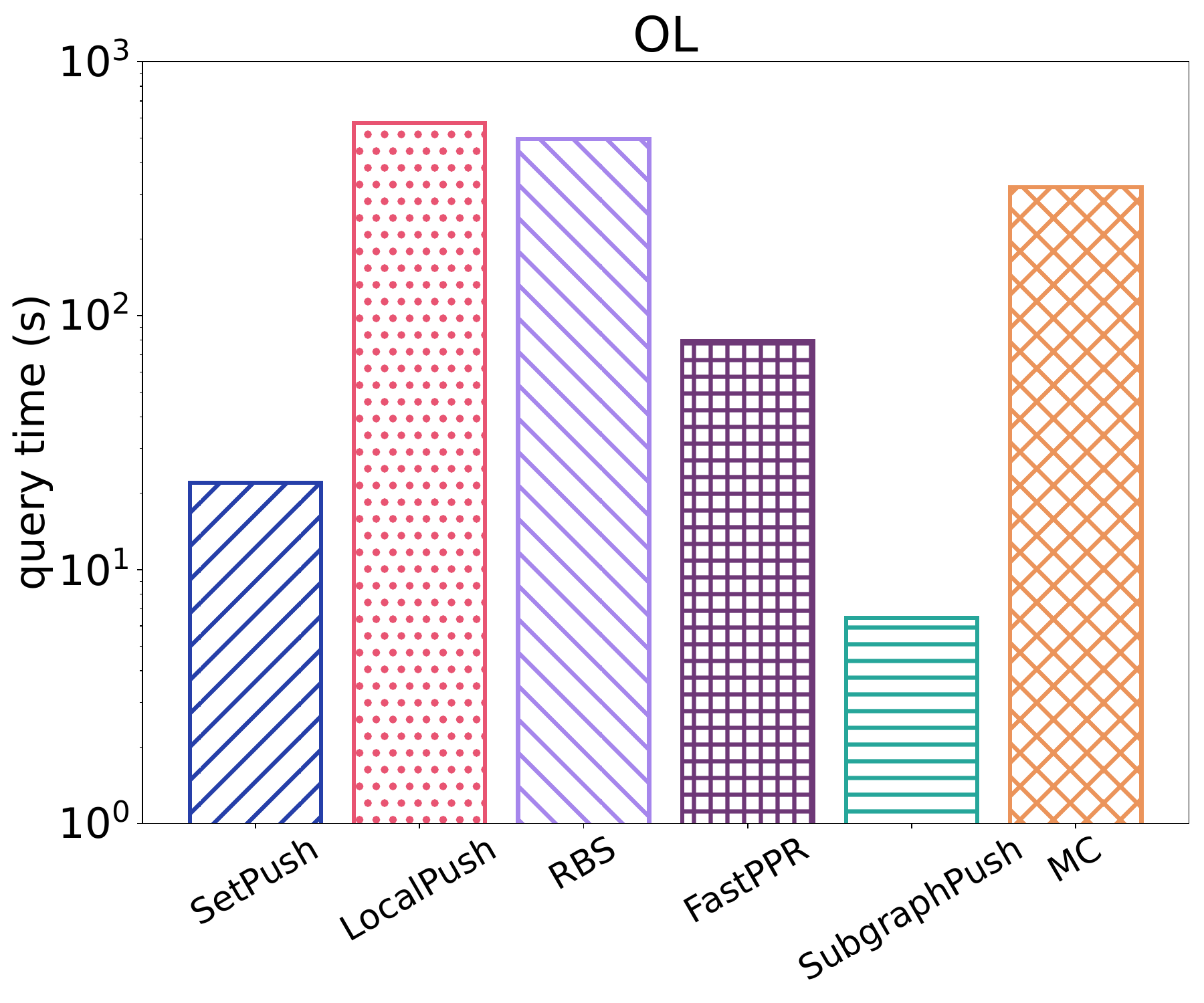} &
\hspace{+1mm} \includegraphics[width=38mm]{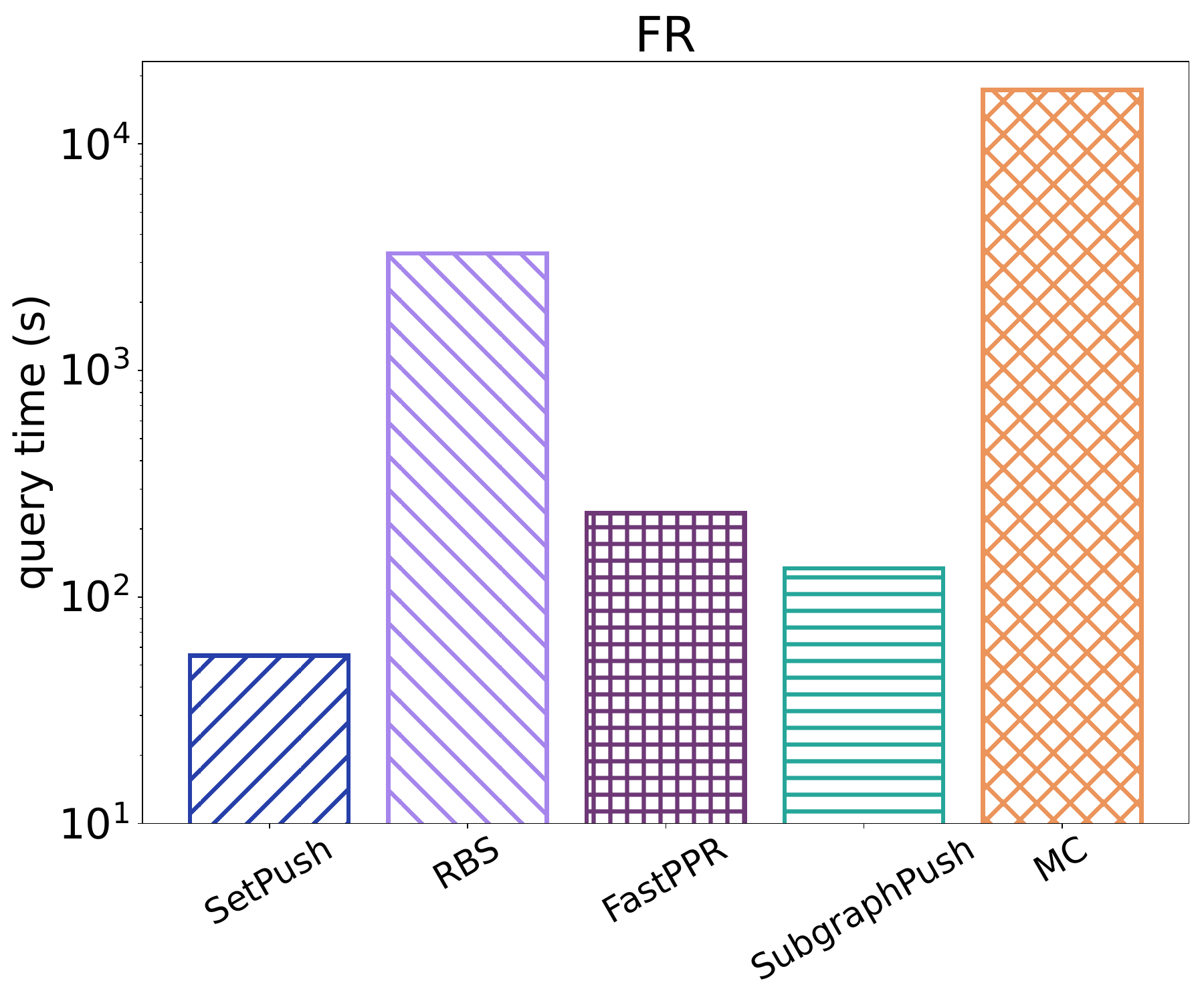} \\
\end{tabular}
\vspace{-6mm}
\caption{\rev Query time (seconds) of each algorithm with degree distributed query nodes, $\rela=0.5$}
\label{fig:query_degree_0.5}
\vspace{-2mm}
\end{minipage}
\end{figure*}

\begin{theorem}\label{thm:finalcost_analysis}
By setting $\theta=\max\left\{\frac{\alpha c^2}{12L\cdot d_t},\frac{\alpha c^2}{12L}\hspace{-0.5mm}\cdot \hspace{-0.5mm}\sqrt{\frac{2(1-\alpha)}{m}}\right\}$, Algorithm~\ref{alg:VBES} returns a $(c,p_f)$-approximation $\epi(t)$ of $\vpi(t)$, such that $|\vpi(t)-\epi(t)|\le c\cdot \vpi(t)$ holds with constant probability. The expected time cost of Algorithm~\ref{alg:VBES} is bounded by 
\begin{align*}
{\rev 
\frac{12 \cdot \left(\log_{1-\alpha}\frac{c\alpha}{2n}\right)}{\alpha^2 c^2}\cdot \min \left\{d_t, \sqrt{\frac{m}{2(1-\alpha)}}\right\}
}
~=\tilde{O}\left(\min\left\{d_t,\sqrt{m}\right\}\right). 
\end{align*} 
\end{theorem}

\begin{proof} 
Recall that the variance of $\epi(t)$ obtained by Algorithm~\ref{alg:VBES} is bounded by $\frac{L\cdot \th \cdot d_t}{n}\cdot \epi(t)$ as shown in Theorem~\ref{thm:variance}. Plugging into the Chebyshev's Inequality, we can further derive: 
\begin{align*}
\Pr\left\{\epi(t)-\bpi(t)\ge \frac{c}{2} \cdot \vpi(t)\right\} \le \frac{4 \cdot \Var\left[\epi(t)\right]}{c^2 \cdot \left(\vpi(t)\right)^2}\le \frac{4L\th d_t}{c^2 \cdot n\vpi(t)}. 
\end{align*}
Thus, by setting $\th= \frac{c^2 \cdot p_f \cdot n\vpi(t)}{4Ld_t}$, $\epi(t)-\bpi(t)\le \frac{c}{2} \cdot \vpi(t)$ holds with probability at least $p_f$. In particular, we note $\frac{c^2 \cdot p_f \cdot n\vpi(t)}{4Ld_t}\ge \frac{\alpha c^2 \cdot p_f}{4Ld_t}$ based on the fact that $\vpi(t)\ge \frac{\alpha}{n}$ as illustrated in Equation~\eqref{eqn:ite_pagerank}. If we set $\theta=\frac{\alpha c^2\cdot p_f}{4Ld_t}$, then according to Lemma~\ref{lem:cost_theta}, the expected time cost of Algorithm~\ref{alg:VBES} can be bounded by 
{\rev 
$\frac{1}{\alpha \theta}=\frac{4L d_t}{\alpha^2 c^2 \cdot p_f}=\tilde{O}\left(d_t\right)$, where $\alpha, c, p_f$ are all constants, and $L=\log_{1-\alpha}{\frac{c\alpha}{2n}}$ (see Section~\ref{subsec:trunc_pagerank} for the details of setting $L$). 
}
Moreover, as we shall prove below, $\frac{n \vpi(t)}{d_t}\hspace{-0.5mm}\ge \hspace{-0.5mm}\alpha  \hspace{-0.5mm} \cdot  \hspace{-0.5mm}\sqrt{\frac{2(1-\alpha)}{m}}$ holds for any $t\in V$. Thus, by setting $\th \hspace{-0.5mm}= \hspace{-0.5mm}\frac{\alpha c^2 \cdot p_f}{4L}\hspace{-0.5mm}\cdot \hspace{-0.5mm}\sqrt{\frac{2(1-\alpha)}{m}}$, the expected time cost of Algorithm~\ref{alg:VBES} is bounded by 
{\rev 
$\frac{1}{\alpha \theta}=\frac{4L}{\alpha^2 c^2 \cdot p_f}\cdot \hspace{-1mm}\sqrt{\frac{m}{2(1-\alpha)}}=\tilde{O}(\sqrt{m})$. 
}

Now we present the proof of $\frac{n \vpi(t)}{d_t}\hspace{-0.5mm}\ge \hspace{-0.5mm}\alpha \cdot \sqrt{\frac{2(1-\alpha)}{m}}$. By Equation~\eqref{eqn:ite_pagerank}, we have: 
\begin{equation}\label{eqn:sqrtm}
\begin{aligned}
\vpi(t)\ge (1-\alpha)\hspace{-2mm}\sum_{u\in N(t)}\hspace{-1mm}\frac{\vpi(u)}{d_u}+\frac{\alpha}{n} \ge (1-\alpha)\hspace{-1mm}\sum_{u\in N(t)}\frac{1}{d_u}\cdot \frac{\alpha}{n}+\frac{\alpha}{n}. 
\end{aligned}    
\end{equation}
We note $\sum_{u\in N(t)}\hspace{-1mm}\frac{1}{d_u}\hspace{-0.5mm}\ge \hspace{-0.5mm}\frac{d_t^2}{2m}$ since
$\left(\sum_{u\in N(t)}\frac{1}{d_u}\right)\cdot \left(\sum_{u\in N(t)}d_u\right)\ge \left(\sum_{u\in N(t)}1\right)^2=d_t^2$ 
holds by the Cauchy-Schwarz Inequality~\cite{steele2004cauchy}. Plugging into Inequality~\eqref{eqn:sqrtm}, we can further derive: 
\begin{align*}
\vpi(t)\hspace{-0.5mm}\ge \hspace{-0.5mm}\frac{\alpha}{n}\hspace{-0.5mm}\cdot \hspace{-0.5mm}\left(\frac{(1-\alpha)d_t^2}{2m}\hspace{-0.5mm}+\hspace{-0.5mm}1\right)\hspace{-0.5mm}=\hspace{-0.5mm}\frac{\alpha d_t}{n}\hspace{-0.5mm}\cdot\hspace{-0.5mm} \left(\frac{(1-\alpha)d_t\hspace{-0.5mm}+\hspace{-0.5mm}\frac{2m}{d_t}}{2m}\right)\hspace{-0.5mm}\ge \hspace{-0.5mm}\frac{\alpha d_t}{n}\hspace{-0.5mm}\cdot \hspace{-0.5mm}\sqrt{\frac{2(1\hspace{-0.5mm}-\hspace{-0.5mm}\alpha)}{m}}, 
\end{align*}
where we apply the fact that $(1-\alpha)d_t+\frac{2m}{d_t} \ge 2\cdot \sqrt{(1-\alpha)2m}$ by the AM-GM Inequality. Consequently, $\frac{n \vpi(t)}{d_t}\hspace{-0.5mm}\ge \hspace{-0.5mm}\alpha \cdot \sqrt{\frac{2(1-\alpha)}{m}}$ holds for each $t\in V$, and the theorem follows. 
\end{proof}

\section{Experiments} \label{sec:exp}

\begin{table}[t]
\centering
\caption{Datasets}
\vspace{-4mm}
\begin{small}
\begin{tabular}{|@{\hspace{+4mm}}l@{\hspace{+4mm}}|@{\hspace{+4mm}}r@{\hspace{+4mm}}|@{\hspace{+4mm}}r@{\hspace{+4mm}}|@{\hspace{+4mm}}r@{\hspace{+4mm}}|} \hline
{{\bf Dataset}}& {{\bf $\boldsymbol{n}$}} & {{\bf $\boldsymbol{m}$}} & {{\bf  $\boldsymbol{m/n}$}} \\ \hline
{Youtube(YT)} &{1,138,499} & {5,980,886} & {5.25}\\
{IndoChina (IC)} &{7,414,768} & {301,969,638} & {40.73}\\
{Orkut-Links (OL)} & {3,072,441} & {234,369,798} & {76.28}  \\ 
{Friendster (FR)} & {68,349,466} & {3,623,698,684} & {53.02} \\
\hline
\end{tabular}
\end{small}
\label{tbl:datasets}
\vspace{-3mm}
\end{table}

This section presents the empirical results of \setpush. 
All experiments are conducted on a machine with an Intel(R) Xeon(R) Gold 6126@2.60GHz CPU and 500GB memory with the Linux OS. We implement all algorithms in C++ compiled by g++ with the O3 optimization turned on. 

\header{\bf Datasets. } We use four large-scale real-world datasets in the experiments~\footnote{http://snap.stanford.edu/data}~\footnote{http://law.di.unimi.it/datasets.php}, including Youtube (YT), IndoChina (IC), Orkut-Links (OL) and Friendster (FR). 
The Youtube, Orkut-Links and Friendster datasets are all originated from social networks, where the nodes in the graph correspond to the users in the website, and edges indicates friendship between users. Additionally, the IndoChina is a web dataset for the country domains in Indochina. We summarize the statistics of all the datasets in Table~\ref{tbl:datasets}.

{\rev 
\header{\bf Query Sets. } We generate two sets of query nodes, denoted as $Q_1$ and $Q_2$, in the experiments. First, for the $Q_1$ query set, we select $10$ nodes from the graph's vertex set $V$ uniformly at random. 
For the second query set $Q_2$, we select 10 query nodes from $V$ according to the node degree distribution. The larger the node's degree is, the more likely the node is selected into $Q_2$. 
Note that the PageRank distribution of a real-world network is experimentally observed to follow the power-law distribution~\cite{wei2018topppr, lofgren2016BiPPR,wei2019prsim, bahmani2010fast}. In particular, the power-law exponent of the PageRank distribution is the same as that of the degree distribution of the network. 
Therefore, by sampling query nodes according to the degree distribution, we are more likely to obtain the query nodes with relatively large PageRank scores. 
}


{\rev 
\header{\bf Parameters. } 
We compare our \setpush against five competitors: MC~\cite{fogaras2005MC}, LocalPush~\cite{lofgren2013personalized}, FastPPR~\cite{lofgren2014FastPPR}, RBS~\cite{wang2020RBS} and \sublinear~\cite{bressan2018sublinear}. Among them, MC is a Monte-Carlo method. LocalPush is a reverse exploration method. FastPPR~\cite{lofgren2014FastPPR}, RBS~\cite{wang2020RBS} and \sublinear~\cite{bressan2018sublinear} are all hybrid methods. 
We set the parameters of these competitors strictly according to the theoretical analysis. Specifically, for the MC method~\cite{fogaras2005MC}, it has one parameter $n_r$, the number of $\alpha$-random walks. We set $n_r=\frac{\frac{2}{3}c+2}{c^2 \cdot \vpi(t)}\cdot \ln{\frac{1}{p_f}}$ according to the analysis. The LocalPush method~\cite{lofgren2013personalized} has one parameter: the push threshold $\e$. We set $\e=\frac{c \alpha}{n}$. The FastPPR method has two parameters: the push threshold $r_{\max}$ and the number of random walks $n_r$. We set $r_{\max}=c\cdot \sqrt{\frac{\alpha d_t \cdot \log{(1/(1-\alpha))}}{n\cdot \log{(1/p_f)}\cdot \log{(n/\alpha)}}}$, and $n_r=\frac{45\log_{1-\alpha}{(c\alpha/2n)}}{c}\cdot \sqrt{\frac{n \cdot d_t \cdot \log{(1/(1-\alpha))}\log{(2/p_f)}}{\alpha\cdot \log{(n/\alpha)}}}$ according to the descriptions in FastPPR~\cite{lofgren2014FastPPR}. For RBS, recall that RBS can achieve the $\tilde{O}(n)$ time complexity by setting the threshold $\theta=\frac{\rela^2 \cdot \vpi(t)}{12\cdot \log_{1-\alpha}{(c \alpha /2n)}}$. However, we do not know the real value of $\vpi(t)$ in advance. Thus, the value of $\vpi(t)$ can be only in place of the lower bound $\frac{\alpha}{n}$ of $\vpi(t)$ as indicated in Equation~\eqref{eqn:ite_pagerank}. Thus, in the experiments of RBS, we set $\theta=\frac{c^2 \alpha}{12n \cdot \log_{1-\alpha}{(c \alpha /2n)}}$. For the \sublinear method, it has three parameters: the number of random walks $n_r$, the number of subgraphs $k$, and the maximum iteration number $L$. We set $n_r\hspace{-0.5mm}=\hspace{-0.5mm}\min\left\{n^{\frac{2}{3}}\hspace{-1mm}\cdot \dmax^{1/3} \hspace{-1mm}\cdot \hspace{-0.5mm}\left(\ln{\frac{n}{p_f}}\right)^{\frac{1}{3}}\hspace{-2mm}\cdot \hspace{-0.5mm}\left(\ln{\frac{1}{p_f}}\right)^{\frac{1}{3}}\hspace{-2mm}\cdot c^{-\frac{4}{3}},n^{\frac{4}{5}}\davg^{\frac{1}{5}}\hspace{-1mm} \cdot \hspace{-0.5mm}\left(\ln{\frac{n}{p_f}}\right)^{\frac{1}{5}}\hspace{-2mm}\cdot \hspace{-0.5mm}\left(\ln{\frac{1}{p_f}}\right)^{\frac{2}{5}}\hspace{-2mm}\cdot c^{-\frac{6}{5}}\right\}$, 
$k=\frac{n \cdot \log{1/p_f}}{c^2 \cdot n_r}$, and $L=\ln(c/n)$ following~\cite{bressan2018sublinear}. In all experiments, we set the failure probability $p_f=0.1$, the relative error parameter $c=0.1$, and the damping factor $\alpha=0.2$ unless otherwise specified.

\header{\bf Average Overall Query Time. } 
We first compare the empirical query time of all methods. Specifically, for each method, we issue one single-node PageRank query for each query node in the $Q_1$ query set, and report the average query time of each method over all the query nodes in $Q_1$ in Figure~\ref{fig:query_uniform_0.1} and Figure~\ref{fig:query_uniform_0.5}. In particular, we set the relative error parameter $c=0.1$ and $c=0.5$ in Figure~\ref{fig:query_uniform_0.1} and Figure~\ref{fig:query_uniform_0.5}, respectively. 
From Figure~\ref{fig:query_uniform_0.1} and Figure~\ref{fig:query_uniform_0.5}, we observe that our \setpush consistently outperforms other competitors, which demonstrates the superiority of our \setpush. 
It's worth mentioning that we omit the MC method on the FR dataset in Figure~\ref{fig:query_uniform_0.1} since the query time of MC on the FR dataset exceeds one day. 


Moreover, in Figure~\ref{fig:query_degree_0.1} and Figure~\ref{fig:query_degree_0.5}, we report the average query time of each method over all the query nodes in the query set $Q_2$. We omit the LocalPush method in both Figure~\ref{fig:query_degree_0.1} and Figure~\ref{fig:query_degree_0.5}, and the MC method in Figure~\ref{fig:query_degree_0.1} because the query time of these methods exceed one day. We note that our \setpush still consistently outperforms other competitors when $c=0.1$. When $c=0.5$, the empirical query time of our \setpush outperforms other competitors (except the \sublinear method) by up to an order of magnitude on all datasets. However, on the YT and OL datasets, the \sublinear method slightly outperforms our \setpush. 
We attribute the superiority of \sublinear as shown in Figure~\ref{fig:query_degree_0.5} to the blacklist trick adopted in the \sublinear method. 
In Figure~\ref{fig:dt}, we report the increment of the query time of each method with increasing $d_t$ and fixed $c=0.1$. We observe that our \setpush can consistently outperform \sublinear on all datasets. This demonstrates the superiority and robustness of our \setpush. 

\begin{figure}[t]
\centering
\begin{tabular}{cc}
\hspace{-4mm} \includegraphics[width=40mm]{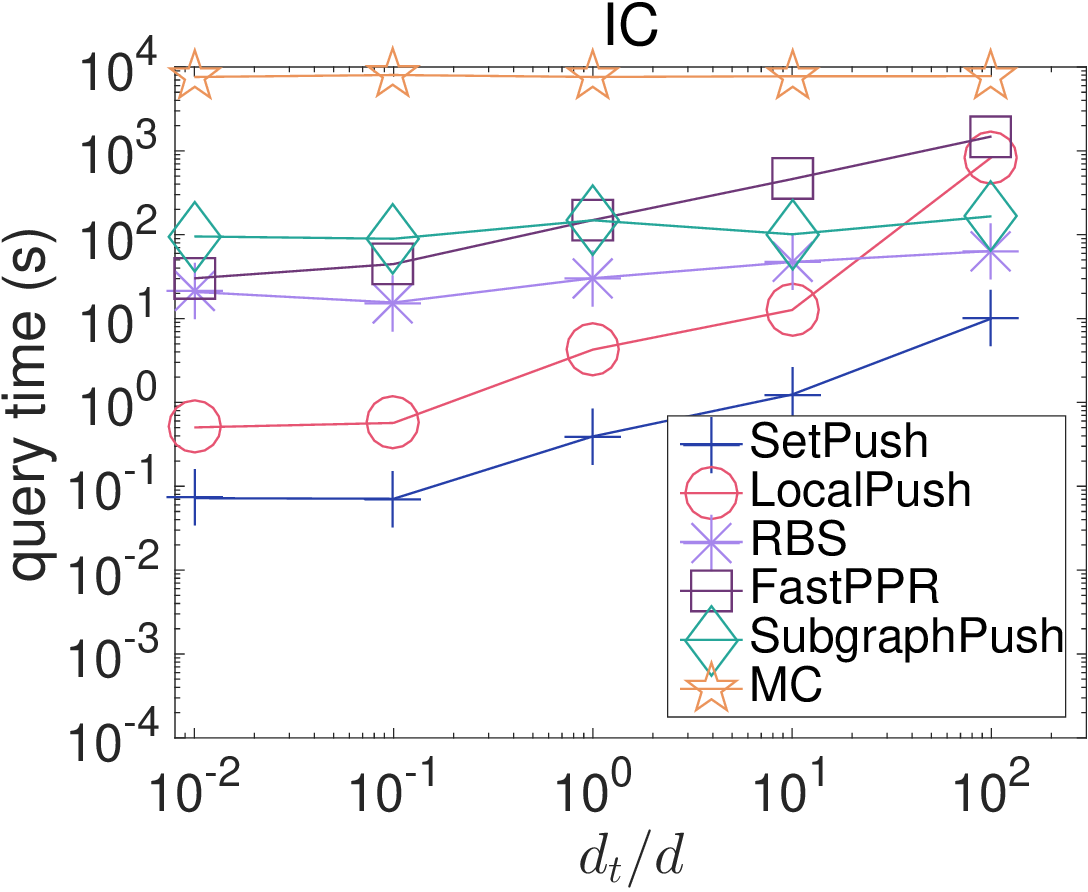} &
\hspace{-3mm} \includegraphics[width=40mm]{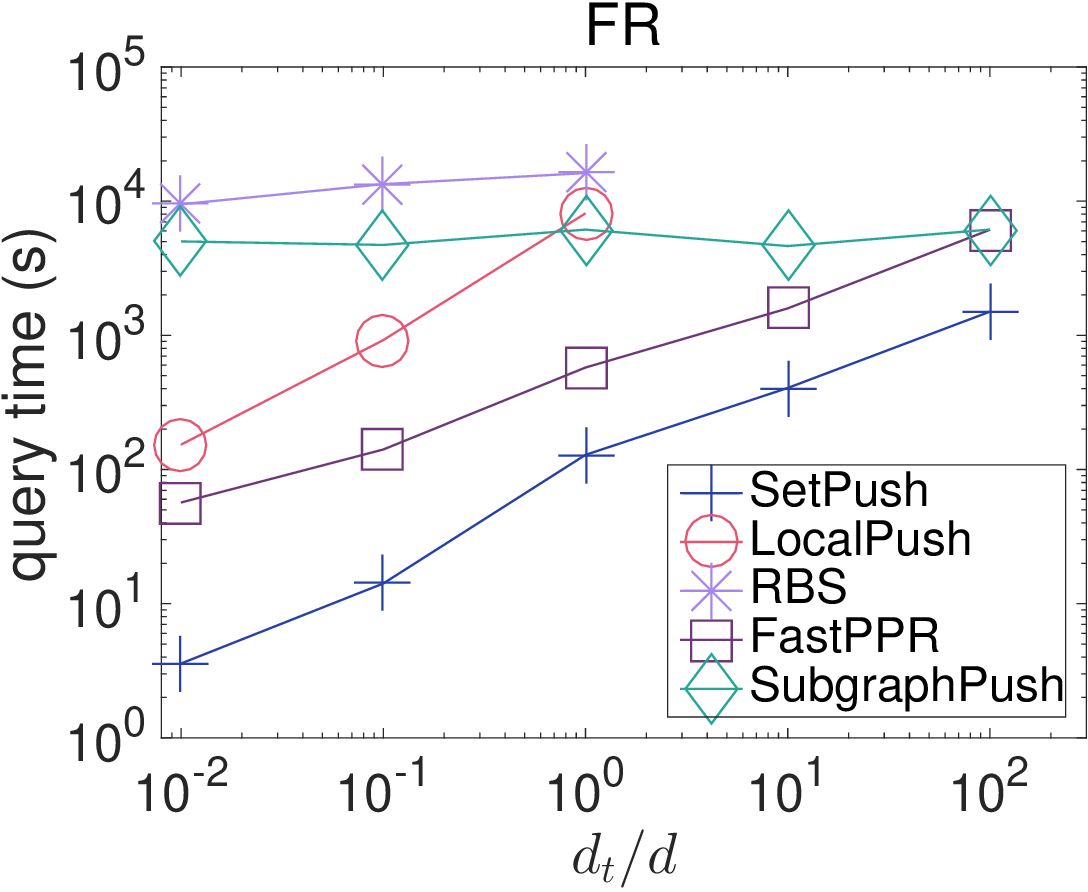} \\
\end{tabular}
\vspace{-6mm}
\caption{\rev $d_t$ v.s. query time. } 
\label{fig:dt}
\vspace{-6mm}
\end{figure}

\begin{figure}[t]
\centering
\begin{tabular}{cc}
\hspace{-3mm} \includegraphics[width=40mm]{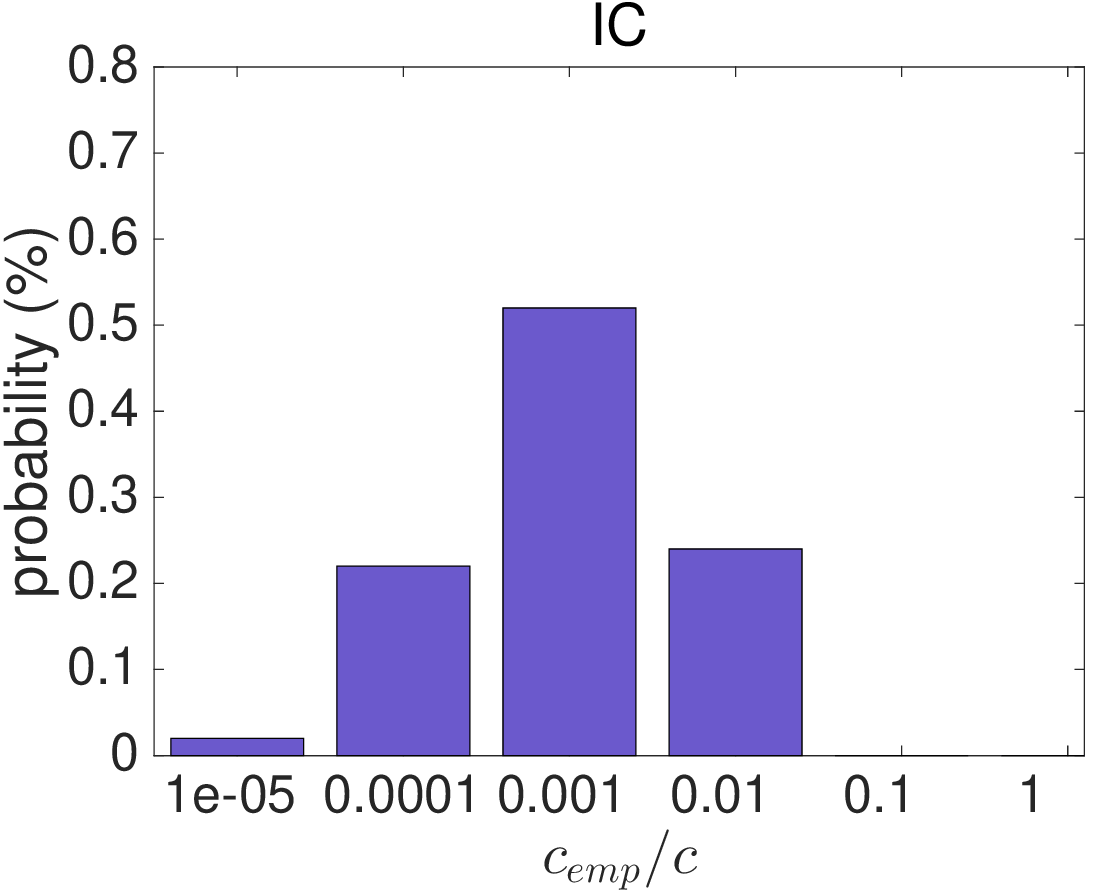} &
\hspace{-3mm} \includegraphics[width=40mm]{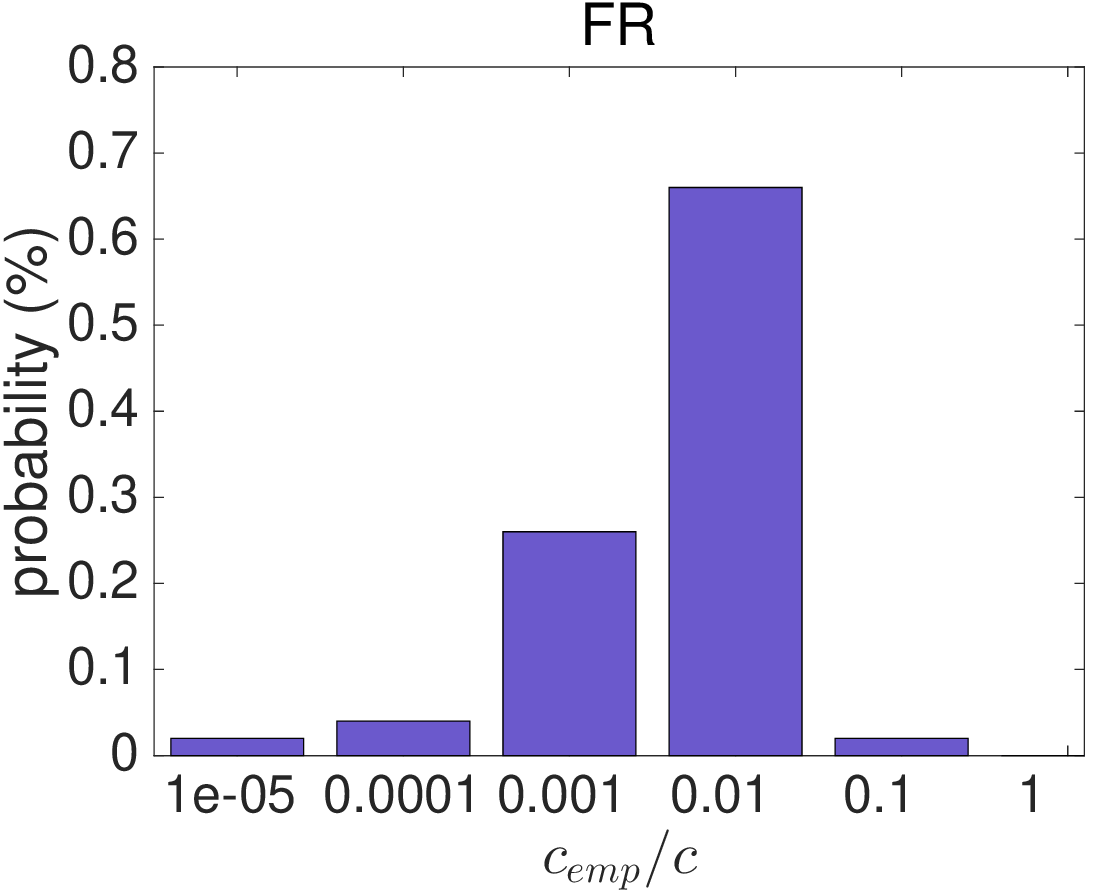} \\
\end{tabular}
\vspace{-6mm}
\caption{\rev Experiments on the empirical errors of \setpush. } 
\label{fig:epsr_exp}
\vspace{-6mm}
\end{figure}


\header{\bf $\boldsymbol{d_t}$ v.s. Average Overall Query Time. } 
In Figure~\ref{fig:dt}, we show the trade-off lines between $d_t$ (i.e., the degree of the target node $t$) and the empirical query time. We leverage such experiments to observe the relationship between the query time of each method and the value of $d_t$. Specifically, we partition the vertex set $V$ into five subsets $V_1, V_2, V_3, V_4, V_5$, such that the average node degrees $\davg_1, \davg_2, \davg_3, \davg_4, \davg_5$ of $V_1, V_2, V_3, V_4, V_5$ satisfy $\davg_1 \ge 100\davg$, $\davg_2 \in [10 \davg, 100 \davg)$, $\davg_3 \in [\davg, 10\davg)$, $\davg_4 \in [0.1 \davg, \davg)$, and $\davg_5 \in [0.01 \davg, 0.1\davg)$, respectively, where $d$ denotes the average node degree in the graph $G$. In each subset (i.e., $V_1, \ldots, V_5$), we select five query nodes uniformly at random, and report the average query time of each method over the five query nodes. 
We omit LocalPush and RBS on the FR dataset when $d_t/\davg\ge 10$ because their query time exceeds one day. We set $c=0.1$ and $p_f=0.1$ in these experiments. 
From Figure~\ref{fig:dt}, we note that our \setpush consistently outperforms all baseline methods on all datasets for all query sets. In particular, for law-degree query nodes, our \setpush achieves $10\times \sim 1000\times$ improvements on the query time over existing methods. For high-degree query nodes, the superiority of \setpush is gradually weakened, but still exists. Additionally, we observe: 
\begin{itemize}
    \item The query time of the Monte-Carlo method, RBS, and the \sublinear method nearly remain unchanged with the increment of $d_t$. This concurs with our analysis that the three methods do not include $d_t$ in their complexity results. 

    \item The query time of FastPPR and BiPPR increase slowly with the increment of $d_t$, while the query time of our \setpush and LocalPush grows linearly to $d_t$. This concurs with our analysis that the time complexities of FastPPR and BiPPR are both $\tilde{O}\left(\sqrt{n \cdot d_t}\right)$, while the time complexities of \setpush and LocalPush both have a linear dependence on $d_t$. 
\end{itemize}

\vspace{-2mm}
\header{\bf Empirical Errors of \setpush. } 
In Figure~\ref{fig:epsr_exp}, we evaluate the empirical error of our \setpush. Specifically, we adopt the power method~\cite{page1999pagerank} with the maximum iteration times $L=100$ to compute the ground truth of PageRank. 
Furthermore, on each dataset, we fix the relative error parameter $c=0.1$ and run our \setpush for each query node in the set $Q_1$. Then we compute the empirical relative error $c_{emp}$ for each query node following $c_{emp}=\frac{|\epi(t)-\vpi(t)|}{\vpi(t)}$. We report the average of the values $\left(\frac{c_{emp}}{c}\right)$ over all query nodes in Figure~\ref{fig:epsr_exp}. Note that $\left(\frac{c_{emp}}{c}\right) \le 1$ implies that the empirical relative error of \setpush meets the requirement of the $(c,p_f)$-approximation of $\vpi(t)$. From Figure~\ref{fig:epsr_exp}, we observe that the empirical relative errors of \setpush on all datasets are consistently smaller than $c$. In particular, on the IC datasets, the empirical relative errors of \setpush are smaller than $c$ by up to two orders of magnitude. This demonstrates the correctness and query efficiency of our \setpush. 

}

\vspace{-2mm}
\section{Conclusion} \label{sec:conclusion}
In this paper, we study the problem of single-node PageRank computation on undirected graphs. We propose a novel method, \setpush, which achieves the $\tilde{O}\left(\min\left\{d_t, \sqrt{m}\right\}\right)$ expected time complexity for estimating the target node $t$'s PageRank with constant relative error and constant success probability. {\rev We prove that this is the best result among existing methods on undirected graphs. We also empirically demonstrate the effectiveness of \setpush on large-scale real-world datasets. For the future work, we note that the lower bound for the problem of single-node PageRank computation on undirected graphs is still unclear. Since we have already achieved the complexity bound $\tilde{O}\left(\min\left\{d_t, \sqrt{m}\right\}\right)$, a natural question is whether this complexity matches the lower bound for the problem. 

}


\vspace{-2mm}
\begin{acks}
This research was supported in part by National Natural Science Foundation of China (No. U2241212, No. 61972401, No. 61932001, No. 61832017), by the major key project of PCL (PCL2021A12), by Beijing Natural Science Foundation (No. 4222028), by Beijing Outstanding Young Scientist Program No.BJJWZYJH012019100020098, by Alibaba Group through Alibaba Innovative Research Program, and by Huawei-Renmin University joint program on Information Retrieval. Hanzhi Wang was also supported by the Outstanding Innovative Talents Cultivation Funded Programs 2020 of Renmin University of China. We also wish to acknowledge the support provided by Engineering Research Center of Next-Generation Intelligent Search and Recommendation, Ministry of Education, and the fund for building world-class universities (disciplines) of Renmin University of China. Additionally, we acknowledge the support from Intelligent Social Governance Interdisciplinary Platform, Major Innovation \& Planning Interdisciplinary Platform for the “Double-First Class” Initiative, Public Policy and Decision-making Research Lab, Public Computing Cloud, Renmin University of China. 
\end{acks}


\bibliographystyle{ACM-Reference-Format}
\bibliography{paper}

\appendix
\section{Appendix} \label{sec:appendix}

\subsection{Proof of Lemma~\ref{lem:var_recur}}
Recall that we have proved 
\begin{align*}
\Var\left[\epi(t)\right]=\frac{\alpha^2}{n^2}\cdot \Var\left[\sum_{\ell=0}^L \sum_{s\in V}\frac{d_t}{d_s}\cdot \er^{(\ell)}_t(s)\right]
\end{align*}
in Section~\ref{sec:analysis}. In the following, we present the proof of: 
\begin{equation}\label{eqn:totalvar_final}
\begin{aligned}
&\Var\left[\sum_{\ell=0}^L \sum_{s\in V}\frac{d_t}{d_s}\cdot \er^{(\ell)}_t(s)\right]\\
&=\sum_{\ell=1}^{L-1}\E\left[\Var\left[\left.\sum_{v\in V}\left(\sum_{s\in V}\frac{d_t}{d_s}\cdot \sum_{i=0}^{L-\ell}\frac{\vpi^{(i)}_v(s)}{\alpha}\right)\cdot \er^{(\ell)}_t(v)~\right|~\er^{(\ell-1)}_t\right]\right].  
\end{aligned}    
\end{equation}
Specifically, by the law of total variance, we can derive 
\begin{equation}\label{eqn:total_var1}
\begin{aligned}
\Var\left[\sum_{\ell=0}^L \sum_{s\in V}\frac{d_t}{d_s}\cdot \er^{(\ell)}_t(s)\right]
&=\E\left[\Var\left[\sum_{\ell=0}^L \sum_{s\in V}\frac{d_t}{d_s}\cdot \er^{(\ell)}_t(s)~\Big|~ \er_t^{(L-1)}\right]\right]\\
&+\Var\left[\E\left[\sum_{\ell=0}^L \sum_{s\in V}\frac{d_t}{d_s}\cdot \er^{(\ell)}_t(s)~\Big|~ \er_t^{(L-1)}\right]\right]
\end{aligned}
\end{equation}
As we shall show in the following, the second term in the right hand side of Equation~\eqref{eqn:total_var1} can be iteratively rewritten as the sum of an expectation and, again, a variance expression. Thus, we can repeatedly adopt the law of total variance to further rewrite the new variance expression as the summation of an expectation and a variance. By repeating the above process, in the end, we will derive Equation~\eqref{eqn:totalvar_final}. Details are presented as below. 

As the first step, we note that the variance term given in the right hand side of Equation~\eqref{eqn:total_var1} can be rewritten as below by the linearity of expectation. 
\begin{equation}\label{eqn:total_var2}
\begin{aligned}
&\hspace{-0.5mm}\Var\hspace{-0.5mm}\left[\hspace{-0.5mm}\E\hspace{-0.5mm}\left[\hspace{-0.5mm}\sum_{\ell=0}^L \sum_{s\in V}\hspace{-0.5mm}\frac{d_t}{d_s}\hspace{-0.5mm}\cdot \hspace{-0.5mm}\er^{(\ell)}_t\hspace{-0.5mm}(s)\Big|~\er_t^{(L-1)}\hspace{-0.5mm}\right]\hspace{-0.5mm}\right]
\hspace{-1mm}=\hspace{-0.5mm}\Var\hspace{-0.5mm}\left[\hspace{-0.5mm}\sum_{s\in V}\hspace{-0.5mm}\frac{d_t}{d_s}\hspace{-0.5mm}\cdot \hspace{-1mm}\sum_{\ell=0}^L  \hspace{-0.5mm}\E\hspace{-0.5mm}\left[\hspace{-0.5mm}\er^{(\ell)}_t\hspace{-0.5mm}(s)\Big|~\er_t^{(L-1)}\hspace{-0.5mm}\right]\hspace{-0.5mm}\right]\\
&=\Var\left[\sum_{s\in V}\frac{d_t}{d_s}\cdot \left(\E\left[\er^{(L)}_t(s)\Big|~\er_t^{(L-1)}\right]+\sum_{\ell=0}^{L-1}  \E\left[\er^{(\ell)}_t(s)\Big|~\er_t^{(L-1)}\right]\right)\right]
\end{aligned}
\end{equation}
We note that for every $\ell\in [0,L-1]$, 
$\E\left[\er^{(\ell)}_t(s)~\Big|~\er_t^{(L-1)}\right]= \er^{(\ell)}_t(s)$.  
And for $\E\left[\er^{(L)}_t(s)\Big|~\er_t^{(L-1)}\right]$, we have: 
\begin{align*}
\E\left[\er^{(L)}_t(s)~\Big|~\er_t^{(L-1)}\right]=\sum_{u\in N(s)}\frac{(1-\alpha)}{d_u}\cdot \r^{(\ell)}_t(u)
\end{align*}
according to Equation~\eqref{eqn:conditional_exp}. 
In particular, we note that $\frac{(1-\alpha)}{d_u}=\frac{\vpi^{(1)}_u(s)}{\alpha}$ holds for every $u\in N(s)$ according to the definition formula of the $\ell$-hop PPR as shown in Equation~\eqref{eqn:def_lhopppr}. Thus, we further have: 
\begin{align*}
\E\left[\er^{(L)}_t(s)\Big|~\er_t^{(L-1)}\right]=\sum_{u\in N(s)}\frac{\vpi^{(1)}_u(s)}{\alpha}\cdot \er^{(L-1)}_t(u)
\end{align*}
Plugging into Equation~\eqref{eqn:total_var2}, we can therefore derive: 
\begin{equation}\label{eqn:total_var3}
\begin{aligned}
&\Var\left[\E\left[\sum_{\ell=0}^L \sum_{s\in V}\frac{d_t}{d_s}\cdot \er^{(\ell)}_t(s)~\Big|~\er_t^{(L-1)}\right]\right]\\
&=\Var\left[\sum_{s\in V}\frac{d_t}{d_s}\cdot \left(\sum_{u\in N(s)}\hspace{-1mm}\frac{1}{\alpha}\cdot \vpi^{(1)}_u(s) \cdot \er^{(L-1)}_t(u)+\sum_{\ell=0}^{L-1}\er^{(\ell)}_t(s)\right)\right]
\end{aligned}
\end{equation}
In particular, by Equation~\eqref{eqn:PPR_recur}, we have the following fact: 
\begin{itemize}
    \item $\vpi^{(0)}_s(s)=\frac{1}{\alpha}$, and $\vpi^{(0)}_s(u)=0$ for every $u\neq s$, 
    \item $\vpi^{(1)}_s(u)=0$ for every $u \notin N(s)$, 
\end{itemize}
Equation~\eqref{eqn:total_var3} can be further expressed as: 
\begin{equation}\label{eqn:total_var_mid}
\begin{aligned}
&\Var\left[\E\left[\sum_{\ell=0}^L \sum_{s\in V}\frac{d_t}{d_s}\cdot \er^{(\ell)}_t(s)\Big|~\er_t^{(L-1)}\right]\right]\\
&=\Var\left[\sum_{s\in V}\frac{d_t}{d_s}\cdot \left(\sum_{i=0}^{1}\sum_{u\in V}\hspace{-1mm}\frac{1}{\alpha}\hspace{-0.5mm}\cdot \hspace{-0.5mm}\vpi^{(i)}_u(s)\hspace{-0.5mm} \cdot \hspace{-0.5mm}\er^{(L-1)}_t(u)+\sum_{\ell=0}^{L-2}\er^{(\ell)}_t(s)\right)\right]. 
\end{aligned}
\end{equation}
Again, we apply the law of total variance (Fact~\ref{fact:totalvar}) to Equation~\eqref{eqn:total_var_mid}, which follows: 
\begin{equation}\label{eqn:total_var4}
\begin{aligned}
&\Var\left[\sum_{s\in V}\frac{d_t}{d_s}\cdot \left(\sum_{i=0}^{1}\sum_{u\in V}\hspace{-1mm}\frac{1}{\alpha}\hspace{-0.5mm}\cdot \hspace{-0.5mm}\vpi^{(i)}_u(s)\hspace{-0.5mm} \cdot \hspace{-0.5mm}\er^{(L-1)}_t(u)+\sum_{\ell=0}^{L-2}\er^{(\ell)}_t\hspace{-0.5mm}(s)\right)\right]\\
&=\E\left[\Var\left[\sum_{s\in V}\frac{d_t}{d_s}\hspace{-0.5mm}\cdot \hspace{-0.5mm}\left(\sum_{i=0}^{1}\sum_{u\in V}\hspace{-1mm}\frac{1}{\alpha}\hspace{-0.5mm}\cdot \hspace{-0.5mm}\vpi^{(i)}_u(s)\hspace{-0.5mm} \cdot \hspace{-0.5mm}\er^{(L-1)}_t\hspace{-0.5mm}(u)+\hspace{-1mm}\sum_{\ell=0}^{L-2}\hspace{-0.5mm}\er^{(\ell)}_t\hspace{-0.5mm}(s)\hspace{-0.5mm}\right)\Big|~ \er^{(L-2)}_t\hspace{-0.5mm}\right]\hspace{-0.5mm}\right]\\
&+\Var\left[\E\left[\sum_{s\in V}\frac{d_t}{d_s}\hspace{-0.5mm}\cdot \hspace{-0.5mm}\left(\sum_{i=0}^{1}\sum_{u\in V}\hspace{-1mm}\frac{1}{\alpha}\hspace{-0.5mm}\cdot \hspace{-0.5mm}\vpi^{(i)}_u(s)\hspace{-0.5mm} \cdot \hspace{-0.5mm}\er^{(L-1)}_t\hspace{-0.5mm}(u)+\hspace{-1mm}\sum_{\ell=0}^{L-2}\hspace{-0.5mm}\er^{(\ell)}_t\hspace{-0.5mm}(s)\hspace{-0.5mm}\right)\Big|~\er^{(L-2)}_t\hspace{-0.5mm}\right]\hspace{-0.5mm}\right]. 
\end{aligned}
\end{equation}
Repeating the above process, we can further rewrite the second term in the right hand side of Equation~\eqref{eqn:total_var4} as a summation of an expectation and a variance. Specifically, consider the right hand side of Equation~\eqref{eqn:total_var4}. By the linearity of expectation, we have: 
\begin{equation}\label{eqn:total_var5}
\begin{aligned}
&\hspace{-1mm}\Var\left[\E\left[\sum_{s\in V}\frac{d_t}{d_s}\hspace{-0.5mm}\cdot \hspace{-0.5mm}\left(\sum_{i=0}^{1}\sum_{u\in V}\hspace{-1mm}\frac{1}{\alpha}\hspace{-0.5mm}\cdot \hspace{-0.5mm}\vpi^{(i)}_u(s)\hspace{-0.5mm} \cdot \hspace{-0.5mm}\er^{(L-1)}_t\hspace{-0.5mm}(u)+\hspace{-1mm}\sum_{\ell=0}^{L-2}\hspace{-0.5mm}\er^{(\ell)}_t\hspace{-0.5mm}(s)\hspace{-0.5mm}\right)\Big|~\er^{(L-2)}_t\hspace{-0.5mm}\right]\hspace{-0.5mm}\right]\\
&\hspace{-2mm}=\hspace{-1mm}\Var\hspace{-0.5mm}\left[\hspace{-0.5mm}\sum_{s\in V}\hspace{-0.5mm}\frac{d_t}{d_s}\hspace{-0.5mm}\cdot \hspace{-0.5mm}\left(\sum_{i=0}^{1}\hspace{-0.5mm}\sum_{u\in V}\hspace{-1.5mm}\frac{\vpi^{(i)}_u\hspace{-0.5mm}(s)}{\alpha}\hspace{-0.5mm} \cdot \hspace{-0.5mm}\E\left[\hspace{-0.5mm}\er^{(L-1)}_t\hspace{-0.5mm}(u)\Big|\er^{(L-2)}_t\hspace{-0.5mm}\right]\hspace{-1mm}+\hspace{-1.5mm}\sum_{\ell=0}^{L-2}\hspace{-0.5mm}\E\hspace{-0.5mm}\left[\hspace{-0.5mm}\er^{(\ell)}_t\hspace{-0.5mm}(s)\Big|\er^{(L-2)}_t\hspace{-0.5mm}\right]\hspace{-0.5mm}\right)\hspace{-0.5mm}\right]
\end{aligned}
\end{equation}
Analogously, we have: 
\begin{align*}
\sum_{\ell=0}^{L-2}\E\left[\er^{(\ell)}_t(s)\Big|~\er^{(L-2)}_t\right]=\sum_{\ell=0}^{L-2}\er^{(\ell)}_t(s), 
\end{align*}
and by Equation~\eqref{eqn:conditional_exp} and Equation~\eqref{eqn:def_lhopppr}:
\begin{align*}
\E\left[\er^{(L-1)}_t\hspace{-0.5mm}(u)~\big|~\er^{(L-2)}_t\hspace{-0.5mm}\right]\hspace{-1mm}=\hspace{-3mm}\sum_{w\in N(u)}\hspace{-4mm}\frac{(1-\alpha)}{d_w}\cdot \er^{(L-2)}_t\hspace{-0.5mm}(w)\hspace{-0.5mm}=\hspace{-4mm}\sum_{w\in N(u)}\hspace{-4mm}\frac{\vpi^{(1)}_w\hspace{-0.5mm}(u)}{\alpha}\cdot \er^{(L-2)}_t\hspace{-0.5mm}(w). 
\end{align*}
Plugging into Equation~\eqref{eqn:total_var5}, we can further derive: 
\begin{equation}\label{eqn:total_var6}
\begin{aligned}
&\Var\left[\E\left[\sum_{s\in V}\frac{d_t}{d_s}\hspace{-0.5mm}\cdot \hspace{-0.5mm}\left(\sum_{i=0}^{1}\sum_{u\in V}\hspace{-1mm}\frac{1}{\alpha}\hspace{-0.5mm}\cdot \hspace{-0.5mm}\vpi^{(i)}_u(s)\hspace{-0.5mm} \cdot \hspace{-0.5mm}\er^{(L-1)}_t\hspace{-0.5mm}(u)+\hspace{-1mm}\sum_{\ell=0}^{L-2}\hspace{-0.5mm}\er^{(\ell)}_t\hspace{-0.5mm}(s)\hspace{-0.5mm}\right)\Big|~\er^{(L-2)}_t\hspace{-0.5mm}\right]\hspace{-0.5mm}\right]\\
&=\hspace{-0.5mm}\Var\hspace{-0.5mm}\left[\hspace{-0.5mm}\sum_{s\in V}\hspace{-0.5mm}\frac{d_t}{d_s}\hspace{-0.5mm}\cdot \hspace{-0.5mm}\left(\sum_{i=0}^{1}\sum_{u\in V}\hspace{-1mm}\sum_{w\in N(u)}\hspace{-3.5mm}\frac{1}{\alpha^2}\hspace{-0.5mm}\cdot \hspace{-0.5mm}\vpi^{(i)}_u(s)\hspace{-0.5mm}\cdot \hspace{-0.5mm}\vpi^{(1)}_w\hspace{-0.5mm}(u) \hspace{-0.5mm}\cdot \hspace{-0.5mm}\er^{(L-2)}_t\hspace{-0.5mm}(w)\hspace{-1mm}+\hspace{-1.5mm}\sum_{\ell=0}^{L-2}\hspace{-0.5mm}\er^{(\ell)}_t\hspace{-0.5mm}(s)\hspace{-0.5mm}\right)\hspace{-0.5mm}\right]\hspace{-1mm}. 
\end{aligned}
\end{equation}
Note that by Equation~\eqref{eqn:PPR_recur} and Equation~\eqref{eqn:birectional_ppr}, we can derive the following fact: 
\begin{equation}\label{eqn:recur_PPR_reverse_mid}
\begin{aligned}
&\vpi^{(i+1)}_w(s)=\frac{d_s}{d_w}\cdot \vpi^{(i+1)}_s(w)=\frac{d_s}{d_w}\cdot \hspace{-2mm}\sum_{u\in N(w)}\hspace{-1mm}\frac{(1-\alpha)}{d_u}\vpi^{(i)}_s(u)\\
&=\sum_{u\in N(w)}\hspace{-1mm}\frac{(1-\alpha)}{d_w}\cdot \left(\frac{d_s}{d_u}\cdot \vpi^{(i)}_s(u)\right)=\sum_{u\in N(w)}\hspace{-1mm}\frac{(1-\alpha)}{d_w}\cdot \vpi^{(i)}_u(s). 
\end{aligned}    
\end{equation}
Meanwhile, Equation~\eqref{eqn:PPR_recur} also indicates the following properties of $\ell$-hop PPR: 
\begin{itemize}
\item $\vpi^{(0)}_w(w)=\alpha$; 
\item $\vpi^{(1)}_w(u)=\frac{(1-\alpha)}{d_w}\cdot \vpi^{(0)}_w(w)=\frac{\alpha\cdot (1-\alpha)}{d_w}$ for every $u \in N(w)$; 
\item $\vpi^{(1)}_w(u)=0$ for every $u \notin N(w)$. 
\end{itemize}
Therefore, the recursive relation shown in Equation~\eqref{eqn:recur_PPR_reverse_mid} can be further expressed as: 
\begin{align}\label{eqn:recur_PPR_reverse}
\vpi^{(i+1)}_w(s)=\sum_{u\in N(w)}\frac{1}{\alpha}\cdot \vpi^{(1)}_w(u)\cdot \vpi^{(i)}_u(s). 
\end{align}
Plugging Equation~\eqref{eqn:recur_PPR_reverse} into Equation~\eqref{eqn:total_var6}, we can derive: 
\begin{equation*}
\begin{aligned}
&\Var\left[\E\left[\sum_{s\in V}\frac{d_t}{d_s}\hspace{-0.5mm}\cdot \hspace{-0.5mm}\left(\sum_{i=0}^{1}\sum_{u\in V}\hspace{-1mm}\frac{1}{\alpha}\hspace{-0.5mm}\cdot \hspace{-0.5mm}\vpi^{(i)}_u(s)\hspace{-0.5mm} \cdot \hspace{-0.5mm}\er^{(L-1)}_t\hspace{-0.5mm}(u)+\hspace{-1mm}\sum_{\ell=0}^{L-2}\hspace{-0.5mm}\er^{(\ell)}_t\hspace{-0.5mm}(s)\hspace{-0.5mm}\right)\Big|~\er^{(L-2)}_t\hspace{-0.5mm}\right]\hspace{-0.5mm}\right]\\
&=\hspace{-0.5mm}\Var\hspace{-0.5mm}\left[\sum_{s\in V}\hspace{-0.5mm}\frac{d_t}{d_s}\hspace{-0.5mm}\cdot \hspace{-0.5mm}\left(\sum_{i=0}^{1}\sum_{w\in V}\frac{1}{\alpha}\hspace{-0.5mm}\cdot \hspace{-0.5mm}\vpi^{(i+1)}_w(s)\hspace{-0.5mm}\cdot \hspace{-0.5mm}\er^{(L-2)}_t\hspace{-0.5mm}(w)\hspace{-0.5mm}+\hspace{-0.5mm}\sum_{\ell=0}^{L-2}\hspace{-0.5mm}\er^{(\ell)}_t(s)\right)\right]. 
\end{aligned}
\end{equation*}
Since $\vpi^{(0)}_s(s)=1$ and $\vpi^{(0)}_w(s)=0$ for any $w\neq s$ as mentioned above, we can further derive: 
\begin{equation}\label{eqn:total_var7}
\begin{aligned}
&\Var\left[\E\left[\sum_{s\in V}\frac{d_t}{d_s}\hspace{-0.5mm}\cdot \hspace{-0.5mm}\left(\sum_{i=0}^{1}\sum_{u\in V}\hspace{-1mm}\frac{1}{\alpha}\hspace{-0.5mm}\cdot \hspace{-0.5mm}\vpi^{(i)}_u(s)\hspace{-0.5mm} \cdot \hspace{-0.5mm}\er^{(L-1)}_t\hspace{-0.5mm}(u)+\hspace{-1mm}\sum_{\ell=0}^{L-2}\hspace{-0.5mm}\er^{(\ell)}_t\hspace{-0.5mm}(s)\hspace{-0.5mm}\right)\Big|~\er^{(L-2)}_t\hspace{-0.5mm}\right]\hspace{-0.5mm}\right]\\
&=\hspace{-0.5mm}\Var\hspace{-0.5mm}\left[\sum_{s\in V}\hspace{-0.5mm}\frac{d_t}{d_s}\hspace{-0.5mm}\cdot \hspace{-0.5mm}\left(\sum_{i=0}^{2}\sum_{w\in V}\frac{1}{\alpha}\hspace{-0.5mm}\cdot \hspace{-0.5mm}\vpi^{(i)}_w(s)\hspace{-0.5mm}\cdot \hspace{-0.5mm}\er^{(L-2)}_t\hspace{-0.5mm}(w)\hspace{-0.5mm}+\hspace{-0.5mm}\sum_{\ell=0}^{L-3}\hspace{-0.5mm}\er^{(\ell)}_t(s)\right)\right]. 
\end{aligned}
\end{equation}
If we apply the law of total variance to Equation~\eqref{eqn:total_var7} one more times, we will have the sum of an expectation and a variance again. 
Repeatedly applying the law of total variance and rewriting the expression of variance, as a consequence, we can derive: 
\begin{equation}\label{eqn:total_var8}
\begin{aligned}
&\Var\left[\sum_{\ell=0}^L \sum_{s\in V}\frac{d_t}{d_s}\cdot \er^{(\ell)}_t(s)\right]\\
&\hspace{-2mm}=\hspace{-1.5mm}\sum_{\ell=1}^{L-1}\hspace{-0.5mm}\E\hspace{-0.5mm}\left[\hspace{-0.5mm}\Var\hspace{-0.5mm}\left[\sum_{s\in V}\hspace{-0.5mm}\frac{d_t}{d_s}\hspace{-0.5mm}\cdot \hspace{-0.5mm}\left(\hspace{-0.5mm}\sum_{i=0}^{\ell}\hspace{-0.5mm}\sum_{w\in V}\hspace{-1mm}\frac{1}{\alpha}\hspace{-0.5mm}\cdot \hspace{-0.5mm}\vpi^{(i)}_w(s)\hspace{-0.5mm}\cdot \hspace{-0.5mm}\er^{(L-\ell)}_t\hspace{-0.5mm}(w)\hspace{-0.5mm}+\hspace{-3mm}\sum_{j=0}^{L-\ell-1}\hspace{-2mm}\er^{(j)}_t(s)\hspace{-0.5mm}\right)\Big| \er^{(L-\ell-1)}_t\hspace{-0.5mm}\right]\hspace{-0.5mm}\right]\\
&+\Var\left[\E\left[\sum_{s\in V}\frac{d_t}{d_s}\hspace{-0.5mm}\cdot \hspace{-0.5mm}\left(\sum_{i=0}^{L-1}\hspace{-0.5mm}\sum_{w\in V}\hspace{-0.5mm}\frac{1}{\alpha}\hspace{-0.5mm}\cdot \hspace{-0.5mm}\vpi^{(i)}_w(s)\hspace{-0.5mm}\cdot \hspace{-0.5mm}\er^{(1)}_t\hspace{-0.5mm}(w)\hspace{-0.5mm}+\er^{(0)}_t(s)\hspace{-0.5mm}\right)~\Big|~ \er^{(0)}_t\right]\right]
\end{aligned}
\end{equation}
For the second term in the right side of Equation~\eqref{eqn:total_var8}, we have: 
\begin{equation}\label{eqn:total_var9}
\begin{aligned}
&\Var\left[\E\left[\sum_{s\in V}\frac{d_t}{d_s}\hspace{-0.5mm}\cdot \hspace{-0.5mm}\left(\sum_{i=0}^{L-1}\hspace{-0.5mm}\sum_{w\in V}\hspace{-0.5mm}\frac{1}{\alpha}\hspace{-0.5mm}\cdot \hspace{-0.5mm}\vpi^{(i)}_w(s)\hspace{-0.5mm}\cdot \hspace{-0.5mm}\er^{(1)}_t\hspace{-0.5mm}(w)\hspace{-0.5mm}+\er^{(0)}_t(s)\hspace{-0.5mm}\right)~\Big|~ \er^{(0)}_t\right]\right]\\
&\hspace{-0.5mm}=\hspace{-0.5mm}\Var\hspace{-0.5mm}\left[\sum_{s\in V}\frac{d_t}{d_s}\hspace{-0.5mm}\cdot \hspace{-0.5mm}\left(\sum_{i=0}^{L-1}\hspace{-0.5mm}\sum_{w\in V}\hspace{-0.5mm}\frac{1}{\alpha}\hspace{-0.5mm}\cdot \hspace{-0.5mm}\vpi^{(i)}_w(s)\hspace{-0.5mm}\cdot \hspace{-0.5mm}\E\left[\er^{(1)}_t\hspace{-0.5mm}(w)\Big|~ \er^{(0)}_t\hspace{-0.5mm}\right]\hspace{-1mm}+\hspace{-0.5mm}\E\left[\hspace{-0.5mm}\er^{(0)}_t(s)\Big|~ \er^{(0)}_t\right]\hspace{-0.5mm}\right)\right]
\end{aligned}
\end{equation}
by the linearity of expectation. In particular, we note that by Equation~\eqref{eqn:conditional_exp}, we can derive: 
\begin{align*}
\E\left[\er^{(1)}_t\hspace{-0.5mm}(w)\Big|~ \er^{(0)}_t\right]=\sum_{x\in N(w)}\frac{(1-\alpha)}{d_x}\cdot \er^{(0)}_t(x). 
\end{align*}
Therefore, Equation~\eqref{eqn:total_var9} actually bounds the variance of $\er^{(0)}_t$. The randomness comes from the values of $\er^{(0)}_t$. However, according to Algorithm~\ref{alg:VBES}, $\er^{(0)}_t$ is deterministically set as $\er^{(0)}_t=\bm{e}_t$. As a consequence, we have: 
\begin{align*}
\Var\left[\E\left[\sum_{s\in V}\frac{d_t}{d_s}\hspace{-0.5mm}\cdot \hspace{-0.5mm}\left(\sum_{i=0}^{L-1}\hspace{-0.5mm}\sum_{w\in V}\hspace{-0.5mm}\frac{1}{\alpha}\hspace{-0.5mm}\cdot \hspace{-0.5mm}\vpi^{(i)}_w(s)\hspace{-0.5mm}\cdot \hspace{-0.5mm}\er^{(1)}_t\hspace{-0.5mm}(w)\hspace{-0.5mm}+\er^{(0)}_t(s)\hspace{-0.5mm}\right)~\Big|~ \er^{(0)}_t\right]\right]=0.  
\end{align*}
Plugging into Equation~\eqref{eqn:total_var8}, we can thus derive: 
\begin{equation*}
\begin{aligned}
&\Var\left[\sum_{\ell=0}^L \sum_{s\in V}\frac{d_t}{d_s}\cdot \r^{(\ell)}_t(s)\right]\\
&\hspace{-1mm}=\hspace{-1.5mm}\sum_{\ell=1}^{L-1}\hspace{-0.5mm}\E\hspace{-0.5mm}\left[\hspace{-0.5mm}\Var\hspace{-0.5mm}\left[\sum_{s\in V}\hspace{-0.5mm}\frac{d_t}{d_s}\hspace{-0.5mm}\cdot \hspace{-0.5mm}\left(\hspace{-0.5mm}\sum_{i=0}^{\ell}\hspace{-0.5mm}\sum_{w\in V}\hspace{-1mm}\frac{1}{\alpha}\hspace{-0.5mm}\cdot \hspace{-0.5mm}\vpi^{(i)}_w(s)\hspace{-0.5mm}\cdot \hspace{-0.5mm}\er^{(L-\ell)}_t\hspace{-0.5mm}(w)\hspace{-0.5mm}+\hspace{-3mm}\sum_{j=0}^{L-\ell-1}\hspace{-2mm}\er^{(j)}_t(s)\hspace{-0.5mm}\right)\Big| \er^{(L-\ell-1)}_t\hspace{-0.5mm}\right]\hspace{-0.5mm}\right]\\
&\hspace{-1mm}=\hspace{-1.5mm}\sum_{\ell=1}^{L-1}\hspace{-0.5mm}\E\hspace{-0.5mm}\left[\hspace{-0.5mm}\Var\hspace{-0.5mm}\left[\sum_{s\in V}\hspace{-0.5mm}\frac{d_t}{d_s}\hspace{-0.5mm}\cdot \hspace{-0.5mm}\left(\hspace{-0.5mm}\sum_{i=0}^{\ell}\hspace{-0.5mm}\sum_{w\in V}\hspace{-1mm}\frac{1}{\alpha}\hspace{-0.5mm}\cdot \hspace{-0.5mm}\vpi^{(i)}_w(s)\hspace{-0.5mm}\cdot \hspace{-0.5mm}\er^{(L-\ell)}_t\hspace{-0.5mm}(w)\hspace{-0.5mm}\right)\Big|~\er^{(L-\ell-1)}_t\hspace{-0.5mm}\right]\hspace{-0.5mm}\right],  
\end{aligned}
\end{equation*}
which follows the lemma.



\end{document}